\documentclass[11pt]{article}

\usepackage[utf8]{inputenc}
\usepackage[T1]{fontenc}
\usepackage{lmodern}
\usepackage[DIV=11]{typearea} 

\usepackage{microtype}

\usepackage{amssymb}
\usepackage{amsmath}
\usepackage{amsthm}
\usepackage{thmtools}

\usepackage[algo2e,procnumbered,ruled,linesnumbered,vlined]{algorithm2e}

\usepackage{tikz}

\usepackage{xcolor}
\usepackage{enumitem}
\usepackage{xspace}
\usepackage{subcaption}
\usepackage{floatpag}

\usepackage{comment}

\usepackage[
	style=alphabetic,
	backref=true,
	doi=false,
	url=false,
	maxcitenames=3,
	mincitenames=3,
	maxbibnames=10,
	minbibnames=10,
	backend=bibtex8,
	sortlocale=en_US
]{biblatex}

\usepackage[ocgcolorlinks]{hyperref} 
\usepackage{cleveref}

\newcommand*\samethanks[1][\value{footnote}]{\footnotemark[#1]}

\makeatletter
\g@addto@macro\bfseries{\boldmath}
\makeatother

\makeatletter
\g@addto@macro\mdseries{\unboldmath}
\g@addto@macro\normalfont{\unboldmath}
\g@addto@macro\rmfamily{\unboldmath}
\g@addto@macro\upshape{\unboldmath}
\makeatother

\DeclareCiteCommand{\citem}
    {}
    {\mkbibbrackets{\bibhyperref{\usebibmacro{postnote}}}}
    {\multicitedelim}
    {}

\renewcommand*{\multicitedelim}{\addcomma\space}

\AtEveryCitekey{\clearfield{note}}

\newcommand{\myhref}[1]{%
  \iffieldundef{doi}
    {\iffieldundef{url}
       {#1}
       {\href{\strfield{url}}{#1}}}
    {\href{http://dx.doi.org/\strfield{doi}}{#1}}%
}

\DeclareFieldFormat{title}{\myhref{\mkbibemph{#1}}}
\DeclareFieldFormat
  [article,inbook,incollection,inproceedings,patent,thesis,unpublished]
  {title}{\myhref{\mkbibquote{#1\isdot}}}

\makeatletter
\AtBeginDocument{%
    \newlength{\temp@x}%
    \newlength{\temp@y}%
    \newlength{\temp@w}%
    \newlength{\temp@h}%
    \def\my@coords#1#2#3#4{%
      \setlength{\temp@x}{#1}%
      \setlength{\temp@y}{#2}%
      \setlength{\temp@w}{#3}%
      \setlength{\temp@h}{#4}%
      \adjustlengths{}%
      \my@pdfliteral{\strip@pt\temp@x\space\strip@pt\temp@y\space\strip@pt\temp@w\space\strip@pt\temp@h\space re}}%
    \ifpdf
      \typeout{In PDF mode}%
      \def\my@pdfliteral#1{\pdfliteral page{#1}}
      \def\adjustlengths{}%
    \fi
    \ifxetex
      \def\my@pdfliteral #1{}
      \def\adjustlengths{\setlength{\temp@h}{-\temp@h}\addtolength{\temp@y}{1in}\addtolength{\temp@x}{-1in}}%
    \fi%
    \def\Hy@colorlink#1{%
      \begingroup
        \ifHy@ocgcolorlinks
          \def\Hy@ocgcolor{#1}%
          \my@pdfliteral{q}%
          \my@pdfliteral{7 Tr}
        \else
          \HyColor@UseColor#1%
        \fi
    }%
    \def\Hy@endcolorlink{%
      \ifHy@ocgcolorlinks%
        \my@pdfliteral{/OC/OCPrint BDC}%
        \my@coords{0pt}{0pt}{\pdfpagewidth}{\pdfpageheight}%
        \my@pdfliteral{F}
        \my@pdfliteral{EMC/OC/OCView BDC}%
        \begingroup%
          \expandafter\HyColor@UseColor\Hy@ocgcolor%
          \my@coords{0pt}{0pt}{\pdfpagewidth}{\pdfpageheight}%
          \my@pdfliteral{F}
        \endgroup%
        \my@pdfliteral{EMC}%
        \my@pdfliteral{0 Tr}
        \my@pdfliteral{Q}%
      \fi
      \endgroup
    }%
}
\makeatother

\usetikzlibrary{positioning,shapes.misc,calc}

\colorlet{DarkRed}{red!50!black}
\colorlet{DarkGreen}{green!50!black}
\colorlet{DarkBlue}{blue!50!black}

\hypersetup{
	linkcolor = DarkRed,
	citecolor = DarkGreen,
	urlcolor = DarkBlue,
	bookmarks = true,
	bookmarksnumbered = true,
	linktocpage = true
}

\declaretheorem[numberwithin=section]{theorem}
\declaretheorem[numberlike=theorem]{lemma}
\declaretheorem[numberlike=theorem]{proposition}
\declaretheorem[numberlike=theorem]{corollary}
\declaretheorem[numberlike=theorem]{definition}

\declaretheorem[numberlike=theorem]{fact}
\declaretheorem[numberlike=theorem]{observation}
\declaretheorem[numberlike=theorem,name={Algorithm},refname={algorithm,algorithms},Refname={Algorithm,Algorithms}]{algothm}

\DontPrintSemicolon
\SetProcNameSty{textsc}
\SetFuncSty{textsc}

\SetKw{KwAnd}{and}
\SetKw{KwBreak}{break}

\SetKwProg{Procedure}{Procedure}{}{}

\SetKwFunction{Insert}{Insert}
\SetKwFunction{Delete}{Delete}
\SetKwFunction{Query}{Query}
\SetKwFunction{Distance}{Distance}
\SetKwFunction{FindCenter}{FindCenter}
\SetKwFunction{Initialize}{Initialize}
\SetKwFunction{Increase}{Increase}
\SetKwFunction{UpdateLevels}{UpdateLevels}
\SetKwFunction{Open}{Open}
\SetKwFunction{Move}{Move}
\SetKwFunction{GreedyOpen}{GreedyOpen}
\SetKwFunction{CenterCover}{CenterCover}
\SetKwFunction{MovingCenter}{MovingCenter}
\SetKwFunction{CenterCoverGreedyOpen}{CenterCover.GreedyOpen}
\SetKwFunction{CenterCoverInitialize}{CenterCover.Initialize}
\SetKwFunction{CenterCoverFindCenter}{CenterCover.FindCenter}
\SetKwFunction{CenterCoverDistance}{CenterCover.Distance}
\SetKwFunction{CenterCoverDelete}{CenterCover.Delete}
\SetKwFunction{MCopen}{MovingCenter.Open}
\SetKwFunction{MCinitialize}{MovingCenter.Initialize}
\SetKwFunction{MCfindCenter}{MovingCenter.FindCenter}
\SetKwFunction{MCdistance}{MovingCenter.Distance}
\SetKwFunction{MCdelete}{MovingCenter.Delete}
\SetKwFunction{MCmove}{MovingCenter.Move}

\newcommand{\poly}{\operatorname{poly}}
\newcommand{\dist}{d}
\newcommand{\treeroot}{r}
\newcommand{\lev}{\ell}
\newcommand{\comp}{\operatorname{Comp}}
\newcommand{\degree}{\operatorname{deg}}
\newcommand{\q}{R^{\text{c}}}

\newcommand{\qp}{R^{\text{c}}_p}
\newcommand{\Q}{R^{\text{d}}}
\newcommand{\Qhat}{\hat{R}^{\text{d}}}

\newcommand{\Qp}{R^{\text{d}}_p}
\newcommand{\cen}{c}
\newcommand{\radius}{r}
\newcommand{\ball}{B}
\newcommand{\T}{C}
\newcommand{\C}{C}
\newcommand{\dom}{D}
\newcommand{\opens}{\ensuremath{N_\textrm{open}}}
\newcommand{\nummoves}{\ensuremath{N_\textrm{move}}}
\newcommand{\moves}{\ensuremath{D_\textrm{move}}}
\newcommand{\cG}{\mathcal{G}}
\newcommand{\cH}{\mathcal{H}}

\title{Dynamic Approximate All-Pairs Shortest Paths: Breaking the $O(mn)$ Barrier and Derandomization\thanks{The definite version of this article is published as: \fullcite{HenzingerKNSICOMP15}. A preliminary version was presented at the \emph{2013 IEEE 54th Annual Symposium on Foundations of Computer Science (FOCS 2013)}.}}

\author{
Monika Henzinger\thanks{University of Vienna, Faculty of Computer Science, Austria. Supported by the Austrian Science Fund (FWF): P23499-N23, the Vienna Science and Technology Fund (WWTF) grant ICT10-002, the University of Vienna (IK \mbox{I049-N}), and a Google Faculty Research Award. The research leading to these results has received funding from the European Research Council under the European Union's Seventh Framework Programme (FP/2007-2013) / ERC Grant Agreement no. 340506 and from the European Union's Seventh Framework Programme (FP7/2007-2013) under grant agreement no.~317532.}
\and Sebastian Krinninger\samethanks[2]
\and Danupon Nanongkai\thanks{University of Vienna, Faculty of Computer Science, Austria. Work partially done while at ICERM, Brown University, USA, and Nanyang Technological University, Singapore 637371, and while supported in part by the following research grants: Nanyang Technological University grant M58110000, Singapore Ministry of Education (MOE) Academic Research Fund (AcRF) Tier 2 grant MOE2010-T2-2-082, and Singapore MOE  AcRF Tier 1 grant MOE2012-T1-001-094.}
}

\date{}

\hypersetup{
	pdftitle = {Dynamic Approximate All-Pairs Shortest Paths: Breaking the O(mn) Barrier and Derandomization},
	pdfauthor = {Monika Henzinger, Sebastian Krinninger, Danupon Nanongkai}
}

\begin{document}
\maketitle
\begin{abstract}
We study dynamic $(1+\epsilon)$-approximation algorithms for the all-pairs shortest paths problem in unweighted undirected $n$-node $m$-edge graphs under edge deletions. The fastest algorithm for this problem is a randomized algorithm with a total update time of $\tilde O(mn/\epsilon)$ and constant query time by Roditty and Zwick \citem[FOCS 2004]{RodittyZ12}. The fastest deterministic algorithm is from a 1981 paper by Even and Shiloach \citem[JACM 1981]{EvenS81}; it has a total update time of $O(mn^2)$ and constant query time.
We improve these results as follows:

\begin{enumerate}[label=(\arabic{*})]
\item We present an algorithm with a total update time of $\tilde O(n^{5/2}/\epsilon)$ and constant query time that has an additive error of $ 2 $ in addition to the $ 1+\epsilon $ multiplicative error. 
This beats the previous $\tilde O(mn/\epsilon)$ time when $m=\Omega(n^{3/2})$.
Note that the additive error is {\em unavoidable} since, even in the {\em static} case, an $O(n^{3-\delta})$-time (a so-called {\em truly subcubic}) combinatorial algorithm with $ 1+\epsilon $ multiplicative error cannot have an additive error less than $2-\epsilon$, unless we make a major breakthrough for Boolean matrix multiplication \citem[Dor et al.\ FOCS 1996]{DorHZ00} and many other long-standing problems \citem[Vassilevska Williams and Williams FOCS 2010]{WilliamsW10}.

The algorithm can also be turned into a $(2+\epsilon)$-approximation algorithm (without an additive error) with the same time guarantees, improving the recent $(3+\epsilon)$-approximation algorithm with $\tilde O(n^{5/2+O(\sqrt{\log{(1/\epsilon)} / \log n})})$ running time of Bernstein and Roditty \citem[SODA 2011]{BernsteinR11} in terms of both approximation and time guarantees. 

\item We present a deterministic algorithm with a total update time of $\tilde O(mn/\epsilon)$ and a query time of $O(\log\log n)$. The algorithm has a multiplicative error of $ 1+\epsilon $ and gives the first improved deterministic algorithm since 1981. It also answers an open question raised by Bernstein \citem[STOC 2013]{Bernstein13}.
The deterministic algorithm can be turned into a deterministic fully dynamic $ (1+\epsilon) $-approximation with an amortized update time of $ \tilde O (m n / (\epsilon t)) $ and a query time of $ \tilde O (t) $ for every $ t \leq \sqrt{n} $.
\end{enumerate}  
In order to achieve our results, we introduce two new techniques: (1) A {\em monotone Even--Shiloach tree} algorithm which  maintains a bounded-distance shortest-paths tree on a certain type of emulator called {\em locally persevering emulator}.  (2) A derandomization technique based on {\em moving Even--Shiloach trees}
as a way to derandomize the standard random set argument.
These techniques might be of independent interest.

\end{abstract}
\newpage

\tableofcontents
\newpage

\section{Introduction}\label{sec:intro}

Dynamic graph algorithms is one of the classic areas in theoretical computer science with a countless number of applications. It concerns maintaining properties of dynamically changing graphs. The objective of a dynamic graph algorithm is to efficiently process an online sequence of update operations, such as edge insertions and deletions, and query operations on a certain graph property. It has to quickly maintain the graph property despite an {\em adversarial} order of edge deletions and insertions. 
Dynamic graph problems are usually classified according to the types of updates allowed: 
{\em decremental} problems allow only deletions, {\em incremental} problems allow only insertions, and {\em fully dynamic} problems allow both.

\subsection{The Problem} 
We consider the {\em decremental all-pairs shortest paths} (APSP) problem where we wish to maintain the distances in an undirected unweighted graph under a sequence of the following delete and distance query operations: 
\begin{itemize}
\item \Delete{$u$, $v$}: delete edge $(u, v)$ from the graph, and
\item \Distance{$x$, $y$}: return the distance between node $x$ and node $y$ in the current graph $G$, denoted by $\dist_{G}(x, y)$.
\end{itemize}
We use the term {\em single-source shortest paths} (SSSP) to refer to the special case where the distance query can be done only when $x=s$ for a prespecified {\em source node} $s$. The efficiency is judged by two parameters: {\em query time}, denoting the time needed to answer {\em each} distance query, and {\em total update time}, denoting the time needed to process {\em all} edge deletions. The running time will be in terms of $n$, the number of nodes in the graph, and $m$, the number of edges {\em before} any deletion. We use  $\tilde O$-notation to hide an $O(\poly\log n)$ term. When it is clear from the context, we use ``time'' instead of ``total update time,'' and, unless stated otherwise, the query time is $O(1)$.
One of the main focuses of this problem in the literature, which is also the goal in this paper, is to {\em optimize the total update time} while keeping the query time and {\em  approximation guarantees} small.
We say that an algorithm provides an {\em $(\alpha, \beta)$-approximation} if the distance query on nodes $x$ and $y$ on the current graph $G$ returns an estimate $ \delta (x, y) $ such that $\dist_{G}(x, y)\leq \delta (x, y) \leq \alpha\dist_G(x, y)+\beta$.
We call $\alpha$ and $\beta$ {\em multiplicative} and {\em additive errors}, respectively. 
We are particularly interested in the case where $\alpha= 1+\epsilon $, for an arbitrarily small constant $\epsilon>0$, $\beta$ is a small constant, and the query time is constant or near-constant.

\paragraph*{Previous Results.}
Prior to our work, the best total update time for {\em deterministic} decremental APSP algorithms was $\tilde O(mn^2)$ by one of the earliest papers in the area from 1981 by Even and Shiloach \cite{EvenS81}. The fastest {\em exact randomized} algorithms are the $\tilde O(n^3)$-time algorithms by  Demetrescu and Italiano \cite{DemetrescuI06} and Baswana, Hariharan, and Sen \cite{BaswanaHS07}. The fastest {\em approximation} algorithm is the $\tilde O(mn)$-time $(1+\epsilon, 0)$-approximation algorithm by Roditty and Zwick \cite{RodittyZ12}. If we insist on an $O(n^{3-\delta})$ running time, for some constant $\delta>0$, Bernstein and Roditty \cite{BernsteinR11} obtain an $\tilde O(n^{2+1/k+O(1/\sqrt{\log n})})$-time $(2k-1+\epsilon, 0)$-approximation algorithm, for any integer $k\geq 2$, which gives, e.g., a $(3+\epsilon, 0)$-approximation guarantee in $\tilde O(n^{5/2+O(1/\sqrt{\log n})})$ time. All these algorithms have an $O(1)$ worst-case query time. See \Cref{sec:related work} for more detail and other related results.

\subsection{Our Results}
We present improved randomized and deterministic algorithms. Our deterministic algorithm provides a $(1+\epsilon, 0)$-approximation and runs in $\tilde O(mn/\epsilon)$ total update time. Our randomized algorithm runs in $\tilde O(n^{5/2}/\epsilon)$ time and can guarantee both $(1+\epsilon, 2)$- and $(2+\epsilon, 0)$-approximations. 
\Cref{table:compare} compares our results with previous results. In short, we make the following improvements over previous algorithms (further discussions follow). 
\begin{itemize}
\item The total running time of deterministic algorithms is improved from Even and Shiloach's $\tilde O(mn^2)$ to $\tilde O(mn)$ (at the cost of $(1+\epsilon, 0)$-approximation and $O(\log \log n)$ query time). This is the first improvement since 1981.
\item For $m=\omega(n^{3/2})$, the total running time is improved from Roditty and Zwick's $\tilde O(mn/\epsilon)$ to $\tilde O(n^{5/2}/\epsilon)$, at the cost of an additive error of $2$, which appears only when the distance is $O(1/\epsilon)$ (since otherwise it could be treated as a multiplicative error of $O(\epsilon)$) and is {\em unavoidable} (as discussed below). 
\item Our $(2+\epsilon, 0)$-approximation algorithm improves the algorithm of Bernstein and Roditty in terms of both total update time and approximation guarantee. The multiplicative error of $ 2+\epsilon $ is essentially the best we can hope for, if we do not want any additive error.
\end{itemize}
To obtain these algorithms, we present two novel techniques, called {\em moving Even--Shiloach tree} and {\em monotone Even--Shiloach tree}, based on a classic technique of Even and Shiloach \cite{EvenS81}. These techniques are reviewed in \Cref{sec:techniques}.

\begin{table}
\centering
\begin{tabular}{|c|c|c|c|}
\hline
\centering{\bf\footnotesize Reference} & {\bf\footnotesize Total Running Time}& {\bf\footnotesize Approximation} & {\bf\footnotesize Deterministic?} \\
\hline
  \cite{EvenS81} & $\tilde O(mn^2)$ & Exact & Yes\\
  {\bf This paper} & $\tilde O(mn / \epsilon)$ & $(1+\epsilon, 0)$ & Yes\\
\hline
  \cite{DemetrescuI06,BaswanaHS07} & $\tilde O(n^3)$  & Exact & No \\
  \cite{RodittyZ12} & $\tilde O(mn / \epsilon)$  & $(1+\epsilon, 0)$ & No \\
  {\bf This paper} & $\tilde O(n^{5/2} / \epsilon)$ & $(1+\epsilon, 2)$ & No\\
\hline
  \cite{BernsteinR11} & $\tilde O(n^{5/2+\sqrt{\log (6/\epsilon)}/\sqrt{\log n}})$  & $(3+\epsilon, 0)$ & No \\
  {\bf This paper} & $\tilde O(n^{5/2} / \epsilon)$ & $(2+\epsilon, 0)$ & No\\
\hline
\end{tabular}

\caption{Comparisons between our and previous algorithms that are closely related. For details of these and other results see \Cref{sec:related work}. All algorithms, except our deterministic algorithm, have $O(1)$ query time. Our deterministic algorithm has $O(\log \log n)$ query time.} 
\label{table:compare}
\end{table}

\paragraph*{Improved Deterministic Algorithm.}
In 1981, Even and Shiloach \cite{EvenS81} presented a deterministic decremental SSSP algorithm for undirected, unweighted graphs with a total update time of $O(mn)$ over all deletions. 
By running this algorithm from $n$ different nodes, we get an $O(mn^2)$-time decremental algorithm  for APSP. No progress on deterministic decremental APSP has been made since then. 
Our algorithm achieves the first improvement over this algorithm, at the cost of a $(1+\epsilon, 0)$-approximation guarantee and $O(\log \log n)$ query time. (Note that our algorithm is also faster than the current fastest randomized algorithm \cite{RodittyZ12} by a $\log n$ factor.)
Our deterministic algorithm also answers a question recently raised by Bernstein \cite{Bernstein13} which asks for a deterministic algorithm with a total update time of $\tilde O(mn/\epsilon)$. As pointed out in \cite{Bernstein13} and several other places, this question is important due to the fact that deterministic algorithms can deal with an {\em adaptive offline adversary} (the strongest adversary model in online computation \cite{BorodonE98,Ben-DavidBKTW94}), while the randomized algorithms developed so far assume an {\em oblivious adversary} (the weakest adversary model) where the order of edge deletions must be fixed before an algorithm makes random choices. 
Our deterministic algorithm answers exactly this question. 
Using known reductions, we also obtain a deterministic fully dynamic $ (1+\epsilon) $-approximation with an amortized running time of $ \tilde O (m n / (\epsilon t)) $ per update and a query time of $ \tilde O (t) $ for every $ t \leq n $.

\paragraph*{Improved Randomized Algorithm.} 
Our aim is to improve the $\tilde O(mn)$ running time of Roditty and Zwick \cite{RodittyZ12} to so-called {\em truly subcubic time}, i.e., $O(n^{3-\delta})$ time for some constant $\delta>0$, a running time that is highly sought after in many problems (e.g., \cite{WilliamsW10,VassilevskaWY09,RodittyT13}).
Note, however, that this improvement has to come at the cost of worse approximation.

\begin{fact}[\cite{DorHZ00,WilliamsW10}]\label{fact:truly subcubic lower bound}
For any $\alpha\geq 1$ and $\beta\geq 0$ such that $2\alpha+\beta<4$, there is \textbf{no} combinatorial $(\alpha, \beta)$-approximation algorithm, not even a static one, for APSP on unweighted undirected graphs that is truly subcubic, unless we make a major breakthrough on many long-standing open problems, such as a combinatorial Boolean matrix multiplication and triangle detection.
\end{fact}

This fact is due to the reductions of Dor, Halperin, and Zwick \cite{DorHZ00} and Vassilevska Williams and Williams~\cite{WilliamsW10} (see \Cref{sec:proof of fact:truly subcubic lower bound} for a proof sketch). (Roditty and Zwick \cite{RodittyZ11} also showed a similar fact for decremental exact SSSP. For the weighted case, lower bounds can be obtained even for noncombinatorial algorithms by assuming the hardness of APSP computation~\cite{RodittyZ11,AbboudW14}.)
Very recently (after the preliminary version of this paper appeared), Henzinger~et~al.~\cite{HenzingerKNS15} showed that Fact~\ref{fact:truly subcubic lower bound} holds even for {\em noncombinatorial} algorithms assuming that there is no truly subcubic-time algorithm for a problem called {\em online Boolean matrix-vector multiplication}. Henzinger~et~al. \cite{HenzingerKNS15} argue that refuting this assumption will imply the same breakthrough as mentioned in Fact~\ref{fact:truly subcubic lower bound} if the term ``combinatorial algorithm'' (which is not a well-defined term) is interpreted in a certain way (in particular if it is interpreted as a ``Strassen-like algorithm'' as defined in \cite{BallardDHS12}, which captures all known fast matrix multiplication algorithms).
Thus, the best approximation guarantee we can expect from truly subcubic algorithms is, e.g., a multiplicative or additive error of at least $2$. Our algorithms achieve essentially these {\em best approximation guarantees}:
in $\tilde O(n^{5/2} / \epsilon)$ time, we get a $(1+\epsilon, 2)$-approximation, and, if we do not want any additive error, we can get a $(2+\epsilon, 0)$-approximation (see \Cref{thm:main_APSP_approximation} and \Cref{thm:2-approx} for the precise statements of these results).\footnote{We note that there is still some room to eliminate the $\epsilon$-term, i.e., to get a $(1, 2)$-approximation algorithm. But anything beyond this is unlikely to be possible.}
We note that, prior to our work, Bernstein and Roditty's algorithm \cite{BernsteinR11} could achieve, e.g., a $(3+\epsilon, 0)$-approximation guarantee in $\tilde O(n^{5/2+O(\sqrt{1/\log n})})$ time. This result is improved by our $(2+\epsilon, 0)$-approximation algorithm in terms of both time and approximation guarantees and is far worse than our $(1+\epsilon, 2)$-approximation guarantee, especially when the distance is large. 
Also note that the running time of our $(1+\epsilon, 2)$-approximation algorithm improves the $\tilde O(mn)$ one of Roditty and Zwick \cite{RodittyZ12} when $m=\omega(n^{3/2})$, except that our algorithm gives an additive error of $2$ which is unavoidable and appears only when the distance is $O(1/\epsilon)$ (since otherwise it could be counted as a multiplicative error of $O(\epsilon)$).

\subsection{Techniques}\label{sec:techniques}
Our results build on two previous algorithms. The first algorithm is the classic SSSP algorithm of Even and Shiloach \cite{EvenS81} (with the more general analysis of King~\cite{King99}), which we will refer to as the {\em Even--Shiloach tree}. 
The second algorithm is the $(1+\epsilon, 0)$-approximation APSP algorithm of Roditty and Zwick \cite{RodittyZ12}. We actually view the algorithm of Roditty and Zwick as a {\em framework} which runs several Even--Shiloach trees and maintains some properties while edges are deleted. We would like to alter the Roditty--Zwick framework but doing so usually makes it hard to bound the cost of maintaining Even--Shiloach trees (as we will discuss later). 
Our main technical contribution is the development of new variations of the Even--Shiloach tree, called {\em moving Even--Shiloach tree} and {\em monotone Even--Shiloach tree}, which are suitable for our modified Roditty--Zwick frameworks. Since there are many other algorithms that run Even--Shiloach trees as subroutines, it might be possible that other algorithms will benefit from our new Even--Shiloach trees as well.

\paragraph*{Review of Even--Shiloach Tree.}
The Even--Shiloach tree has two parameters: a root (or source) node $s$ and the range (or depth) $R$. It maintains a shortest paths tree rooted at $s$ and the distances between $s$ and all other nodes in the dynamic graph, up to distance $R$ (if the distance is more than $R$, it will be set to $\infty$). It has a query time of $O(1)$ and a total update time of $O(mR)$ over all deletions. 
The total update time crucially relies on the fact that the distance between $s$ and any node $v$ changes {\em monotonically}: it will increase at most $R$ times before it exceeds $R$ (i.e., from $1$ to $R$). 
This ``monotonicity'' property heavily relies on the ``decrementality'' of the model, i.e., the distance between two nodes never decreases when we delete edges, and is easily destroyed when we try to use the Even--Shiloach tree in a more general setting (e.g., when we want to allow edge insertions or alter the Roditty--Zwick framework). Most of our effort in constructing both randomized and deterministic algorithms will be spent on recovering from the destroyed decrementality.

\subsubsection{Monotone Even--Shiloach Tree for Improved Randomized Algorithms}

The high-level idea of our randomized algorithm is to run an existing decremental algorithm of Roditty and Zwick \cite{RodittyZ12} on a sparse {\em weighted} graph that approximates the distances in the original graph, usually referred to as an {\em emulator} (see \Cref{sec:emulator} for more detail). This approach is commonly used in the static setting (e.g., \cite{AingworthCIM99,DorHZ00,Elkin05,ElkinP04,AwerbuchBCP98,Cohen98,CohenZ01,ThorupZ05,Zwick02}), and it was recently used for the first time in the decremental setting by Bernstein and Roditty \cite{BernsteinR11}. As pointed out by Bernstein and Roditty, while it is a simple task to run an existing APSP algorithm on an emulator in the static setting, doing so in the decremental setting is not easy since it will destroy the ``decrementality'' of the setting:
when an edge in the original graph is deleted, we might have to {\em insert} an edge into the emulator. Thus, we cannot run decremental algorithms on an arbitrary emulator, because from the perspective of this emulator, we are not in a decremental setting.

Bernstein and Roditty manage to get around this problem by constructing an emulator with a special property.\footnote{In fact, their emulator is basically identical to one used earlier by Bernstein \cite{Bernstein09}, which is in turn a modification of a spanner developed by Thorup and Zwick~\cite{ThorupZ05, ThorupZ06}. However, the properties they proved are entirely new.} Roughly speaking, 
they show that their emulator guarantees that {\em the distance between any two nodes changes $\tilde O(n)$ times}.
Based on this simple property, they show that the $(2k-1, 0)$-approximation algorithm of Roditty and Zwick \cite{RodittyZ12} can be run on their emulator with a small running time. 
However, they {\em cannot} run the $(1+\epsilon, 0)$-approximation algorithm of Roditty and Zwick on their emulator. The main reason is that this algorithm relies on a more general property of a graph under deletions: for any $ R $ between $ 1 $ and $ n $, the distance between any two nodes changes at most $R$ times {\em before it exceeds $R$}  (i.e., it changes from $1$ to $R$). They suggested finding an emulator with this more general property as a future research direction. 

In our algorithm, we manage to run the $(1+\epsilon, 0)$-approximation algorithm of Roditty and Zwick on our emulator, but in a \emph {conceptually different} way from Bernstein and Roditty. 
In particular, we do not construct the emulator asked for by Bernstein and Roditty; rather, we show that there is a type of emulator such that, while edge insertions can occur often, their effect can be {\em ignored}. We then modify the algorithm of Roditty and Zwick to incorporate this ignoring of edge insertions. More precisely, the algorithm of Roditty and Zwick relies on the classic Even--Shiloach tree.  We develop a simple variant of this classic algorithm called the {\em monotone Even--Shiloach tree} that can handle restricted kinds of insertions and use it to replace the classic Even--Shiloach tree in the algorithm of Roditty and Zwick. 

Our modification to the Even--Shiloach tree is as follows. Recall that the Even--Shiloach tree can maintain the distances between a specific node $s$ and all other nodes, up to $R$, in $O(mR)$ total update time under edge deletions. This is because, for any node $v$, it has to do work $O(\degree(v))$ (the degree of $v$) only when the distance between $s$ and $v$ changes, which will happen at most $R$ times (from $1$ to $R$) in the decremental model. Thus, the total work on each node $v$ will be $O(R\degree(v))$ which sums to $O(mR)$ in total. This algorithm does not perform well when there are edge insertions: one edge insertion could cause a {\em decrease} in the distance between $s$ and $v$ by as much as $\Omega(R)$, causing an additional $\Omega(R)$ distance changes. 
The idea of our monotone Even--Shiloach tree is extremely simple: {\em ignore distance decreases}! It is easy to show that the total update time of our algorithm remains the same $O(mR)$ as the classic one. The hard part is proving that it gives a good approximation when run on an emulator. This is because it does not maintain the {\em exact} distances on an emulator anymore. So, even when the emulator gives a good approximate distance on the original graph, our monotone Even--Shiloach tree might not. Our monotone Even--Shiloach tree does not give any guarantee for the distances in the emulator, but we can show that it still approximates the distances in the original graph. 
Of course, this will {\em not} work on any emulator; but we can show that it works on a specific type of emulator that we call {\em locally persevering emulators}.\footnote{We remark that there are other emulators that can be maintained in the decremental setting; see, e.g., \cite{ThorupZ05,ThorupZ06,RodittyZ12,Bernstein09,BernsteinR11,AusielloFI06,Elkin11,BaswanaKS12}. We are the first to introduce the notion of locally persevering emulators and show that there is an emulator that has this property.}
Roughly speaking, a locally persevering emulator is an emulator where, for any ``nearby''\footnote{Note that the word ``nearby'' will be parameterized by a parameter $\tau$ in the formal definition. So, formally, we must use the term $(\alpha, \beta, \tau)$-locally persevering emulator where $\alpha$ and $\beta$ are multiplicative and additive approximation factors, respectively. See \Cref{sec:emulator} for detail.} nodes $u$ and $v$ in the original graph, either 
\begin{enumerate}[label=(\arabic{*})]
\item there is a shortest path from $ u $ to $ v $ in the original graph that also appears in the emulator, or 
\item \label{item:decremental path}there is a path in the emulator that approximates the distance in the original graph and {\em behaves in a persevering way}, in the sense that all edges of this path are in the emulator since before the first deletion and their weights never decrease.
We call the latter path a {\em persevering path}.
\end{enumerate}
Once we have the right definition of a locally persevering emulator, proving that our monotone Even--Shiloach tree gives a good distance estimate is conceptually simple (we sketch the proof idea below). 
Our last step is to show that such an emulator exists and can be efficiently maintained under edge deletions. We show (roughly) that we can maintain an emulator, which $(1+\epsilon, 2)$-approximates the distances and has $\tilde O(n^{3/2})$ edges, in $\tilde O(n^{5/2} / \epsilon)$ total update time under edge deletions. By running the $\tilde O(mn)$-time algorithm of Roditty and Zwick on this emulator, replacing the classic Even--Shiloach tree by our monotone version, we have the desired $\tilde O(n^{5/2} / \epsilon)$-time $(1+\epsilon, 2)$-approximation algorithm. To turn this algorithm into a $(2+\epsilon, 0)$-approximation, we  observe that we can check if two nodes are of distance $1$ easily; thus, we only have to use our $(1+\epsilon, 2)$-approximation algorithm to answer a distance query when the distance between two nodes is at least $2$. In this case, the additive error of $2$ can be treated as a multiplicative factor.

\paragraph*{Proving the Approximation Guarantee of the Monotone Even--Shiloach Tree.}
To illustrate why our monotone Even--Shiloach tree gives a good approximation when run on a locally persevering emulator, we sketch a result that is weaker and simpler than our main results; we show how to $(3, 0)$-approximate distances from a particular node $s$ to other nodes. This fact easily leads to a $(3+\epsilon, 0)$-approximation $\tilde O(n^{5/2} / \epsilon)$-time algorithm, which gives the same approximation guarantee as the algorithm of Bernstein and Roditty \cite{BernsteinR11} and is slightly faster and reasonably simpler. To achieve this, we use the following emulator which is a simple modification of the emulator of Dor, Halperin, and Zwick~\cite{DorHZ00}: Randomly select $ \tilde \Theta(\sqrt{n})$ nodes. At any time, the emulator consists of all edges incident to nodes of degree at most $\sqrt{n}$ and edges from each random node $c$ to every node $v$ of distance at most $2$ from $c$ with weight equal to the distance between $ v $ and $ c $. When the distance exceeds $ 2 $, the edge is deleted from the emulator. It can be shown that this emulator can be maintained in $\tilde O(mn^{1/2})=\tilde O(n^{5/2})$ time under edge deletions. Moreover, it is a $(3, 0)$-emulator with high probability, since for every edge $(u, v)$ either 
\begin{enumerate}[label=(\roman{*})]
\item $(u, v)$ is in the emulator, or 
\item\label{item:there is a path} there is a path $\langle u, c, v\rangle$ of length at most three, where $c$ is a random node.
\end{enumerate}
Observe further that if \ref{item:there is a path} happens, then the path $\langle u, c, v\rangle$ is {\em persevering} (as in \Cref{item:decremental path} above):
\begin{enumerate}[resume,label=(ii')]
\item $\langle u, c, v\rangle$ must be in this emulator since before the first deletion, and the weights of the edges $(u, c)$ and $(c, v)$ have never decreased. \label{item:two prime}
\end{enumerate}
It follows that this emulator is locally persevering.\footnote{We note that we are being vague here. To be formal, we later define the notion of $(\alpha, \beta, \tau)$-locally persevering emulator in \Cref{def:locally decremental emulator}, and the emulator we just defined will be $(3, 0, 1)$-locally persevering.} 
Now we show that when we run the monotone Even--Shiloach tree on the above emulator, it gives $(3, 0)$-approximate distances between $s$ and all other nodes. Recall that the monotone Even--Shiloach tree maintains a distance estimate, say $\ell(v)$, between $s$ and every node $v$ in the emulator.\footnote{Here $\ell$ stands for ``level'' as $\ell(v)$ is the level of $v$ in the breadth-first search tree rooted at $s$.} For every node $v$, the value of $\ell(v)$ is regularly updated, except that when the degree of a node drops to $ \sqrt{n} $ and the resulting insertion of an edge, say $(u, v)$, decreases the distance between $v$ and $s$ in the emulator; in particular, $\ell(v)>{\ell(u)+w(u, v)}$, where $w(u, v)$ is the weight of edge $(u, v)$.
A usual way to modify the Even--Shiloach tree for dealing with such an insertion~\cite{BernsteinR11} is to decrease the value of $\ell(v)$ to $\ell(u)+w(u, v)$.
Our monotone Even--Shiloach tree will \emph{not} do this and keeps $\ell(v)$ unchanged.
In this case, we say that the node $v$ and the edge $(u, v)$ {\em become stretched}. In general, an edge $(u, v)$ is {\em stretched} if $\ell(v)>\ell(u)+w(u, v)$ or $\ell(u)>\ell(v)+w(u, v)$, and a node is stretched if it is incident to a stretched edge. Two observations that we will use are
\begin{enumerate}[label=(O\arabic{*})]
\item \label{item:observation one} as long as a node $v$ is stretched, it will not change $\ell(v)$, and
\item \label{item:observation two} a stretched edge must be an inserted edge.
\end{enumerate}
We will argue that $\ell(v)$ of every node $v$ is at most three times its true distance to $s$ in the original graph. To prove this for a stretched node $v$, we simply use the fact that this is true before $v$ becomes stretched (by induction), and $\ell(v)$ has not changed since then (by \ref{item:observation one}). If $v$ is not stretched, we consider a shortest path $\langle v, u_1, u_2, \ldots, s\rangle$ from $ v $ to $ s $ in the original graph. We will prove that 
\[
\ell(v)\leq \ell(u_1)+3;
\] 
thus, assuming that $\ell(u_1)$ satisfies the claim (by induction), $\ell(v)$ will satisfy the claim as well. 
To prove this, observe that if the edge $(v, u_1)$ is contained in the emulator then we know that $\ell(v)\leq \ell(u_1)+1$ (since $v$ is not stretched), and we are done. Otherwise, by the fact that this emulator is locally persevering, we know that there is a path $\pi=\langle v, c, u_1\rangle$ of length at most $3$ in the emulator, and it is persevering (see \Cref{item:two prime}).
By \ref{item:observation two}, {\em edges in $\pi$ are not stretched}. It follows that 
$$\ell(v)\leq \ell(c)+w(v, c)\leq \ell(u_1)+w(v, c)+w(c, u_1)\leq  \ell(u_1)+3,$$ 
where $w(v, c)$ and $w(c, u_1)$ are the current weights of edges $(v, c)$ and $(c, u_1)$, respectively, in the emulator. The claim follows.

In \Cref{sec:faster}, we show how to refine the above argument to obtain a $(1+\epsilon, 2)$-approximation guarantee. The first refinement, which is simple, is extending the emulator above to a $(1+\epsilon, 2)$-emulator. This is done by adding edges from every random node $c$ to all nodes in distance at most $1/\epsilon$ from $c$. The next refinement, which is the main one, is the formal definition of $(\alpha, \beta, \tau)$-locally persevering emulators for some parameters $\alpha$, $\beta$, and $\tau$, and extending the proof outlined above to show that the monotone Even--Shiloach tree on such an emulator will give an $(\alpha+\beta/\tau, \beta)$-approximate distance estimate. We finally show that our simple $(1+\epsilon, 2)$-emulator is a $(1, 2, 1/\epsilon)$-locally persevering emulator.

\subsubsection{Moving Even--Shiloach Tree for Improved Deterministic Algorithms}

Many distance-related algorithms in both dynamic and static settings use the following {\em randomized argument} as an important technique:
if we select $\tilde O(h)$ nodes, called {\em centers}, uniformly at random, then every node will be at  distance at most $n/h$ from one of the centers with high probability~\cite{UllmanY91,RodittyZ12}. This even holds in the decremental setting (assuming an oblivious adversary). Like other algorithms, the Roditty--Zwick algorithm also heavily relies on this argument, which is the only reason it is randomized. Our goal is to derandomize this argument. 
Specifically, for several different values of $h$, the Roditty--Zwick framework selects $\tilde O(h)$ random centers and uses the randomized argument above to argue that every node in a connected component of size at least $n/h$ is {\em covered} by a center in the sense that it will always be within distance at most $n/h$ from at least one center; we call this set of centers a {\em center cover}. It also maintains an Even--Shiloach tree of depth $R=O(n/h)$ from these $h$ centers, which takes a total update time of $\tilde O(mR)$ for each tree and thus $\tilde O(hmR)=\tilde O(mn)$ over all trees. 
To derandomize the above process, we have two constraints: 

\begin{enumerate}[label=(\arabic{*})]
\item the center cover must be maintained (i.e., every node in a component of size at least $n/h$ has a center nearby), and
\item the number of centers (and thus Even--Shiloach trees maintained) must be $\tilde O(h)$ in total.
\end{enumerate}

Maintaining these constraints in the {\em static} setting is fairly simple, as in the following algorithm.

\begin{algothm}\label{algo:deterministic intuition 1}
As long as there is a node $v$ in a ``big'' connected component (i.e., of size at least $n/h$) that is not covered by any center,  make $v$ a new center.
\end{algothm}

\Cref{algo:deterministic intuition 1} clearly guarantees the first constraint.
The second constraint follows from the fact that the distance between any two centers is more than $n/h$. Since understanding the proof for guaranteeing the second constraint is important for understanding our {\em charging argument} later, we sketch it here. Let us label the centers by numbers $j=1, 2, \ldots, h$. For a center with number $j$, we let $\ball^j$ be a ``ball'' of radius $n/(2h)$; i.e., $\ball^j$ is a set of nodes in distance at most $n/(2h)$ from center number $j$.
Observe that $\ball^j$ and $\ball^{j'}$ are disjoint for distinct centers $j$ and $j'$ since the distance between these centers is more than $n/h$.
Moreover, $|\ball^j|\geq n/(2h)$ since every center is in a big connected component. So, the number of balls (thus the number of centers) is at most $n/(n/(2h))=2h$. This guarantees the second constraint. Thus, we can guarantee both constraints in the static setting. 

This, however, is not enough in the dynamic setting since {\em after} edge deletions, some nodes in big components might not be covered anymore, and if we keep repeating \Cref{algo:deterministic intuition 1}, we might have to keep creating new centers to such an extent that the second constraint is violated. 
The key idea that we introduce to avoid this problem is to allow a center and the Even--Shiloach tree rooted at it to {\em move}. We call this a {\em moving Even--Shiloach tree} or {\em moving centers} data structure. Specifically, in the moving Even--Shiloach tree, we view a root (center) $s$ {\em not} as a node, but as a {\em token} that can be placed on any node, and the task of the moving Even--Shiloach tree is to maintain the distance between the node on which the root is placed and all other nodes, up to distance $R$. We allow a {\em move operation} where we can move the root to a new node and the corresponding Even--Shiloach tree must be adjusted accordingly. To illustrate the power of the move operation, consider the following simple modification of \Cref{algo:deterministic intuition 1}. (Later, we also have to modify this algorithm due to other problems that we will discuss next.)
\begin{algothm}\label{algo:deterministic intuition 2}
As long as there is a node $v$ in a big connected component that is not covered by any center, we make it a center as follows. If there is a center in a small connected component, we move this center to $v$; otherwise, we open a new center at $v$.
\end{algothm}

\Cref{algo:deterministic intuition 2} reuses centers and Even--Shiloach trees in small connected components\footnote{We note the detail that we need a deterministic dynamic connectivity data structure~\cite{HenzingerK01,HolmLT01} to implement \Cref{algo:deterministic intuition 2}. The additional cost incurred is negligible.} without violating the first constraint since nodes in small connected components do not need to be covered. The second constraint can also be guaranteed by showing that $|\ball^j|\geq n/(2h)$ for all $j$ when we open a new center. Thus, by using moving Even--Shiloach trees, we can guarantee the two constraints above.
We are, however, {\em not done yet}. This is because our new move operation also incurs a cost! {\em The most nontrivial idea in our algorithm is a charging argument to bound this cost.} There are two types of cost. First, the {\em relocation cost}, which is the cost of constructing a new breadth-first search tree rooted at the new location of the center. This cost can be bounded by $O(m)$ since we can construct a breadth-first search tree by running the static $O(m)$-time algorithm. Thus, it will be enough to guarantee that we do not move Even--Shiloach trees more than $O(n)$ times. In fact, this is already guaranteed in \Cref{algo:deterministic intuition 2}
since we will {\em never} move an Even--Shiloach tree back to a previous node. 
The second cost, which is {\em much harder} to bound, is the {\em additional maintenance cost}. Recall that we can bound the total update time of an Even--Shiloach tree by $O(mR)$ because of the fact that the distance between its root (center) and each other node changes at most $R$ times before exceeding $R$, by increasing from $1$ to $R$. However, when we move the root from, say, a node $u$ to its neighbor $v$, the distance between the new root $v$ and some node, say $x$, might be smaller than the previous distance from $u$ to $x$. In other words, {\em the decrementality property is destroyed}. Fortunately, observe that the distance change will be {\em at most one} per node when we move a tree to a neighboring node. Using a standard argument, we can then conclude that {\em moving a tree between neighboring nodes costs an additional distance maintenance cost of $O(m)$}. This motivates us to define the notion of {\em moving distance} to measure how far we move the Even--Shiloach trees in total. We will be able to bound the maintenance cost by $O(mn)$ if we can show that the total moving distance (summing over all moving Even--Shiloach trees) is $O(n)$. Bounding the total moving distance by $O(n)$ while having only $O(h)$ Even--Shiloach trees is the most challenging part in obtaining our deterministic algorithm. We do it by using a careful charging argument. We sketch this argument here. For more intuition and detail, see \Cref{sec:deterministic}.

\paragraph*{Charging Argument for Bounding the Total Moving Distance.}
Recall that we denote the centers by numbers $j=1, 2, \ldots, h$. We make a few modifications to \Cref{algo:deterministic intuition 2}. The most important change is the introduction of the set $\T^j$ for each center $j$ (which is the root of a moving Even--Shiloach tree). This will lead to a few other changes. The importance of $\T^j$ is that we will ``charge'' the moving cost of center $j$ to nodes in $\T^j$; in particular, we bound the total moving distance to be $O(n)$ by showing that the moving distance of center $j$ can be bounded by $|\T^j|$, and $\T^j$ and $\T^{j'}$ are disjoint for distinct centers $j$ and $j'$. 
The other important changes are the definitions of ``ball'' and ``small connected component'' which will now depend on $\T^j$. 
\begin{itemize}
\item We change the definition of $\ball^j$ from a ball of radius $n/(2h)$ to a ball of radius $(n/(2h))-|\T^j|$. 
\item We redefine the notion of ``small connected component'' as follows: we say that a center $j$ is in a small connected component if the connected component containing it has less than $(n/(2h))-|\T^j|$ nodes (instead of $n/h$ nodes).
\end{itemize}
These new definitions might not be intuitive, but they are crucial for the charging argument. 
We also have to modify \Cref{algo:deterministic intuition 2} in a counterintuitive way: the most important modification is that we have to give up the nice property that the distance between any two centers is more than $n/(2h)$ as in \Cref{algo:deterministic intuition 1,algo:deterministic intuition 2}.
In fact, we will {\em always} move a center out of a small connected component, and we will move it {\em as little as possible}, even though the new location could be near other centers.  
In particular, consider the deletion of an edge $(u, v)$. It can be shown that there is {\em at most one} center $j$ that is in a small connected component (according to the new definition), and this center $j$ must be in the same connected component as $u$ or $v$. Suppose that such a center $j$ exists, and it is in the same connected component as $u$, say $X$. Then we will move center $j$ to $v$, which is just enough to move $j$ out of component $X$ (it is easy to see that $v$ is the node outside of $X$ that is nearest to $j$ before the deletion).  We will also update $\T^j$ by adding all nodes of $X$ to $\T^j$. This finishes the moving step, and it can be shown that there is no center in a small connected component now.
Next, we make sure that every node is covered by opening a new center at nodes that are not covered, as in \Cref{algo:deterministic intuition 1}.
To conclude, our algorithm is as follows.
\begin{algothm}\label{algo:deterministic intuition 3}
Consider the deletion of an edge $(u, v)$. Check whether there is a center $j$ that is in a ``small'' connected component $X$ (of size less than $(n/(2h))-|\T^j|$). If there is such a $j$ (there will be at most one such $j$), move it out of $X$ to a new node which is the unique node in $\{u, v\}\setminus X$. After moving, execute the static algorithm as in \Cref{algo:deterministic intuition 1}.
\end{algothm}

To see that the total moving distance is $O(n)$, observe that when we move a center~$j$ out of component $X$ in \Cref{algo:deterministic intuition 3}, we incur a moving distance of at most $|X|$ (since we can move $j$ along a path in $X$). Thus, we can always bound the total moving distance of center $j$ by $|\T^j|$. We additionally show that $\T^j$ and $\T^{j'}$ are disjoint for different centers $j$ and $j'$. So, the total moving distance over all centers is at most $\sum_j |\T^j| \leq n$. 
We also have to bound the number of centers. Since we give up the nice property that centers are far apart, we cannot use the same argument to show that the sets $\ball^j$ are disjoint and big (i.e., $|\ball^j|\geq n/(2h)$), as in \Cref{algo:deterministic intuition 2,algo:deterministic intuition 3}. However, using $\T^j$, we can still show something very similar: $\ball^j\cup \T^j$ and $\ball^{j'}\cup \T^{j'}$ are disjoint for distinct $j$ and $j'$, and $|\ball^j\cup \T^j|\geq n/(2h)$. Thus, we can still bound the number of centers by $O(h)$ as before.

\subsection{Related Work} \label{sec:related work}

Dynamic APSP has a long history, with the first papers dating back to 1967 \cite{Loubal67,Murchland67}\footnote{The early papers \cite{Loubal67,Murchland67}, however, were not able to beat the naive algorithm where we compute APSP from scratch after every change.}. It also has a tight connection with its {\em static} counterpart (where the graph does not change), which is one of the most fundamental problems in computer science: On the one hand, we wish to devise a dynamic algorithm that beats the naive algorithm where we recompute shortest paths {\em from scratch} using static algorithms after every deletion. On the other hand, the best we can hope for is to match the total update time of decremental algorithms to the best running time of static algorithms. 
To understand the whole picture, let us first recall the current situation in the static setting. We will focus on combinatorial algorithms\footnote{The vague term ``combinatorial algorithm'' is usually used to refer to algorithms that do not use algebraic operations such as matrix multiplication.} since our and most previous decremental algorithms are combinatorial. 
Static APSP on unweighted undirected graphs can be solved in $O(mn)$ time by simply constructing a breadth-first search tree from every node. Interestingly, this algorithm is the fastest combinatorial algorithm for APSP (despite other fast noncombinatorial algorithms based on matrix multiplication). In fact, a faster combinatorial algorithm will be a {\em major breakthrough}, not just because computing shortest paths is a long-standing problem by itself, but also because it will imply faster algorithms for other long-standing problems, as stated in \Cref{fact:truly subcubic lower bound}. 

The fact that the best static algorithm takes $O(mn)$ time means two things: First, the naive algorithm will take $O(m^2n)$ total update time. Second, the best total update time we can hope for is $O(mn)$.
A result that is perhaps the first to beat the naive $O(m^2n)$-time algorithm is from 1981 by Even and Shiloach \cite{EvenS81} for the special case of SSSP. Even and Shiloach actually studied decremental connectivity, but their main data structure gives an $O(mn)$ total update time with $O(1)$ query time for decremental SSSP; this implies a total update time of $O(mn^2)$ for decremental APSP.
Roditty and Zwick \cite{RodittyZ11} later provided evidence that the $O(mn)$-time decremental unweighted SSSP algorithm of Even and Shiloach is the fastest possible by showing that this problem is at least as hard as several natural static problems such as Boolean matrix multiplication and the problem of finding all edges of a graph that are contained in triangles.
For the incremental setting, Ausiello et al.~\cite{AusielloIMN91} presented an $\tilde O(n^3)$-time APSP algorithm on unweighted directed graphs. (An extension of this algorithm for graphs with small integer edge weights is given in \cite{AusielloIMN92}.) After that, many efficient fully dynamic algorithms have been proposed (e.g., \cite{HenzingerKRS97,King99,FakcharoenpholR06,DemetrescuI06,DemetrescuI02}).
Subsequently, Demetrescu and Italiano \cite{DemetrescuI04} achieved a major breakthrough for the fully dynamic case: they obtained a fully dynamic deterministic algorithm for the weighted directed APSP problem with an amortized time of $\tilde O(n^2)$ \emph{per update}, implying a total update time of $\tilde O(mn^2)$ over all deletions in the decremental setting, the same running time as the algorithm of Even and Shiloach. (Thorup \cite{Thorup04} presented an improvement of this result.)  An amortized update time of $\tilde O(n^2)$ is essentially optimal if the distance matrix is to be explicitly maintained, as done by the algorithm of Demetrescu and Italiano \cite{DemetrescuI04}, since each update operation may change $\Omega(n^2)$ entries in the matrix.
Even for unweighted, undirected graphs, no faster algorithm is known.
Thus, the $O(mn^2)$ total update time of Even and Shiloach {\em remains the best} for deterministic decremental algorithms, even on undirected unweighted graphs and if approximation is allowed. 

For the case of randomized algorithms, Demetrescu and Italiano~\cite{DemetrescuI06} obtained an exact decremental algorithm on weighted directed graphs with $\tilde O(n^3)$ total update time\footnote{This algorithm actually works in a much more general setting where each edge weight can assume $ S $ different values. Note that the amortized time per update of this algorithm is $\tilde O(S n)$, but this holds only when there are $\Omega(n^2)$ updates (see \cite[Theorem 10]{DemetrescuI06}). Also note that the algorithm is randomized with one-sided error.} (if weight increments are not considered).
Baswana, Hariharan, and Sen \cite{BaswanaHS07} obtained an exact decremental algorithm on unweighted directed graphs with $\tilde O(n^3)$ total update time. They also obtained a  $(1+\epsilon, 0)$-approximation algorithm with $\tilde O(m^{1/2}n^2)$ total update time. In \cite{BaswanaHS03}, they improved the running time further on undirected unweighted graphs, at the cost of a worse approximation guarantee: they obtained approximation guarantees of $(3, 0)$, $(5, 0)$, $(7, 0)$ in $\tilde O(mn^{10/9})$, $\tilde O(mn^{14/13})$, and $\tilde O(mn^{28/27})$ time, respectively. 
Roditty and Zwick \cite{RodittyZ12} presented two improved algorithms for unweighted, undirected graphs. The first was a $(1+\epsilon, 0)$-approximate decremental APSP algorithm with constant query time and a total update time of $\tilde O(mn)$. This algorithm remains the current fastest. The second algorithm achieves a worse approximation bound of $(2k-1, 0)$ for any $2\leq k\leq \log n$, but has the advantage of requiring less space ($ O(m+n^{1+1/k})$). 
By modifying the second algorithm to work on an emulator, Bernstein and Roditty \cite{BernsteinR11} presented the first truly subcubic algorithm which gives a $(2k-1+\epsilon, 0)$-approximation and has a total update time of $\tilde O (n^{1+1/k+O(1/\sqrt{\log n})})$. They also presented a $(1+\epsilon, 0)$-approximation $\tilde O(n^{2+O(1/\sqrt{\log n}}))$-time algorithm for SSSP, which is the first improvement since the algorithm of Even and Shiloach. Very recently, Bernstein \cite{Bernstein13} presented a $(1+\epsilon, 0)$-approximation $\tilde O(mn\log W)$-time algorithm for the directed weighted case, where $W$ is the ratio of the largest edge weight ever seen in the graph to the smallest such weight.

We note that the $(1+\epsilon, 0)$-approximation $\tilde O(mn)$-time algorithm of Roditty and Zwick matches the state of the art in the static setting; thus, it is essentially tight. However, by allowing additive error, this running time was improved in the static setting. For example, Dor, Halperin, and Zwick \cite{DorHZ00}, extending the approach of Aingworth et al.~\cite{AingworthCIM99}, presented a $(1,2)$-approximation for APSP in unweighted undirected graphs with a running time of $O(\min\{n^{3/2}m^{1/2}, n^{7/3}\})$. Elkin \cite{Elkin05} presented an algorithm for unweighted undirected graphs with a running time of $O(mn^\rho+n^2\zeta)$ that approximates the distances with a multiplicative error of $ 1+\epsilon $ and an additive error that is a function of $\zeta$, $\rho$, and $\epsilon$. There is no decremental algorithm with additive error prior to our algorithm.

\paragraph*{Subsequent Work.}
Independent of our work, Abraham and Chechik~\cite{AbrahamC13}  developed a randomized $ (1+\epsilon, 2) $-approximate decremental APSP algorithm with a total update time of $ \tilde O (n^{5/2+O(1/\sqrt{\log{n}})}) $ and constant query time. This result is very similar to one of ours, except that the running time in \cite{AbrahamC13} is slightly more than $ \tilde O (n^{5/2}) $. 
After the preliminary version of this paper~\cite{HenzingerKNFOCS13} appeared, we extended the randomized algorithm in this paper and obtained the following two algorithms for APSP~\cite{HenzingerKNSODA14}: (i) a $(1+\epsilon, 2(1+2 / \epsilon)^{k-2})$-approximation with total time $\tilde O( n^{2+1/k}(37/\epsilon)^{k-1})$ for any $2\leq k\leq \log n$ (improving the time in this paper with a higher additive error when $k\geq 3$), and (ii) a $(3+\epsilon)$-approximation with total time $\tilde O(m^{2/3}n^{3.8/3+O(1/\sqrt{\log n})})$ (it is faster than the algorithm in this paper for sparse graphs but causes more multiplicative error). These two algorithms heavily rely on the monotone Even--Shiloach tree introduced in this paper. In the same paper, the monotone Even--Shiloach tree was also used in combination with techniques in \cite{HenzingerKNICALP13} to obtain the first subquadratic-time algorithm for approximate SSSP. Very recently, we obtained an almost linear total update time for $ (1+\epsilon) $-approximate SSSP in weighted undirected graphs~\cite{HenzingerKNFOCS14}, where the monotone Even--Shiloach tree again played a central role. 
We also obtained the first improvement over Even--Shiloach's algorithm for single-source reachability and approximate single-source shortest paths on \emph{directed} graphs \cite{HenzingerKNSTOC14}.

\section{Background}\label{sec:background}

\subsection{Basic Definitions}

In the following we give some basic notation and definitions.

\begin{definition}[Dynamic graph]
A \emph{dynamic graph} $ \cG $ is a sequence of graphs $ \cG = (G_i)_{0 \leq i \leq k} $ that share a common set of nodes $ V $.
The set of edges of the graph $ G_i $ (for $ 0 \leq i \leq k $) is denoted by $ E (G_i) $.
The number of nodes of $ \cG $ is $ n = | V | $, and the initial number of edges of $ \cG $ is $ m = | E (G_0) | $.
The set of edges ever contained in $ \cG $ up to time $ t $ (where $ 0 \leq t \leq k $) is $ E_t (\cG) = \cup_{0 \leq i \leq t} E (G_i) $.
A dynamic weighted graph $ \cH $ is a sequence of weighted graphs $ \cH = (H_i, w_i)_{0 \leq i \leq k} $ that share a common set of nodes $ V $.
For $ 0 \leq i \leq k $ and every edge $ (u, v) \in E (H_i) $, the weight of $ (u, v) $ is given by $ w_i (u, v) $.
\end{definition}

Let us clarify how a dynamic graph $ \cG = (G_i)_{0 \leq i \leq k} $ is processed by a dynamic algorithm.
The dynamic graph $ \cG $ is a sequence of graphs picked by an adversary before the algorithm starts.
In its initialization phase, the algorithm may process the initial graph $ G_0 $, and in the $i$-th update phase the algorithm may process the graph $ G_i $.
At the beginning of the $i$-th update phase, the graph $ G_i $ is presented to the algorithm implicitly as the set of updates from $ G_{i-1} $ to $ G_i $.
The algorithm will, for example, be informed which edges were deleted from the graph.
After the initialization phase and after each update phase, the algorithm has to be able to answer queries.
In our case, these queries will usually be distance queries, and the algorithm will answer them in constant or near-constant time.
The \emph{total update time} of the algorithm is the total time spent processing the initialization and \emph{all} $ k $ updates.

\begin{definition}[Updates]
For a dynamic graph $ \cG = (G_i)_{0 \leq i \leq k} $ we say for an edge $ (u, v) $ that
\begin{itemize}
\item $ (u, v) $ is \emph{deleted} at time $ t $ if $ (u, v) $ is contained in $ G_{t-1} $ but not in $ G_t $;
\item $ (u, v) $ is \emph{inserted} at time $ t $ if $ (u, v) $ contained in $ G_t $ but not in $ G_{t-1} $.
\end{itemize}
For a dynamic weighted graph $ \cH = (H_i, w_i)_{0 \leq i \leq k} $, we additionally say for an edge $ (u, v) $ that
\begin{itemize}
\item the weight of $ (u, v) $ is \emph{increased} at time $ t $ if $ w_t (u, v) > w_{t-1} (u, v) $ (and $ (u, v) $ is contained in both $ G_{t-1} $ and $ G_t $);
\item the weight of $ (u, v) $ is \emph{decreased} at time $ t $ if $ w_t (u, v) < w_{t-1} (u, v) $ (and $ (u, v) $ is contained in both $ G_{t-1} $ and $ G_t $).
\end{itemize}
Every deletion, insertion, weight increase, or weight decrease is called an \emph{update}.
The \emph{total number of updates up to time $ t $} of a dynamic (weighted) graph $ \cG $ is denoted by $ \phi_t (\cG) $.
\end{definition}

\begin{definition}[Decremental graph]
A \emph{decremental graph} $ \cG $ is a dynamic graph $ \cG = (G_i)_{0 \leq i \leq k} $ such that for every $ 1 \leq i \leq k $ there is exactly one edge deletion at time $ i $.
Note that $ G_i $ is the graph after the $i$-th edge deletion.
\end{definition}

By our definition decremental graphs are always unweighted.
For a weighted version of this concept it would make sense to additionally allow edge weight increases.
In a decremental graph $ \cG = (G_i)_{0 \leq i \leq k} $ we necessarily have $ k \leq m $ because every edge can be deleted only once.
For decremental shortest paths algorithms the total update time usually does \emph{not} depend on the number of deletions $ k $.
This is the case because of the amortization argument typically used for these algorithms.
For this reason, it will often suffice for our purposes to bound $ \phi_k (\cG) $ or $ |E_k (\cG)| $ by numbers that do not depend on $ k $.

We now formulate the approximate all-pairs shortest paths (APSP) problem we are trying to solve.

\begin{definition}[Distance]
The distance of a node $ x $ to a node $ y $ in a graph $ G $ is denoted by $ \dist_G (x, y) $.
If $ x $ and $ y $ are not connected in $ G $, we set $ \dist_G (x, y) = \infty $.
In a weighted graph $ (H, w) $ the distance of $ x $ to $ y $ is denoted by $ \dist_{H, w} (x, y) $.
\end{definition}

\begin{definition}\label{def:approximate_APSP}
An \emph{$ (\alpha, \beta) $-approximate decremental all-pairs shortest paths (APSP) data structure} for a decremental graph $ \cG = (G_i)_{0 \leq i \leq k} $ maintains, for all nodes $ x $ and $ y $ and all $ 0 \leq i \leq k $, an estimate $ \delta_i (x, y) $ of the distance between $ x $ and $ y $ in $ G_i $.
After the $i$-th edge deletion (where $ 0 \leq i \leq k $), it provides the following operations:
\begin{itemize}
\item \Delete{$u$, $v$}: Delete the edge $ (u, v) $ from $ G_i $.
\item \Distance{$x$, $y$}: Return an estimate $ \delta_i (x, y) $ of the distance between $ x $ and $ \treeroot $ in $ G_i $ such that $ \dist_{G_i} (x, y) \leq \delta_i (x, y) \leq \alpha \dist_{G_i} (x, y) + \beta $.
\end{itemize}
The \emph{total update time} is the total time needed for performing all $ k $ delete operations and the initialization, and the \emph{query time} is the worst-case time needed to answer a single distance query.
The data structure is \emph{exact} if $ \alpha = 1 $ and $ \beta = 0 $.
\end{definition}

Similarly, we define a data structure for decremental single-source shortest paths (SSSP).
We incorporate two special requirements in this definition.
First, we are interested in SSSP data structures that only need to work up to a certain distance range\footnote{In this paper, there are two related parameters $\Q$ (introduced here) representing the ``distance range'' of an SSSP data structure (e.g., the Even--Shiloach tree described in \Cref{sec:ES tree}) and $\q$ (which will be introduced in \Cref{sec:Roditty_Zwick_framework}) representing the ``cover range'' of the center cover data structure.}  $ \Q $ from the source node which is specified by a parameter $ \Q $.
Second, we demand that the data structure tells us whenever a node leaves this distance range.
The latter is a technical requirement that simplifies some of our proofs.

\begin{definition}\label{def:approximate_SSSP}
An \emph{$ (\alpha, \beta) $-approximate decremental single-source shortest paths (SSSP) data structure} with \emph{source (or: root) node $ \treeroot $} and \emph{distance range parameter $ \Q $} for a decremental graph $ \cG = (G_i)_{0 \leq i \leq k} $ maintains, for every node $ x $ and all $ 0 \leq i \leq k $, an estimate $ \delta_i (x, \treeroot) \in \{ 0, 1, \ldots, \lfloor \alpha \Q + \beta \rfloor, \infty \} $ of the distance between $ x $ and $ \treeroot $ in $ G_i $.
After the $i$-th edge deletion (where $ 0 \leq i \leq k $), it provides the following operations:
\begin{itemize}
\item \Delete{$u$, $v$}: Delete the edge $ (u, v) $ from $ G_i $ and return the set of all nodes $ x $ such that $ \delta_i (x, \treeroot) \leq \alpha \Q + \beta $ and $ \delta_{i+1} (x, \treeroot) > \alpha \Q + \beta $.
\item \Distance{$x$}: Return an estimate $ \delta_i (x, \treeroot) $ of the distance between $ x $ and $ \treeroot $ in $ G_i $ such that $ \delta_i (x, \treeroot) \geq \dist_{G_i} (x, \treeroot) $, and if $ \dist_{G_i} (x, \treeroot) \leq \Q $, then also $ \delta_i (x, \treeroot) \leq \alpha \dist_{G_i} (x, \treeroot) + \beta $.
\end{itemize}
The \emph{total update time} is the total time needed for performing all $ k $ delete operations and the initialization, and the \emph{query time} is the worst-case time needed to answer a single distance query.
The data structure is \emph{exact} if $ \alpha = 1 $ and $ \beta = 0 $.
\end{definition}

Finally, we define the remaining notions on graphs we will use.

\begin{definition}[Degree]
We say that \emph{$ v $ is a neighbor of $ u $} if there is an edge $ (u, v) $ in $ G $.
The \emph{degree} of a node $ u $ in the graph $ G $, denoted by $ \deg_G (u) $, is the number of neighbors of $ u $ in $ G $.
The \emph{dynamic degree} of a node $ u $ in a dynamic graph $ \cG $ is $ \deg_\cG (u) = | \{ (u, v) \mid (u, v) \in E_k (\cG) \} | $.
\end{definition}

\begin{definition}[Paths]
Let $ (H, w) $ be a weighted graph, and let $ \pi $ be a path in $ (H, w) $.
The number of nodes on the path $ \pi $ is denoted by $ |\pi| $, and the total weight of the path (i.e., the sum of the weights of its edges) is denoted by $ w (\pi) $.
\end{definition}

\begin{definition}[Connected component]
For every graph $ G $ and every node $ x $ we denote by $ \comp_G (x) $ the connected component of $ x $ in $ G $, i.e., the set of nodes that are connected to $ x $ in $ G $. 
\end{definition}

\subsection{Decremental Shortest-Path Tree Data Structure (Even--Shiloach Tree)} \label{sec:ES tree}

The central data structure in dynamic shortest paths algorithms is the dynamic SSSP tree introduced by Even and Shiloach, in short \emph{ES-tree}. Even and Shiloach~\cite{EvenS81} developed this data structure for undirected, unweighted graphs. Later on, Henzinger and King~\cite{HenzingerK95} observed that it can be adapted to work on directed graphs, and King~\cite{King99} gave a modification for directed, weighted graphs with positive integer edge weights. In the following we review some important properties of this data structure.

We describe an ES-tree on dynamic weighted undirected graphs for a given root node $ \treeroot $ and a given distance range parameter $ \Q $. The data structure can handle arbitrary edge deletions and weight increases.  The data structure maintains, for every node $ v $, a label $ \ell (v) $, called the \emph{level} of $ v $. The level of $ v $ corresponds to the distance between $ v $ and the root $ \treeroot $. Any node $v$ whose distance to $\treeroot$ is more than $\Q$ has $\ell(v)=\infty$. Initially, the values of $\ell(v)$ can be computed in $\tilde O(m)$ time using, e.g., Dijkstra's algorithm. The level $\ell(v)$ implicitly implies the shortest-path tree since the parent of every node $v$ is a node $z$ such that $\ell(v)=\ell(z)+w(v, z)$. (Every node $v$ such that $\ell(v)=\infty$ will not be in the shortest-paths tree.) Every deletion of an edge $ (u, v) $ possibly affects the levels of several nodes. The algorithm tries to adjust the levels of these nodes as follows. 

\begin{figure}
\hspace{-1cm}%
\centering
\begin{subfigure}[t]{0.34\textwidth}
\centering
\scalebox{0.95}{\colorlet{darkgreen}{green!70!black}

\tikzstyle{level0}=[circle,fill=darkgreen!20,draw=darkgreen,thick,minimum size=20pt,inner sep=0pt]
\tikzstyle{levelother}=[circle,fill=blue!10,draw=blue,thick,minimum size=20pt,inner sep=0pt]
\tikzstyle{levellabel}=[rectangle,minimum size=20pt,inner sep=0pt]
\tikzstyle{edge-black} = [draw,ultra thick,-]
\tikzstyle{edge-grey} = [draw,thick,-,color=lightgray]

\begin{tikzpicture}[node distance=1.2cm and 0.5cm]

\node[levellabel] (level0) at (-2.5,0) {\small level 0};
\node[levellabel] (level1) [below= of level0] {\small level 1};
\node[levellabel] (level2) [below= of level1] {\small level 2};
\node[levellabel] (level3) [below= of level2] {\small level 3};

\node[level0] (r) at (0,0) {$r$};

\node[levelother] (b) [below= of r] {$b$};
\node[levelother] (a) [left= of b] {$a$};
\node[levelother] (c) [right= of b] {$c$};

\node[levelother] (d) [below= of a] {$d$};
\node[levelother] (e) [below= of b] {$e$};

\path[edge-black] (r) -- (a);
\path[edge-black] (r) -- (b);
\path[edge-black] (r) -- (c);
\path[edge-black] (a) -- (d);
\path[edge-black] (a) -- (e);

\path[edge-grey] (a) -- (b);
\path[edge-grey] (b) -- (e);
\path[edge-grey] (c) -- (e);
\path[edge-grey] (d) -- (e);

\end{tikzpicture}}
\caption{}\label{fig:ES-tree-one}
\end{subfigure}
\begin{subfigure}[t]{0.32\textwidth}
\centering
\scalebox{0.95}{\colorlet{darkgreen}{green!70!black}

\tikzstyle{level0}=[circle,fill=darkgreen!20,draw=darkgreen,thick,minimum size=20pt,inner sep=0pt]
\tikzstyle{levelother}=[circle,fill=blue!10,draw=blue,thick,minimum size=20pt,inner sep=0pt]
\tikzstyle{levellabel}=[rectangle,minimum size=20pt,inner sep=0pt]
\tikzstyle{edge-black} = [draw,ultra thick,-]
\tikzstyle{edge-grey} = [draw,thick,-,color=lightgray]
\tikzstyle{edge-grey} = [draw,thick,-,color=lightgray]
\tikzstyle{edge-red} = [draw,thick,-,color=red,dotted]
\tikzstyle{message} = [draw,thick,->,color=violet,dashed]

\tikzset{cross/.style={cross out, draw=black, ultra thick, minimum size=5*(#1-\pgflinewidth), inner sep=0pt, outer sep=0pt}, 
cross/.default={3pt}}

\begin{tikzpicture}[node distance=1.2cm and 0.5cm]

\node[level0] (r) at (1,0) {$r$};

\node[levelother] (b) [below= of r] {$b$};
\node[levelother] (c) [right= of b] {$c$};

\node[levelother] (e) [below= of b] {$e$};
\node[levelother] (d) [left= of e] {$d$};
\node[levelother] (a) [left= of d] {$a$}
	edge[ultra thick,-,bend right=45] (e);

\node[levelother,draw=none,fill=none] (x) [below= of a] {};

\path[edge-black] (r) -- (b);
\path[edge-black] (r) -- (c);
\path[edge-black] (a.north) -- (b);
\path[edge-black] (a) -- (d);

\path[edge-grey] (b) -- (e);
\path[edge-grey] (c) -- (e);
\path[edge-grey] (d) -- (e);

\path[edge-red] (r) -- (a.north)
	node [align=center,midway,cross,red,solid] {};

\path[message] ($ (a.south east) + (0.5,-0.6) $) -- ($ (a.south east) + (1.5,-0.6) $)
	node[below,align=center,midway,violet] {\tiny $\mathrm{level}(a)=2$};

\path[message] ($ (a.north east) + (0.75,0.85) $) -- ($ (b.south west) + (-0.4,-0.05) $)
	node[above,align=center,midway,violet,rotate=34] {\tiny $\mathrm{level}(a)=2$};
	
\path[message] ($ (a.east) + (0.1,0.2) $) -- ($ (d.west) + (-0.1,0.2) $)
	node[above,align=right,violet] {\vspace{1cm} \tiny ~~~$\mathrm{level}(a)=2$};

\end{tikzpicture}}
\caption{}\label{fig:ES-tree-two}
\end{subfigure}
\begin{subfigure}[t]{0.32\textwidth}
\centering
\scalebox{0.95}{\colorlet{darkgreen}{green!70!black}

\tikzstyle{level0}=[circle,fill=darkgreen!20,draw=darkgreen,thick,minimum size=20pt,inner sep=0pt]
\tikzstyle{levelother}=[circle,fill=blue!10,draw=blue,thick,minimum size=20pt,inner sep=0pt]
\tikzstyle{levellabel}=[rectangle,minimum size=20pt,inner sep=0pt]
\tikzstyle{edge-black} = [draw,ultra thick,-]
\tikzstyle{edge-grey} = [draw,thick,-,color=lightgray]
\tikzstyle{edge-grey} = [draw,thick,-,color=lightgray]
\tikzstyle{edge-red} = [draw,thick,-,color=red,dotted]
\tikzstyle{message} = [draw,thick,->,color=violet,dashed]

\tikzset{cross/.style={cross out, draw=black, ultra thick, minimum size=5*(#1-\pgflinewidth), inner sep=0pt, outer sep=0pt}, 
cross/.default={3pt}}

\begin{tikzpicture}[node distance=1.2cm and 0.5cm]

\node[level0] (r) at (1,0) {$r$};

\node[levelother] (b) [below= of r] {$b$};
\node[levelother] (c) [right= of b] {$c$};

\node[levelother] (e) [below= of b] {$e$};
\node (leer) [left= of e] {};
\node[levelother] (a) [left= of leer] {$a$};
	
\node[levelother] (d) [below= of a] {$d$};

\path[edge-black] (r) -- (b);
\path[edge-black] (r) -- (c);
\path[edge-black] (a.north) -- (b);
\path[edge-black] (a) -- (d);

\path[edge-black] (b) -- (e);
\path[edge-grey] (c) -- (e);
\path[edge-grey] (d) -- (e);
\path[edge-grey] (a) -- (e);

\path[edge-red] (r) -- (a.north)
	node [align=center,midway,cross,red,solid] {};

\path[message] ($ (d.north west) + (0,0.2) $) -- ($ (a.south west) + (0,-0.2) $)
	node[above,align=center,midway,violet,rotate=90] {\tiny $\mathrm{level}(d)=3$};

\path[message] ($ (d.east) + (0.25,0.45) $) -- ($ (e.south west) + (-0.4,-0.5) $)
	node[below,align=center,midway,violet,rotate=45] {\tiny $\mathrm{level}(d)=3$};

\end{tikzpicture}}
\caption{}\label{fig:ES-tree-three}
\end{subfigure}
\caption{Example of the view of the ES-tree as nodes talking to each other. (\subref{fig:ES-tree-one}) The ES-tree before the edge deletion. (\subref{fig:ES-tree-two}) After deleting the edge $(r, a)$, the level of the node $a$ changes to $2$. The node $a$ sends a message to all neighbors to inform them about this change. (\subref{fig:ES-tree-three}) This causes the node $d$ to change its level, and thus $d$ sends a message to inform its neighbors. There are no other changes, so the new ES-tree is as in (\subref{fig:ES-tree-three}). Thus, there are $5$ messages involved in constructing the new ES-tree, and \Cref{alg:ES_tree} shows that this process can be implemented in $5$ time units.}\label{fig:ES-tree}
\end{figure}
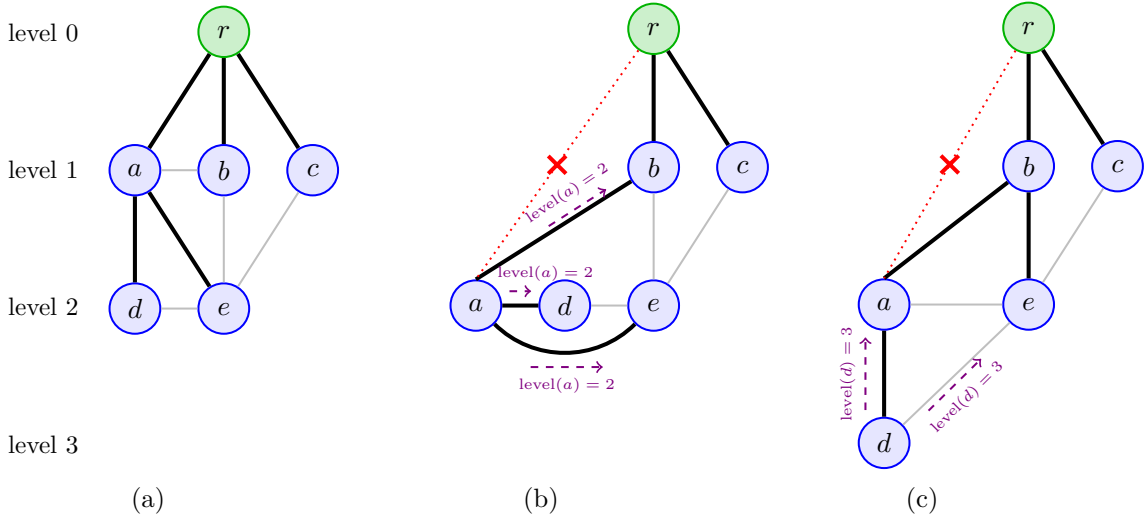

\paragraph*{Informal Description.}
How the ES-tree handles deletions can be intuitively viewed as nodes in the input graph talking to each other as follows. Imagine that every node $v$ in the input graph is a computing unit that tries to maintain its level $\ell(v)$ corresponding to its current distance to the root. It knows the levels of its neighbors and has to make sure that 
\begin{align}
\ell(v)&=\min_{u} \left(\ell(u)+w(u, v)\right)\label{eq:ES tree level condition}
\end{align} 
where the minimum is over all current neighbors $u$ of $v$. When we delete an edge incident to $v$, the value of $\ell(v)$ might change. If this happens, $v$ {\em sends a message} to each of its neighbors to inform about this change, since the levels of these nodes might have to change as well. Every neighbor of $v$ then updates its level accordingly, and if its level changes, it sends messages to its neighbors (including $v$), too.
(See \Cref{fig:ES-tree} for an example.) An important point, which we will show soon, is that {\em we can implement the ES-tree in time  proportional to the number of messages}. This means that when a node $v$'s level is changed, we can bound the time we need to maintain the ES-tree by its current degree. 
Thus, the contribution of a node $v$ to the running time to update the ES-tree after the $i$-th deletion or weight increase is $\deg_{G_i}(v)$ times its level change, i.e., $\min(\Q, \ell_{i} (v)) - \min(\Q, \ell_{i-1} (v))$ (the minimum is to avoid the case where $\ell_{i}(v)=\infty$). 
This intuitively leads to the following lemma.
\begin{lemma}[King \cite{King99}]\label{thm:king}
The ES-tree is an exact decremental SSSP data structure for shortest paths up to a given length $ \Q $.
It has constant query time, and in a decremental graph $ \cG = (G_i, w_i)_{0 \leq i \leq k} $ its total update time can be bounded by
\begin{equation*}
O \left(\phi_k(\cG) + t_{\mathrm{SP}} + \sum_{1 \leq i \leq k} \sum_{v \in V} \deg_{G_i} (v) \left(\min(\Q, \ell_{i} (v)) - \min(\Q, \ell_{i-1} (v))\right) \right) \, ,
\end{equation*}
where $ t_{\mathrm{SP}} $ is the time needed for computing an SSSP tree up to depth $ \Q $ and, for $ 0 \leq i \leq k $, $ \ell_i (v) $ is the level of $ v $ after the ES-tree has processed the $i$-th deletion or weight increase.
\end{lemma}

Recall that $\phi_k(\cG)$ is the total number of updates (deletions and weight increases). Note that when the graph is unweighted, only deletions are allowed. In this case, the $\phi_k(\cG)$ term can be ignored. 
\Cref{thm:king} can be simplified by using two specific bounds. These bounds are in fact what we need later in this paper. 

\begin{corollary}\label{cor:ES_running_time_distance_increase}
There is an exact decremental SSSP data structure for paths up to a given length $ \Q $ that has constant query time, and in a decremental graph $ \cG = (G_i, w_i)_{0 \leq i \leq k} $ with source node $ \treeroot $ its total update time can be bounded by
\begin{equation*}\label{eq:ES running time simplified}
O \left(\phi_k(\cG) + t_{\mathrm{SP}} + \sum_{v \in V} \deg_{G_0} (v) \cdot \left( \min(\Q, \dist_{G_k}(v, \treeroot)) - \min(\Q, \dist_{G_0}(v, \treeroot)) \right) \right)
\end{equation*}
and $ O (m \Q) $,
where $ t_{\mathrm{SP}} $ is the time needed for computing an SSSP tree up to depth $ \Q $, and $ \dist_{G_0}(v, \treeroot) $ is the initial distance of $ \treeroot $ to $ v $ and $ \dist_{G_k}(v, \treeroot) $ is the distance of $ v $ to $ \treeroot $ after all $ k $ edge deletions.
\end{corollary}

The first bound in \Cref{cor:ES_running_time_distance_increase} is because, for every node $v$, we can use $\deg_{G_i}(v)\leq \deg_{G_0}(v)$, and we can express the running time caused by $v$'s level change in terms of its initial level ($\min(\Q, \dist_{G_0}(v, \treeroot)$) and its final level ($\min(\Q, \dist_{G_k}(v, \treeroot))$). We will need this bound in \Cref{sec:deterministic}.
The second bound follows easily from the first one and we will need it in \Cref{sec:faster}.

\begin{algorithm2e}
\caption{ES-tree}
\label{alg:ES_tree}

\tcp{\textrm{The ES-tree is formulated for weighted undirected graphs.}}
\BlankLine

\tcp{\textrm{Internal data structures:
\begin{itemize}
\item $ N (u) $: for every node $ u $ a heap $ N (u) $ whose intended use is to store for every neighbor $ v $ of $ u $ in the current graph the value of $ \lev (v) + w (u, v) $, where $ w (u, v) $ is the weight of the edge $ (u, v) $ in the current graph
\item $ Q $: global heap whose intended use is to store nodes whose levels might need to be updated
\end{itemize}
}}
\vspace{-3ex}

\BlankLine

\Procedure{\Initialize{}}{
	Compute shortest paths tree from $ \treeroot $ in $ (G_0, w_0) $ up to depth $ \Q $\;
	\ForEach{node $ u $}{
		Set $ \lev (u) = \dist_{G_0} (u, \treeroot) $\;
		\lFor{every edge $ (u, v) $}{
			insert $ v $ into heap $ N(u) $ of $ u $ with key $ \lev(v) + w(u, v) $
		}
	}
}

\Procedure{\Delete{$u$, $v$}}{
	\Increase{$u$, $v$, $\infty$} 
}

\Procedure{\Increase{$u$, $v$, $w(u, v)$}}{
	\tcp{\textrm{Increase weight of edge $ (u, v) $ to $ w(u, v) $}}
	Insert $ u $ and $v$ into heap $ Q $ with keys $ \lev(u) $ and $\lev(v)$, respectively\;\label{line:insert u}
	Update key of $ v $ in heap $ N(u) $ to $ \lev(v) + w(u, v) $ and key of $ u $ in heap $ N(v) $ to $ \lev(u) + w(u, v) $\;\label{line:update N after increase}
	\UpdateLevels{}\;
}

\Procedure{\UpdateLevels{}}{
	\While{heap $ Q $ is not empty}{\label{line:while loop}
	    Take node $ y $ with minimum key $ \lev (y) $ from heap $ Q $ and remove it from $ Q $ \label{line: take y from Q}\;
		$ \lev' (y) \gets \min_{z} (\lev (z) + w (y, z)) $ \label{line: update ell(y)}\;
		\tcp{\textrm{$\lev'(y)$ can be retrieved from the heap $ N(y) $.  $\arg\min_{z} (\lev (z) + w (y, z)) $ is $y$'s parent in the ES-tree}}
		\If{$ \lev'(y) > \lev (y) $}{
			$\lev(y)\gets \lev'(y)$\;
			\lIf{$ \lev' (y) > \Q $}{
				$ \lev (y) \gets \infty $\label{line:set to infty}
			}

		\ForEach{neighbor $ x $ of $ y $}{\label{line: update neighbors' heaps}
				update key of $ y $ in heap $ N(x) $ to $ \lev(y) + w(x, y) $\;
				insert $ x $ into heap $ Q $ with key $ \lev(x) $ if $Q$ does not already contain~$ x $
			}
		}
	}
}
\end{algorithm2e}

\paragraph*{Implementation.}
The pseudocode for achieving the above result can be found in \Cref{alg:ES_tree}. (For simplicity we show an implementation using heaps, which causes an extra $\log n$ factor in the running time. King~\cite{King99} explains how to avoid heaps in order to improve the running time by a factor of $ \log{n} $.)
For every node $ x $ the ES-tree maintains a heap $ N (x) $ that stores for every neighbor $ y $ of $ x $ in the current graph the value of $ \ell (y) + w (x, y) $ where $ w (x, y) $ is the weight of the edge $ (x, y) $ in the current graph. (Intuitively, $N(x)$ corresponds to the ``knowledge'' of $x$ about its neighbors.) These data structures can be initialized in $\tilde O(m)$ time by running, for example, Dijkstra's algorithm (see procedure \Initialize{}).\footnote{Alternatively we could compute the initial shortest paths tree using the Even--Shiloach algorithm itself: Let $ G_0' $ be the modification of $ G_0 $ where we add an edge $ (\treeroot, v) $ of weight $ 1 $ for every node $ v $. We obtain $ G_0 $ from $ G_0' $ by deleting each such edge. Starting from a trivial shortest paths tree in $ G_0' $ in which the parent of every node $ v \neq \treeroot $ is $ \treeroot $, we obtain the shortest paths tree of $ G_0 $ in time $ O (\sum_{v \in V} \deg_{G_0} (v) \cdot \min(\Q, \dist_{G_0}(v, \treeroot))) = O (m \Q) $.}

Edge deletions and weight increases are handled in procedure \Delete{} and \Increase{}, respectively; in fact, deletion is a special case of weight increase where we set the edge weight to $\infty$. 
Every weight increase of an edge $ (u, v) $ might cause the levels of some nodes to increase. The algorithm uses a heap $Q$ to keep track of such nodes. Initially (at the time $w(u, v)$ is increased) the algorithm inserts $u$ and $v$ to $Q$ as the levels of $u$ and $v$ might increase (see Line~\ref{line:insert u}). It also updates $N(u)$ and $N(v)$ as in Line~\ref{line:update N after increase}. Then it updates the levels on nodes in $Q$ using procedure \UpdateLevels{}. 

Procedure \UpdateLevels{} processes the nodes in $Q$ in the order of their current level (see the while-loop starting on Line~\ref{line:while loop}). In every iteration it will process $y$ in $Q$ with smallest $\ell(y)$ (as in Line \ref{line: take y from Q}). The lowest level that is possible for a node $ y $ is $\ell'(y)= \min_{z} (\ell (z) + w (y, z)) $, the minimum of $ \ell (z) + w(y, z) $ over all neighbors $ z $ of $ y $ in the current graph (following \Cref{eq:ES tree level condition}). 
Therefore every node $ y $ will repeatedly update its level to $ \ell' (y) $ (unless its level already has this value); see Line \ref{line: update ell(y)}. (An exception is when the level of a node $ x $ exceeds the desired depth $ \Q $. In this case the level of $ x $ is set to $ \infty $ and $ x $ will never be connected to the tree again. See Line~\ref{line:set to infty}.) If this updating rule leads to a level increase, the algorithm has to update the heap $ N(x) $ of every neighbor $ x $ and put $x$ to the heap $Q$ (since the level of $x$ might increase), as in the for-loop starting on Line \ref{line: update neighbors' heaps} (this is equivalent to having $y$ send a message to $x$ in the informal description). 

The running time analysis takes into account the level increases occurring in the ES-tree.
It is based on the following observation: For every node $ x $ processed in the while-loop of the procedure \UpdateLevels{}  in \Cref{alg:ES_tree}, if the level of $ x $ increases, the algorithm has to spend time $ O (\deg(x) \log{n}) $ updating the heaps $N(y)$ of all neighbors $y$ of $x$ and adding these neighbors to heap $ Q $. If the level of $ x $ does not increase, the algorithm only has to spend time $ O (\log{n}) $. In the second case the running time can be charged to one of the following events that causes $x$ to be in~$Q$: (1) a weight increase of some edge $(x, y)$, and (2) a level increase of some neighbor of $ x $. This leads to the result in \Cref{thm:king}.

\subsection{The Framework of Roditty and Zwick}\label{sec:Roditty_Zwick_framework}

In the following we review the algorithm of Roditty and Zwick~\cite{RodittyZ12} because its main ideas are the basis of our own algorithms.
We will put their arguments in a certain structure that clarifies for which part of the algorithm we obtain improvements.
Their algorithm is based on the following observation.
Consider approximating the distance $\dist_G(x, y)$ for some pair of nodes $x$ and $y$. For some $0<\epsilon\leq 1$, we want a $(1+O(\epsilon), 0)$-approximate value of $\dist_G(x, y)$. Assume that we know an integer $p$  such that $2^p$ is a ``distance guess'' of $\dist_G(x, y)$, i.e.,
\begin{align}2^p \leq \dist_G (x, y) \leq 2^{p+1}.\label{eq:distance guess}\end{align}
Now, suppose that there is a node $z$ that is close to $x$, i.e., 
\begin{align} \dist_G (x, z) \leq \epsilon 2^p.\label{eq:near center}
\end{align}
Then it follows that we can use $ \dist_G (x, z) + \dist_G (z, y) $ as a $ (1 + 2 \epsilon) $-approximation of the true distance $ \dist_G (x, y) $; this follows from applying the triangle inequality twice (also see \Cref{fig:CenterCover-one}):
\begin{multline}
\dist_G (x, y) \leq \dist_G (x, z) + \dist_G (z, y) \\ \leq \dist_G (x, z)+(\dist_G (z, x)+\dist_G(x, y)) \leq (1 + 2 \epsilon) \dist_G (x, y) \, . \label{eq:approximate using center}
\end{multline}
Thus, under the assumption that for any $ x $ we only want to determine the distances from $ x $ to nodes $ y $ with $ \dist_G (x, y) $ in the range from $ 2^p $ to $ 2^{p+1} $, we only have to make sure that there is a node $z$ that satisfies \Cref{eq:near center}; we call such node $z$ a {\em center}. We will maintain a set $ U $ of nodes such that for every node $x$ there is a node $ z \in U $ that satisfies \Cref{eq:near center}. In fact, we only need this to be true for nodes $x$ that are in a ``big'' connected component since if the connected component containing $x$ is too small, then there is no node $y$ that satisfies \Cref{eq:distance guess}.
We call such $ U $ a \emph{center cover}. The following definition states this more precisely.

\begin{figure}
\hspace{-1cm}%
\centering
\begin{subfigure}[b]{0.45\textwidth}
\centering
\scalebox{0.88}{\tikzstyle{vertex}=[circle,fill=black,minimum size=5pt,inner sep=0pt,outer sep=0pt]
\tikzstyle{vertez}=[draw=blue!50,fill=blue!20,minimum size=5pt,inner sep=0pt,outer sep=0pt]
\tikzstyle{shortest-path} = [draw,thick]
\tikzstyle{shortest-path-dash} = [draw,thick,dashed]
\tikzstyle{hop-path-small} = [draw,ultra thick,-,color=blue]
\tikzstyle{hop-path-big} = [draw,ultra thick,-,color=red]

\begin{tikzpicture}

\draw[thin,draw=white] (1,-1) rectangle (2,0);

\node[vertex] (x) at (0,0) {};
\node[below] at (x.south) {$x$};

\node[vertex] (y) at (6,0) {};
\node[right] at (y.east) {$y$};

\node[vertez] (z) at (1,2) {};
\node[above] at (z.north) {$z$};

\path[shortest-path] (x) -- (y) 
	node [below,align=center,midway] {$2^p \leq \dist_G(x,y) \leq 2^{p+1}$};
\path[shortest-path-dash] (x) -- (z)
	node [above,align=center,midway,rotate=63.43] {$\dist_G(x,z) \leq \epsilon 2^p$};
\path[shortest-path-dash] (y) -- (z)
	node [above,align=center,midway,rotate=-21.8] {$\dist_G(y,z) \leq \dist_G(x,y) + \epsilon 2^p$};
	
\draw[->,red,thick] (0.4,0.2) .. controls (1,1.9) .. (5,0.1);
\node[red] at (2,0.5) {$\leq (1+2\epsilon)\dist_G(x,y)$};

\end{tikzpicture}}
\caption{}\label{fig:CenterCover-one}
\end{subfigure}
\hspace{0.25cm}
\begin{subfigure}[b]{0.45\textwidth}
\centering
\hspace{-1cm}%
\scalebox{0.88}{\pgfdeclarelayer{foreground}
\pgfdeclarelayer{background}    
\pgfsetlayers{background,main,foreground}  

\begin{tikzpicture}

\tikzstyle{vertex}=[circle,fill=black,minimum size=5pt,inner sep=0pt,outer sep=0pt]
\tikzstyle{vertez}=[draw=blue!50,fill=blue!20,minimum size=5pt,inner sep=0pt,outer sep=0pt]
\tikzstyle{shortest-path} = [draw,thick]
\tikzstyle{shortest-path-dash} = [draw,thick,dashed]
\tikzstyle{hop-path-small} = [draw,ultra thick,-,color=blue]
\tikzstyle{hop-path-big} = [draw,ultra thick,-,color=red]

\begin{pgfonlayer}{foreground} 
\node[vertex] (x) at (0,0) {};
\node[right] at (x.east) {$x$};

\node[vertex] (y) at (6,0) {};
\node[right] at (y.east) {$y$};

\node[vertez] (z) at (1,2) {};
\node[above] at (z.north) {$z$};

\draw[thick,draw=blue!50] (z) -- ++(30:2.5cm)
	node [above,align=center,midway,rotate=30] {$R^c$};

\draw[thick,draw=blue!50] (z) -- ++(-15:5.5cm)
	node [above,align=center,midway,rotate=-15] {$R^d$};
\end{pgfonlayer}

\filldraw[fill=white,draw=blue!50] (1,2) circle (2.5cm);

\begin{pgfonlayer}{background}    
	\clip (-2,-1) rectangle (6.5,5);	
	\filldraw[fill=blue!10,draw=blue!50] (1,2) circle (5.5cm);
\end{pgfonlayer}

\end{tikzpicture}}
\caption{}\label{fig:CenterCover-two}
\end{subfigure}

\caption{(\subref{fig:CenterCover-one}) depicts \Cref{eq:distance guess,eq:near center,eq:approximate using center}. (\subref{fig:CenterCover-two}) shows the cover range (small circle) and distance range (big circle) used by the center cover data structure (\Cref{def:CenterCover}).}\label{fig:CenterCover}
\end{figure}

\begin{definition}[Center cover]\label{def:center cover simple}
Let $ U $ be a set of nodes of a graph $ G $, and let $ \q $ be a positive integer denoting the \emph{cover range}.
We say that a node $ x $ is \emph{covered} by a node $ \cen \in U $ in $ G $ if $ \dist_G (x, \cen) \leq \q $.
We say that $ U $ is a \emph{center cover of $ G $ with parameter $ \q $} if every node $ x $ that is in a connected component of size at least $ \q $ is covered by some node $ \cen \in U $ in $ G $.
\end{definition}

One main component of Roditty and Zwick's framework as we describe it is the {\em center cover data structure}. This data structure maintains a center cover $U$ as above. Furthermore, for every center $ z \in U $, we will maintain the distance to every node $ y $ such that $ \dist_G(z, y) \leq 2^{p+2} $. This will allow us to compute $\dist_G(x, z)+\dist_G(z, y)$ as an approximate value of $\dist_G(x, y)$ (as in \Cref{eq:approximate using center}). In general, we treat the number $2^{p+2}$ as another parameter of the data structure denoted by $\Q$ (called distance range parameter). The values of $\q$ and $\Q$ are typically closely related; in particular, $\q\leq \Q = O(\q)$. The center cover data structure is defined formally as follows (also see \Cref{fig:CenterCover-two}). 

\begin{definition}[Center cover data structure]\label{def:CenterCover}
A \emph{center cover data structure} with \emph{cover range parameter $ \q $} and \emph{distance range parameter $ \Q $} for a decremental graph $ \cG = (G_i)_{0 \leq i \leq k} $ maintains, for every $ 0 \leq i \leq k $, a set of \emph{centers} $ \C_i = \{ 1, 2, \ldots, l \} $ and a set of nodes $ U_i = \{ \cen_i^1, \cen_i^2, \ldots, \cen_i^l \} $ such that $ U_i $ is a center cover of $ G_i $ with parameter $ \q $.
For every center $ j \in \C_i $ and every $ 0 \leq i \leq k $, we call $ \cen_i^j \in U_i $ the \emph{location} of center $ j $ in $ G_i $, and for every node $ x $ we say that $ x $ is covered by $ j $ if $ x $ is covered by $ \cen^j_i $ in $ G_i $.
After the $i$-th edge deletion (where $ 0 \leq i \leq k $), the data structure provides the following operations:
\begin{itemize}
\item \Delete{$u$, $v$}: Delete the edge $ (u, v) $ from $ G_i $.
\item \Distance{$j$, $x$}: Return the distance $ \dist_{G_i} (\cen_i^j, x) $ between the location $ \cen_i^j $ of center $ j $ and the node $ x $, provided that $ \dist_{G_i} (\cen_i^j, x) \leq \Q $.
If $ \dist_{G_i} (\cen_i^j, x) > \Q $, then return $ \infty $.
\item \FindCenter{$x$}: If $ x $ is in a connected component of size at least $ \q $ in $ G_i $, return a center $ j $ (with location $ \cen_i^j $) such that $ \dist_{G_i} (x, \cen_i^j) \leq \q $.
If $ x $ is in a connected component of size less than $ \q $ in $ G_i $, then either return $ \bot $ or return a center $ j $ (with location $ \cen_i^j $) such that $ \dist_{G_i} (x, \cen_i^j) \leq \q $.
\end{itemize}
The \emph{total update time} is the total time needed for performing all $ k $ delete operations and the initialization.  The \emph{query time} is the worst-case time needed to answer a single {\tt distance} or {\tt findCenter} query.
\end{definition}

As the update time of the data structure will depend on the number $ l $ of centers, the goal is to keep $ l $ as small as possible, preferably $ l = \tilde O (n / \q) $. As an example, consider the following {\em randomized} implementation of Roditty and Zwick \cite{RodittyZ12}: randomly pick a set $U$ of $((n/\q)\poly\log n)$ nodes as the set of centers. It can be shown that, with high probability, this set will remain a center cover during all deletions. The {\tt distance} and {\tt findCenter} queries can be answered in $O(1)$ time by maintaining an ES-tree of depth $m\Q$ for every center. The total time to maintain this data structure is thus $\tilde O(mn\Q/\q)$. We typically set $\q=\Omega(\Q)$. In this case, the total time becomes $\tilde O(mn)$.

Note that while the implementation of Roditty and Zwick always uses the same set of centers $U$, the center cover data structure that we define is flexible enough to allow this set to {\em change over time}: i.e., it is possible that $U_i\neq U_{i+1}$ for some $i$. In fact, our definition separates between the notion of {\em centers} (set $C_i$) and {\em locations} (set $U_i$) as it will allow one center to change its location over time. 
This is necessary when we want to maintain $ o(n) $ centers deterministically since if we fix the centers and their locations, then an adversary can delete all edges adjacent to the centers, making all noncenter nodes uncovered. (The randomized algorithm of Roditty and Zwick can avoid this by using randomness and assuming that the adversary is oblivious.)

\paragraph*{Using Center Cover Data Structure to Solve APSP.}
Given a center cover data structure, an approximate decremental APSP data structure is obtained as follows. We maintain $ \lceil \log n \rceil $ ``instances'' of the center cover data structure where the $ p $-th instance has parameters $ \q = \epsilon 2^p $ and $ \Q = 2^{p+2} $ and is responsible for the distance range from $ 2^p $ to $ 2^{p+1} $ (for all $ 0 \leq p \leq \lfloor \log{n} \rfloor $).  Suppose that after the $i$-th deletion we want to answer a query for the approximate distance between the nodes $ x $ and~$ y $. For every $ p $, we first query for a center covering $ x $ from the $p$-th instance of the center cover data structure. Denote the {\em location} of this center by $z_p$. 
The distance estimate provided by the $ p $-th instance is $\dist_{G_i}(z_p, x)+\dist_{G_i}(z_p, y)$. We will output $\min_p \dist_{G_i}(z_p, x)+\dist_{G_i}(z_p, y)$ as an estimate of $\dist_{G_i}(x, y)$. 
(Note that it is possible that $z_p=\bot$; i.e., there is no center covering $x$ in the $p$-th instance. This might happen if $x$ is in a connected component of size less than $\q$. In this case we set $\dist_{G_i}(z_p, x)+\dist_{G_i}(z_p, y)=\infty$.)

To see the approximation guarantee, let $p^*$ be such that $ 2^{p^*} \leq \dist_{G_i} (x, y) \leq 2^{p^*+1} $.
Observe that if $ p = p^* $, then $\dist_{G_i}(z_p, x)+\dist_{G_i}(z_p, y)$ is a $ (1+O(\epsilon), 0) $-approximate distance estimate (due to \Cref{eq:approximate using center}), and if $p\neq p^*$, then $\dist_{G_i}(z_p, x)+\dist_{G_i}(z_p, y) \geq \dist_{G_i}(x, y)$ (by the triangle inequality). Thus, $\min_p \dist_{G_i}(z_p, x)+\dist_{G_i}(z_p, y)$ is a $(1+O(\epsilon), 0)$-approximate value of $\dist_{G_i}(x, y)$. The query time, which is the time to compute  $\min_p \dist_{G_i}(z_p, x)+\dist_{G_i}(z_p, y)$, is $O(\log n)$.  

The query time can be reduced to $ O (\log{\log{n}}) $ as follows. Observe that for any $ p < p^* $, the distance  $\dist_{G_i}(z_p, x)+\dist_{G_i}(z_p, y)$ might be $ \infty $ if $ \dist_{G_i} (z_p, y) > \Q $; however, if it is finite, it will provide a $ (1+\epsilon, 0) $-approximation (since $\dist_{G_i}(z_p, x)\leq \epsilon p$).  
In other words, it suffices to find the smallest index $ p^{**} $ for which $\dist_{G_i}(z_{p^{**}}, x)+\dist_{G_i}(z_{p^{**}}, y)$ is finite; this value will be a $(1+O(\epsilon), 0)$-approximate value of $\dist_{G_i} (x, y)$. 
To find this index, observe further that for any $ p > p^* $, either $z_p=\bot$ or $\dist_{G_i}(z_p, x)+\dist_{G_i}(z_p, y)$ is finite. So, we can find $p^{**}$ by a binary search (since for any $p$, if  $z_p=\bot$ or $\dist_{G_i}(z_p, x)+\dist_{G_i}(z_p, y)$ is finite, then we know that $p^{**}\leq p$). 

\begin{theorem}[\cite{RodittyZ12}]\label{lem:centers_data_structure_to_APSP}
Assume that for all parameters $ \q $ and $ \Q $ such that $ \q \leq \Q $ there is a center cover data structure that has constant query time and a total update time of $ T (\q, \Q) $.
Then, for every $ \epsilon \leq 1 $, there is a $ (1 + \epsilon, 0) $-approximate decremental APSP data structure with $ O(\log{\log n}) $ query time and a total update time of $ \sum_p T (\qp, \Qp) $ where $ \qp = \epsilon 2^p $ and $ \Qp = 2^{p+2} $ (for $ 0 \leq p \leq \lfloor \log{n} \rfloor $).
\end{theorem}

As shown before, Roditty and Zwick~\cite{RodittyZ12} obtain a randomized center cover data structure with constant query time and a total update time of $ \tilde O (m n \Q / \q) $.
By \Cref{lem:centers_data_structure_to_APSP} they get a $ (1 + \epsilon, 0) $-approximate decremental APSP data structure with a total update time of $ \tilde O (m n / \epsilon) $ and a query time of $ O (\log{\log{n}}) $.
Note that the query time can sometimes be reduced further to $O(1)$, and this is the case for their algorithm as well as our randomized algorithm in \Cref{sec:faster}. This is essentially because there is a $(3, 0)$-approximation randomized algorithm for APSP, which can be used to approximate $p^*$ (we defer details to \Cref{lem:approximate_centers_data_structure_to_APSP}). 
To analyze the total update time of their data structure for $ k $ deletions, observe that
\begin{align*}
\sum_{p=0}^{\lfloor \log{n} \rfloor} \tilde  O (m n \Qp / \qp) = \sum_{p=0}^{\lfloor \log{n} \rfloor} \tilde O (m n 2^{p+1}/(2^{p} \epsilon)) &= \sum_{p=0}^{\lfloor \log{n} \rfloor} \tilde O (m n / \epsilon) \\
 &= \tilde O (m n \log{n} / \epsilon) = \tilde O (m n / \epsilon) \, . 
\end{align*}

In \Cref{sec:faster} we show that we can maintain an \emph{approximate} version of the center cover data structure in time $ \tilde O (n^{5/2} \Q / (\epsilon \q)) $.
Using this data structure, we will get a  $ (1 + \epsilon, 2) $-approximate decremental APSP data structure with a total update time of $ \tilde O (n^{5/2} / \epsilon) $ and constant query time.
In \Cref{sec:deterministic} we show how to maintain an exact \emph{deterministic} center cover data structure with a total update time of $ O (m n \Q / \q) $.
By \Cref{lem:centers_data_structure_to_APSP} this immediately implies a deterministic $ (1 + \epsilon, 0) $-approximate decremental APSP data structure with a total update time of $ O (m n \log{n}) $ and a query time of $ O (\log{\log{n}}) $.

\section{$ \tilde O (n^{5/2}) $-Total Time $(1+\epsilon, 2)$- and $(2+\epsilon, 0)$-Approximation Algorithms}\label{sec:faster}

In this section, we present a data structure for maintaining APSP under edge deletions with multiplicative error $ 1+\epsilon $ and additive error $ 2 $ that has a total update time of $\tilde O(n^{5/2} / \epsilon^2)$.
The data structure is correct with high probability.
We also show a variant of this data structure with multiplicative error $ 2 + \epsilon $ and no additive error.
In doing this, we introduce the notion of a {\em persevering path} (see \Cref{def:decremental path}) and a {\em locally persevering emulator} (\Cref{def:locally decremental emulator}). In \Cref{sec:emulator}, we then present the locally persevering emulator that we will use to obtain our result.
Then in \Cref{sec:monotone ES tree} we explain our main technique, called the {\em monotone Even--Shiloach tree}, where we maintain the distances from a single node to all other nodes, up to some distance $\Q$, in a locally persevering emulator. (Recall that $\Q$ is a parameter called the ``distance range.'')
In \Cref{sec:approximate_SSSP_to_approximate_APSP} we show how approximate decremental SSSP helps in solving approximate decremental APSP.
Finally, in \Cref{sec:putting together for faster algorithms}, we show how to put the results in \Cref{sec:emulator,sec:monotone ES tree,sec:approximate_SSSP_to_approximate_APSP} together to obtain the desired $(1+\epsilon, 2)$- and $(2+\epsilon, 0)$-approximate decremental APSP data structures. 

\begin{definition}[Persevering path]\label{def:decremental path}
Let $ \cH = (H_i, w_i)_{0 \leq i \leq k} $ be a dynamic weighted graph.
We say that a path $ \pi = \langle v_0, v_1, \ldots, v_\ell\rangle$ is {\em persevering up to time~$t$} (where $ t \leq k $) if 
for all $0\leq i\leq \ell-1$, 
\[
\forall 0\leq j\leq t:~  (v_i,v_{i+1})\in E(H_j) ~~~~~\mbox{and} ~~~~~ \forall 0\leq j<t:~w_j(v_i,v_{i+1})\leq w_{j+1}(v_i,v_{i+1}).
\]
In other words, edges in $\pi$ always exist in $ \cH $ up to time $t$ and their weights never decrease.
\end{definition}

We now introduce the notion of a {\em locally persevering emulator}.
An $ (\alpha, \beta) $-emulator of a dynamic graph $ \cG = (G_i)_{0 \leq i \leq k} $ is usually another dynamic weighted graph $ \cH = (H_i, w_i)_{0 \leq i \leq k} $ with the same set of nodes as $ \cG $ that preserves the distance of the original dynamic graph; i.e., for all $i \leq k$ and all nodes $x$ and $y$, there is a path $ \pi_{xy} $ in $ H_i $ such that $\dist_{G_i}(x, y)\leq w_i (\pi_{xy}) \leq \alpha \dist_{G_i}(x, y)+\beta.$
The notion of a {\em locally persevering} emulator has another parameter $\tau$. It requires the condition $\dist_{G_i}(x, y)\leq w_i (\pi_{xy}) \leq \alpha \dist_{G_i}(x, y)+\beta$ to hold only when $\dist_{G_i}(x, y)\leq \tau$. More importantly, it puts an additional restriction that the path $ \pi_{xy} $ must be either a shortest path in $G_i$ or a persevering path.

\begin{definition}[Locally persevering emulator]\label{def:locally decremental emulator}
Consider parameters $ \alpha \geq 1 $, $\beta \geq 0 $ and $ \tau \geq 1$, a dynamic graph $ \cG = (G_i)_{0 \leq i \leq k} $, and a dynamic weighted graph $ \cH = (H_i, w_i)_{0 \leq i \leq k} $. 
For every $ i \leq k $, we say that a path $ \pi $ in $ G_i $ is {\em contained} in $ (H_i, w_i) $ if every edge of $ \pi $ is contained in $ H_i $ and has weight $ 1 $.
We say that $\cH$ is an {\em ($\alpha$, $\beta$, $\tau$)-locally persevering emulator of $ \cG $} if for all nodes $ x $ and $ y $ we have 
\begin{enumerate}[label=(\arabic{*})]
\item \label{item:emulator one} $ \dist_{G_i} (x, y) \leq \dist_{H_i, w_i} (x, y) $ for all $ 0 \leq i \leq k $, and 
\item \label{item:emulator two} there are $ t_1 $ and $ t_2 $ with $ 0 \leq t_1 < t_2 \leq k+1 $ such that the following hold:
\begin{enumerate}
\item There is a path $ \pi $ from $ x $ to $ y $ in $ \cH $ that is persevering (at least) up to time $ t_1 $ and satisfies $ w_t (\pi) \leq \alpha \dist_{G_t} (x, y) + \beta $.\label{item:emulator a}
\item For every $ t_1 < i \leq t_2 $, a shortest path from $ x $ to $ y $ in $ G_i $ is contained in $ (H_i, w_i) $. \label{item:emulator b}
\item For every $ i \geq t_2 $, $ \dist_{G_i} (x, y) > \tau $.\label{item:emulator c}
\end{enumerate}
\end{enumerate}
\end{definition}

Condition \ref{item:emulator one} simply says that $\cH$ does not underestimate the distances in $\cG$. 
Condition \ref{item:emulator two} says that the distance between $x$ and $y$ must be preserved in $\cH$ in the following specific way: In the beginning (see \ref{item:emulator a}), it must be approximately preserved by a single path $\pi$ (thus $\pi$ is a persevering path). Whenever $\pi$ disappears, the shortest path between $x$ and $y$ must appear in $\cH$ (see \ref{item:emulator b}). However, we can remove all these conditions whenever the distance between $x$ and $y$ is greater than $\tau$  (see \ref{item:emulator c}).

\subsection{$(1, 2, \lceil 2 / \epsilon \rceil )$-Locally Persevering Emulator of Size $\tilde O (n^{3/2}) $}\label{sec:emulator}\label{sec:locally decremental emulator}

In the following we present the locally persevering emulator that we will use to achieve a total update time of $\tilde O(n^{5/2} / \epsilon^2)$  for decremental approximate APSP. Roughly speaking, we can replace the running time of $\tilde O(m n / \epsilon)$ by $\tilde O(n^{5/2} / \epsilon^2)$ because this emulator always has $\tilde O(n^{3/2})$ edges. However, to be technically correct, we have to use the stronger fact that the number of edges ever contained in the emulator is $\tilde O(n^{3/2})$, as in the following statement.

\begin{lemma}[\textrm{Existence of $(1, 2, \lceil 2 / \epsilon \rceil)$-locally persevering emulator with  $\tilde O (n^{3/2}) $ edges}]\label{thm:emulator}
For every $ 0 < \epsilon \leq 1 $ and every decremental graph $ \cG = (G_i)_{0 \leq i \leq k} $, there is data structure that maintains a dynamic weighted graph $ \cH = (H_i, w_i)_{0 \leq i \leq k} $ in $ O (m n^{1/2} \log{n} / \epsilon) $ total time such that $ \cH $ is a $ (1, 2, \lceil 2 / \epsilon \rceil) $-locally persevering emulator with high probability. 
Moreover, the number of edges ever contained in the emulator is $ | E_k (\cH) | = O (n^{3/2} \log{n}) $, and the total number of updates in $ \cH $ is $ \phi_k (\cH) = O (n^{3/2} \log{n} / \epsilon) $.
\end{lemma}

We construct a dynamic weighted graph $ \cH = (H_i, w_i)_{0 \leq i \leq k} $ as follows.
Pick a set $ \dom $ of nodes by including every node independently with probability $ (a \ln n) / \sqrt{n} $ for a large enough constant $ a $.
Note that the size of $ \dom $ is $ O (\sqrt{n} \log n) $ in expectation.
It is well known that by this type of sampling every node with degree more than $ \sqrt{n} $ has a neighbor in $ \dom $ with high probability (see, e.g., \cite{UllmanY91, DorHZ00}); i.e., $ \dom $ \emph{dominates} all high-degree nodes.
This is even true for every version $ G_i $ of a decremental graph $ \cG = (G_i)_{0 \leq i \leq k} $.
For every $ 0 \leq i \leq k $, we define that the graph $ H_i $ contains the following two types of edges.
For every node $ x \in \dom $ and every node $ y $ such that $ \dist_{G_i} (x, y) \leq \lceil 2 / \epsilon \rceil + 1 $, $ H_i $ contains an edge $ (x, y) $ of weight $ \dist_{G_i} (x, y) $.
For every node~$ x $ such that $ \deg_{G_i} (x) \leq \sqrt{n} $, $ H_i $ contains every edge $ (x, y) $ of $ G_i $.\footnote{Our construction is very similar to the classic $ (1,2) $-emulator given by Dor, Halperin, and Zwick~\cite{DorHZ00}. The main difference is that we can only insert edges of limited weight into the emulator; in particular, we only have edges of weight $O(1/\epsilon)$ in the emulator. One reason for this choice is that it is not known whether the $ (1,2) $-emulator of Dor et al.~can be maintained in $\tilde O(mn)$ time under edge deletions.}
Note that as edges are deleted from $ \cG $, distances between nodes might increase, which in turn increases the weights of the corresponding edges in $ \cH $.
When the distance between $ x $ and $ y $ in $ \cG $ exceeds $ \lceil 2 / \epsilon \rceil + 1 $, the edge $ (x, y) $ is deleted from $ \cH $.

In the following we prove \Cref{thm:emulator} by arguing that the dynamic graph $ \cH $ described above has the following four desired properties:
\begin{itemize}
\item $ \cH $ is a $ (1, 2, \lceil 2 / \epsilon \rceil) $-locally persevering emulator of $ \cG $.
\item The expected number of edges ever contained in the emulator is $ | E_k (\cH) | = O (n^{3/2} \log{n}) $.
\item The expected total number of updates in $ \cH $ is $ \phi_k (\cH) = O (n^{3/2} \log{n} / \epsilon) $.
\item The edges of $ \cH $ can be maintained in expected total time $ O (m n^{1/2} \log{n} / \epsilon) $ for $ k $ deletions in $ \cG $.
\end{itemize}
The last item refers to the time needed to determine, after every deletion in $ \cG $, which edges are contained in $ \cH $ and what their weights are.

\begin{lemma}[Locally persevering]
The dynamic graph $ \cH $ described above is a $ (1, 2, \lceil 2 / \epsilon \rceil) $-locally persevering emulator of $ \cG $ with high probability.
\end{lemma}

\begin{proof}
Let $ t \leq k $, and let $ x $ and $ y $ be a pair of nodes.
We first argue that $ \dist_{G_t} (x, y) \leq \dist_{H_t, w_t} (x, y) $.
It is clear from the construction of $ (H_t, w_t) $ that every edge in $ (H_t, w_t) $ corresponds to an edge in $ G_t $ or to a path in $ G_t $.
Therefore no path in $ (H_t, w_t) $ from $ x $ to $ y $ can be shorter than the distance $ \dist_{G_t} (x, y) $ of $ x $ to $ y $ in $ G_t $.

We now argue that $ \cH $ fulfills the second part of the definition of a $ (1, 2, \lceil 2 / \epsilon \rceil) $-locally persevering emulator of $ \cG $.
Assume that $ \dist_{G_t} (x, y) \leq \lceil 2 / \epsilon \rceil $ and that no shortest path from $ x $ to $ y $ in $ G_t $ is also contained in $ (H_t, w_t) $.
Let $ \pi $ be an arbitrary shortest path from $ x $ to $ y $ in $ G_t $.
Since $ \pi $ is not contained in $ H_t $, there must be some edge $ (u, v) $ on $ \pi $ such that $ (u, v) \notin E (H_t) $.
This can only happen if $ u $ has degree more than $ \sqrt{n} $ in $ G_t $.
With high probability $ u $ has a neighbor $ z \in D $ in $ G_t $ (see, e.g., \cite{UllmanY91, DorHZ00}).
Now consider any $ i \leq t $.
Note that $ \dist_{G_t} (x, u) \leq \dist_{G_t} (x, y) \leq \lceil 2 / \epsilon \rceil  $ and that $ \dist_{G_i} (x, z) \leq \dist_{G_t} (x, z)  $ because distances never decrease in a decremental graph.
By the triangle inequality we get
\begin{equation*}
\dist_{G_i} (x, z) \leq \dist_{G_t} (x, z) \leq \dist_{G_t} (x, u) + \dist_{G_t} (u, z) \leq \lceil 2 / \epsilon \rceil + 1 \, .
\end{equation*}
Therefore, for every $ i \leq t $, $ H_i $ contains an edge $ (x, z) $ of weight $ w_i (x, z) = \dist_{G_i} (x, z) $, which means that the edge $ (x, z) $ is persevering up to time~$ t $.
The same argument shows that $ H_i $ also contains an edge $ (z, y) $ of weight $ w_i (z, y) = \dist_{G_i} (z, y) $ for every $ i \leq t $; i.e., $ (z, y) $ is also persevering up to time~$ t $.
Now consider the path $ \pi' $ in $ H_t $ consisting of the edges $ (x, z) $ and $ (z, y) $.
Since both edges are persevering up to time~$ t $, also the path $ \pi' $ is persevering up to time~$ t $.
Furthermore, $ \pi' $ guarantees the desired approximation:
\begin{align*}
w_t (\pi') = w_t (x, z) + w_t (z, y) &= \dist_{G_t} (x, z) + \dist_{G_t} (z, y) \\
 &\leq \dist_{G_t} (x, u) + \dist_{G_t} (u, z) + \dist_{G_t} (z, u) + \dist_{G_t} (u, y) \\
 &= \dist_{G_t} (x, u) + \dist_{G_t} (u, y) + 2 \\
 &= \dist_{G_t} (x, y) + 2 \, .
\end{align*}
To explain the last equation, remember that $ u $ lies on a shortest path from $ x $ to~$ y $ and therefore $ \dist_{G_t} (x, y) = \dist_{G_t} (x, u) + \dist_{G_t} (u, y) $.
Thus, $ \cH $ is a $ (1, 2, \lceil 2 / \epsilon \rceil) $-locally persevering emulator of $ \cG $.
\end{proof}

\begin{lemma}[Number of edges]\label{lem:emulator_number_of_edges}
The number of edges ever contained in the dynamic graph $ \cH $ is $ | E_k (\cH) | = O (n^{3/2} \log{n}) $ in expectation.
\end{lemma}

\begin{proof}
Every edge in $ H_i $ either was inserted at some time or is an edge that is also contained in $ H_0 $.
Thus, it is sufficient to bound the number of inserted edges and the number of edges in $ H_0 $ by $ O (n^{3/2} \log{n}) $.

We first show that the number of edges in $ H_0 $ is $ | E (H_0) | = O (n^{3/2} \log{n}) $.
We can charge each edge in $ H_0 $ either to a node in $ \dom $ or to a node with degree at most $ \sqrt{n} $.
For every node $ x \in \dom $ there might be $ O(n) $ edges adjacent to $ x $ in $ H_0 $.
Since there are $ O (\sqrt{n} \log{n}) $ many nodes in $ \dom $, the number of edges charged to these nodes is $ O (n^{3/2} \log{n}) $.
For nodes with degree at most $ \sqrt{n} $ there are $ O(\sqrt{n}) $ edges adjacent to $ x $ in $ H_0 $.
Since there are $ O (n) $ such nodes, the number of edges charged to these nodes is $ O(n^{3/2}) $.
In total, we get
\begin{equation*}
| E (H_0) | = O (n^{3/2} \log{n}) + O(n^{3/2}) = O (n^{3/2} \log{n}) \, .
\end{equation*}

We now show that the number of edges inserted into $ \cH $ over all deletions in $ \cG $ is $ O(n^{3/2}) $.
Every time the degree of a node $ x $ changes from $ \deg_{G_i} (x) > \sqrt{n} $ to $ \deg_{G_{i+1}} (x) = \sqrt{n} $ (for some $ 0 \leq i < k $) we insert all $ \sqrt{n} $ edges adjacent to $ x $ in $ G_{i+1} $ into $ H_{i+1} $.
In a decremental graph it can happen at most once for every node that the degree of a node drops to $ \sqrt{n} $.
Therefore at most $ n^{3/2} $ edges are inserted in total.
\end{proof}

\begin{lemma}[Number of updates]\label{lem:number_of_updates}
The total number of updates in the dynamic graph $ \cH $ described above is $ \phi_k (\cH) = O (n^{3/2} \log{n} / \epsilon) $ in expectation.
\end{lemma}

\begin{proof}
Only the following kinds of updates appear in $ \cH $: edge insertions, edge deletions, and edge weight increases.
Every edge that is inserted or deleted has to be contained in $ \cH $ at some time.
Thus, we can bound the number of insertions and deletions by $ | E_k (\cH) | $, which is $ O (n^{3/2} \log{n}) $ by \Cref{lem:emulator_number_of_edges}

It remains to bound the number of edge weight increases by $ O (n^{3/2} \log{n} / \epsilon) $.
All weighted edges are incident to at least one node in $ \dom $.
The maximum weight of these edges is $ \lceil 2 / \epsilon \rceil + 1 $ and the minimum weight is $ 1 $.
As all edge weights are integer, the weight of such an edge can increase at most $ \lceil 2 / \epsilon \rceil + 1 $ times.
As there are $ O (\sqrt{n} \log{n}) $ nodes in $ \dom $, each having $ O(n) $ weighted edges, the total number of edge weight increases is $ O (n^{3/2} \log{n} / \epsilon) $.
\end{proof}

\begin{lemma}[Running time]
The edges of the dynamic graph $ \cH $ described above can be maintained in expected total time $ O (m n^{1/2} \log{n} / \epsilon) $ over all $ k $ edge deletions in~$ \cG $.
\end{lemma}

\begin{proof}
We use the following data structures:
(A) For every node, we maintain its incident edges in $ \cH $ with a dynamic dictionary using dynamic perfect hashing~\cite{DietzfelbingerKMHRT94} or cuckoo hashing~\cite{PaghR04}.
This graph representation allows us to perform insertions and deletions of edges as well as edge weight increases.
(B) For every node $ x $, we maintain the degree of $ x $ in~$ \cG $.
(C) For every node $ x \in \dom $ we maintain a (classic) ES-tree (see \Cref{sec:ES tree}) rooted at $ x $ up to distance $ \lceil 2 / \epsilon \rceil + 1 $.

We now explain how to process the $i$-the edge deletion in $ \cG $ of, say, the edge $ (u, v) $.
First, we update (B) by decreasing the number that stores the degree of $ u $ in $ \cG $ and then do the same for $ v $.
If the degree of $ u $ (or $ v $) drops to $ \sqrt{n} $, we insert all edges incident to $ u $ (or $ v $) in $ G_i $ into $ H_i $ in (A).
After this procedure, for every node $ x \in D $, we do the following to update (C):
First of all, we report the deletion of $ (u, v) $ to the ES-tree rooted at $ x $.
Every node $ y $ has a level in this ES-tree.
If the level of $ y $ increases to $ \infty $, then $ \dist_{G_i} (x, y) > \lceil 2 / \epsilon \rceil + 1 $ and therefore we remove the edge $ (x, y) $ from $ \cH $ in (A).
If the level of $ y $ increases, but does not reach $ \infty $, then $ \dist_{G_i} (x, y) \leq \lceil 2 / \epsilon \rceil + 1 $ and we update the weight of the edge $ (x, y) $ in (A).

We can perform each deletion, insertion, and edge weight increase in expected amortized constant time.
As there are $ O (n^{3/2} \log{n} / \epsilon) $ updates in $ \cH $ in expectation by \Cref{lem:number_of_updates}, the expected total time for maintaining (A) is $ O (n^{3/2} \log{n} / \epsilon) $.
We need constant time per deletion in $ \cG $ to update (B) and thus time $ O(m) $ in total.
Maintaining the ES-tree takes total time $ O (m/\epsilon) $ for each node in $ \dom $ (see \Cref{sec:ES tree}).
Since in expectation there are $ O(n^{1/2} \log{n}) $ nodes in $ \dom $, the expected total time for maintaining (A) is $ O (m n^{1/2} \log{n} / \epsilon) $ in total.

Thus, the expected total update time for maintaining $ \cH $ under deletions in $ \cG $ is $ O(n^{3/2} \log{n} / \epsilon + m + m n^{1/2} \log{n} / \epsilon) $, which is $ O (m n^{1/2} \log{n} / \epsilon) $.
\end{proof}

\subsection{Maintaining Distances Using Monotone Even--Shiloach Tree}\label{sec:monotone ES tree}

In this section, we show how to use a locally persevering emulator to maintain the distances from a specific node $ \treeroot $ (called {\em root}) to all other nodes, up to distance $\Q$, for some parameter $\Q$. The hope of using an emulator is that the total update time will be smaller since an emulator has a smaller number of edges. In particular, recall that if we run an ES-tree on an input graph, the total update time is $\tilde O(m\Q)$. Now consider running an ES-tree on an emulator $\cH$ instead; we might hope to get a running time of $\tilde O(m'\Q)$, where $m'$ is the number of edges ever appearing in $\cH$. This is beneficial when $m'\ll m$ (for example, the emulator we construct in the previous section has $m'=\tilde O(n^{1.5})$, which is less than $m$ when the input graph is dense). The main result of this section is that we can achieve exactly this when $\cH$ is a locally-persevering emulator, and we run a variant of the ES-tree called the {\em monotone ES-tree} on $\cH$.

\begin{lemma}[Monotone ES-tree + Locally Persevering Emulator]\label{thm:monotone ES tree}
For every distance range parameter $ \Q $, every source node $ \treeroot $, and every decremental graph $ \cG = (G_i)_{0 \leq i \leq k} $ with an $ (\alpha, \beta, \tau) $-locally persevering emulator $ \cH = (H_i, w_i)_{0 \leq i \leq k} $, the monotone ES-tree on $ \cH $ is an $ (\alpha + \beta / \tau, \beta) $-approximate decremental SSSP data structure for $ \cG $.
It has constant query time and a total update time of 
\begin{equation*}
O(\phi_k(\cH) + |E_k (\cH) | \cdot ((\alpha + \beta/\tau) \Q + \beta) ) \, ,
\end{equation*}
where $ \phi_k(\cH) $ is the total number of updates in $ \cH $ up to time $ k $ and $ E_k (\cH) $ is the set of all edges ever contained in $ \cH $ up to time $ k $.
\end{lemma}

Note that we need to modify the ES-tree because although the input graph undergoes only edge deletions, the emulator might have to undergo some edge {\em insertions}. If we straightforwardly extend the ES-tree to handle insertions, we will have to keep the level of any node $y$ at $ \lev (y) = \min_{z} (\lev (z) + w (y, z)) $ as in Line~\ref{line: update ell(y)} of \Cref{alg:ES_tree}.
This might cause the level $\ell(y)$ of some node $y$ in the ES-tree to {\em decrease}. This {\em destroys the monotonicity} of levels of nodes, which is the key to guaranteeing the running time of the ES-tree, as shown in \Cref{sec:ES tree}. The monotone ES-tree is a variant that insists on keeping the nodes' levels monotone (thus the name); it never decreases the level of any node.

\paragraph*{Implementation of Monotone ES-Tree.}
Our monotone ES-tree data structure is a modification of the ES-tree, which always maintains the level $ \ell (x) $ of every node $ x $ in a shortest paths tree rooted at $ \treeroot $ up to depth $ \Q $, as presented in \Cref{sec:ES tree}. \Cref{alg:monotone_ES_tree} shows the pseudocode of the monotone ES-tree.
Our modification can deal with edge insertions, but does this in a \emph{monotone} manner: it will never decrease the level of any node.
In doing so, it will lose the property of providing a shortest paths tree of the underlying dynamic graph, which in our case is the emulator $ \cH $.
However, due to special properties of the emulator, we can still guarantee that the level provided by the monotone ES-tree is an approximation of the distance in the original decremental graph $ \cG $.
The distance estimate provided for a node $ x $ is the level of $ x $ in the monotone ES-tree.

The overall algorithm now is as follows (see \Cref{alg:monotone_ES_tree} for details).
We initialize the monotone ES-tree by computing a shortest paths tree in the emulator $ H_0 $ up to depth $ (\alpha+\beta/\tau) \Q + \beta $.
For every node $ x $ in this tree we set $ \ell (x) = \dist_{H_0} (x, \treeroot) $, and for every other node $ x $ we set $ \ell (x) = \infty $.
Starting with these levels, we maintain an ES-tree rooted at $ \treeroot $ up to depth $ (\alpha+\beta/\tau) \Q + \beta $ on the graph $ \cH $.
This ES-tree alone cannot deal with edge insertions and edge weight increases.
Our additional procedure that is called after the insertion of an edge $ (u, v) $ only updates the value of $ v $ in the heap $ N (u) $ of $ u $ to $ \ell (v) + w (u, v) $.
In particular, the level of $ u $ is not changed after such an insertion.

\begin{algorithm2e}
\caption{Monotone ES-tree}
\label{alg:monotone_ES_tree}

\tcp{\textrm{The algorithm is like the usual ES-tree (\Cref{alg:ES_tree}) with three modifications:
\begin{enumerate}
\item The algorithm runs on $ \cH = (H_0, H_1, \ldots)$ instead of $ \cG $.
\item The depth of the tree is $ (\alpha+\beta/\tau) \Q + \beta $ instead of $ \Q $
\item There are additional procedures for the insertion of edges and edge weight increases.
\end{enumerate}}}
\vspace{-3ex}
\tcp{\textrm{Procedures \Delete{} and \Increase{} are the same as before.}}
\tcp{\textrm{Line numbers in the form $i$* indicate lines that are different from \Cref{alg:ES_tree}. Blue color marks the changes.}} 

\BlankLine

\Procedure{\Initialize{}}{
	\SetNlSty{textbf}{}{*}
    Compute shortest paths tree from $ \treeroot $ in $ ({\color{blue!80} H_0}, w_0) $ up to depth ${\color{blue!80} (\alpha+\beta/\tau) \Q + \beta }$\; 
    \SetNlSty{textbf}{}{}
	\ForEach{node $ u $}{
		\SetNlSty{textbf}{}{*}
		Set $ \lev (u) = \dist_{{\color{blue!80}H_0}} (u, \treeroot) $\;\label{line:initialize_level}
		\SetNlSty{textbf}{}{}
		\lFor{every edge $ (u, v) $}{
			insert $ v $ into heap $ N(u) $ of $ u $ with key $ \lev(v) + w(u, v) $
		}
	}
}

\BlankLine

\Procedure{\Insert{$u$, $v$, $w(u, v)$}}{
	Insert $ v$ into heap $ N(u) $ with key $ \ell (v) + w (u, v) $ and $u$ into heap $N(v)$ with key $ \ell (u) + w (u, v) $\; 
}

\BlankLine

\Procedure{\UpdateLevels{}}{
	\While{heap $ Q $ is not empty}{
	    Take node $ y $ with minimum key $ \lev (y) $ from heap $ Q $ and remove it from $ Q $\;
		$ \lev'(y) \gets \min_{z} (\lev (z) + w (y, z)) $\;
		\label{line:retrieve_min_monotone_ES_tree} \label{line: update ell(y) modified}
		
		\tcp{\textrm{$\min_{z} (\lev (z) + w (y, z)) $ can be retrieved from the heap $ N(y) $. $\arg\min_{z} (\lev (z) + w (y, z)) $ is $y$'s parent in the ES-tree.}}
	
		\If{$ \lev'(y) > \lev (y) $}{\label{line:check_for_level_increase}
			$\lev(y)\gets \lev'(y)$\;\label{line:level_increase}
			\SetNlSty{textbf}{}{*}
			\lIf{$ \lev' (y) > {\color{blue!80}(\alpha+\beta/\tau) \Q + \beta } $}{
				$ \lev (y) \gets \infty $
			}
			\SetNlSty{textbf}{}{}

		\ForEach{neighbor $ x $ of $ y $}{
				update key of $ y $ in heap $ N(x) $ to $ \lev(y) + w(x, y) $\;
				insert $ x $ into heap $ Q $ with key $ \lev(x) $ if $Q$ does not already contain~$ x $
			}
		}
	}
}
\end{algorithm2e}

\paragraph*{Order of Updates.}
Before we start analyzing our algorithm we clarify a crucial detail about the order of updates in the locally persevering emulator $ \cH $.
Consider an edge deletion in the graph $ G_i $ that results in the graph $ G_{i+1} $.
In the emulator $ \cH $, it might be the case that several updates are necessary to obtain $ (H_{i+1}, w_{i+1}) $ from $ (H_i, w_i) $.
There could be several insertions, edge weight increases, and edge deletions at once.\footnote{We could also allow edge weight decreases and handle them in exactly the same way as edge insertions. For simplicity, we omit this case from our description.}
Our algorithm will process these updates (using the monotone ES-tree) in a specific order: First, it processes the edge insertions, one after the other. Afterward, we process the edge deletions and edge weight increases (also one after the other).
This order is crucial for the correctness of our algorithm.

\paragraph*{Analysis.}
We first argue about the correctness of the monotone ES-tree and afterward argue about its running time.
In the following we let $ \ell_i (u) $ be the level of $ u $ in the monotone ES-tree after it has processed the $i$-th edge deletion in $ \cG $ (which could mean that it has processed a whole series of insertions, weight increases, and deletions of the emulator $ \cH $). Remember that $ (H_i, w_i) $ denotes the emulator after all updates caused by the $i$-th deletion in $ \cG $.
We say that an edge $ (u, v) $ is \emph{stretched} if $ \ell_i (u) \neq \infty $ and $ \ell_i (u) > \ell_i (v) + w_i (u, v) $.
We say that a node $ u $ is \emph{stretched} if it is incident to an edge $ (u, v) $ that is stretched.
Note that for a node $ u $ that is not stretched we either have $ \ell_i (u) = \infty $ or $ \ell_i (u) \leq \ell_i (v) + w_i (u, v) $ for every edge $ (u, v) \in E (H_i) $.
Our analysis uses four simple observations about the algorithm. (Recall that a tree edge is an edge between any node $y$ and its parent as in Line~\ref{line: update ell(y) modified} of \Cref{alg:monotone_ES_tree}; i.e., it is an edge $(y, z')$ for some node $z'=\arg\min_z (\ell(z)+w(y, z))$.)
\begin{observation}\label{obs:simple_observations_monotone_ES}
The following holds for the monotone ES-tree:
\begin{enumerate}[label=(\arabic{*})]
\item \label{item: observation one} The level of a node never decreases.
\item  \label{item: observation two} An edge can only become stretched when it is inserted.
\item  \label{item: observation three} As long as a node $ x $ is stretched, its level does not change.
\item  \label{item: observation four} For every tree edge $ (u, v) $ (where $ v $ is the parent of $ u $), $ \ell (u) \geq \ell (v) + w (u, v) $.
\end{enumerate}
\end{observation}

\begin{proof}
The only places in the algorithm where the level of a node is modified are in Line~\ref{line:initialize_level} during the initialization and in Line~\ref{line:level_increase}.
The if-condition in Line~\ref{line:check_for_level_increase} guarantees that the level of a node never decreases and thus~\ref{item: observation one} holds.
Furthermore, whenever the level of a node $ y $ increases in Line~\ref{line:level_increase} we have $ \lev (y) = \min_{z} (\lev (z) + w (y, z)) \leq \lev (z') + w (y, z') $ for every neighbor $ z' $ of $ y $.
Thus, after such a level increase the edge $ (y, z') $ is nonstretched for every neighbor $ z' $ of $ y $, and so is the node $ y $.

To prove~\ref{item: observation two}, consider an edge $ (x, y) $ that becomes stretched.
This can only happen if the edge $ (x, y) $ was not contained in the graph before and is inserted or if the edge changes from nonstretched to stretched.
When $ (x, y) $ is nonstretched we have $ \lev (x) \leq \lev (y) + w (x, y) $.
For $ (x, y) $ to become stretched (i.e., for $ \lev (x) > \lev (y) + w (x, y) $ to hold) either the left-hand side of this inequality has to increase or the right-hand side has to decrease.
When the left-hand side increases, the level of $ x $ changes and, as argued above, this implies that $ (x, y) $ will be nonstretched.
As the level of $ y $ is nondecreasing, the right-hand side can only decrease when the weight of the edge $ (x, y) $ decreases.
This can only happen after inserting this edge with a smaller weight.

We now prove~\ref{item: observation three}.
Consider a node $ x $ that is stretched.
As long as it is stretched, the level of $ x $ does not increase because, as argued above, each level increase immediately makes $ x $ nonstretched.

Finally, we prove~\ref{item: observation four}.
Consider an edge $ (u, v) $ such that $ v $ is the parent of $ u $.
It is easy to see that $ \lev (u) \geq \lev (v) + w (u, v) $ as long as $ v $ stays the parent of $ u $ because the level of $ u $ increases if and only if the level of $ v $ or the weight of the edge $ (u, v) $ increases.
In such a case we have $ \lev (u) = \lev (v) + w (u, v) $.
The only other possibility for the right-hand side of the inequality to change is when the weight of the edge $ (u, v) $ decreases, which can happen after an insertion.
But decreasing this value does not invalidate the inequality.
\end{proof}

We now prove that the monotone ES-tree provides an $ (\alpha + \beta / \tau, \beta) $-approximation of the true distance if it runs on an $ (\alpha, \beta, \tau) $-locally persevering emulator.
We use an inductive argument to show that, after having processed the $i$-th deletion of an edge in~$ \cG $, the level of every node $ x $ is a $ (\alpha + \beta / \tau, \beta) $-approximation of the distance of $ x $ to the root, i.e., $\dist_{G_i}(x, r)\leq \ell_i(x)\leq (\alpha + \beta / \tau) \dist_{G_i}(x, r) + \beta$.
The intuition why this should be correct is as follows:
If the monotone ES-tree gives the desired approximation before a deletion in $ \cG $ and the deletion does not cause an edge in $ \cH $ to become stretched, then the structure of the monotone ES-tree is similar to the ES-tree, and the same argument that we use for the ES-tree should show that the monotone ES-tree still gives the desired approximation.
If, however, an edge becomes stretched in $ \cH $, then the level of the affected node does not change anymore and, thus, as distances in decremental graphs never decrease, we should still obtain the desired approximation.
This intuition is basically correct, but the correctness proof also requires the emulator to be persevering, as a persevering path does not contain any inserted edge and, thus, no stretched edges.

Remember that processing an edge deletion in $ \cG $ might mean processing a series of updates in $ \cH $.
We will first show that the approximation guarantee holds for every node that is \emph{stretched} after the monotone ES-tree has processed the $i$-th deletion.
Afterward we will show that it holds for \emph{every} node.

\begin{lemma}\label{lem:monotone_ES_correctness_stretched_nodes}
Let $ 0 < i \leq k $ and assume that $ \ell_{i-1} (x') \leq (\alpha + \beta / \tau) \cdot \dist_{G_{i-1}} (x', \treeroot) + \beta $ for every node $ x' $ with $ \ell_{i-1} (x') \neq \infty $.
Then $ \ell_i (x) \leq (\alpha + \beta / \tau) \cdot \dist_{G_i} (x, \treeroot) + \beta $ for every stretched node $ x $.
\end{lemma}

\begin{proof}
Here we need the assumption that the monotone ES-tree sees the updates in the emulator caused by a single edge deletion in a specific order, namely such that all edge insertions can be processed before the edge weight increases and edge deletions.
Since $ x $ is stretched, there must have been a previous insertion of an edge $ (x, y) $ incident to $ x $ such that $ x $ is stretched since the time this edge was inserted (see \Cref{obs:simple_observations_monotone_ES}\ref{item: observation two}).
Let $ \ell'(x) $ denote the level of $ x $ after the insertion of $ (x, y) $ has been processed.
By \Cref{obs:simple_observations_monotone_ES}, nodes do not change their level as long as they are stretched, and therefore $ \ell_i (x) = \ell' (x) $.

We now show that $\ell_i(x)=\ell'(x)=\ell_{i-1}(x)$.
The insertion of $(x, y)$ could either happen at time $i$ or at some earlier time  (i.e., either it was caused by the $i$-th edge deletion or by a previous edge deletion). 
If the insertion was caused by a previous edge deletion, we clearly have $ \ell_{i-1} (x) = \ell'(x) $ because the level of $ x $ has not changed since this insertion.
Consider now the case that the insertion was caused by the $i$-th edge deletion.
Recall that all insertions caused by the $i$-th deletion are processed \emph{before} any other updates of the emulator are processed.
Since edge insertions do not change the level of any node, we have $ \ell' (x) = \ell_{i-1} (x) $.
In both cases we have $ \ell' (x) = \ell_{i-1} (x) $ and thus $ \ell_i (x) = \ell_{i-1} (x) $.
Since $ \ell_i (x) \neq \infty $, we have $ \ell_{i-1} (x) \neq \infty $.
It follows that 
\begin{equation*}
\ell_i (x) = \ell_{i-1} (x) \leq (\alpha + \beta / \tau) \cdot \dist_{G_{i-1}} (x, \treeroot) + \beta \leq (\alpha + \beta / \tau) \cdot \dist_{G_{i}} (x, \treeroot) + \beta 
\end{equation*}
as desired. The first inequality above follows from the assumptions of the lemma, and the second one is because $ \cG $ is a decremental graph in which distances never decrease.
\end{proof}

In order to prove the approximation guarantee for nonstretched nodes, we have to exploit the properties of the $ (\alpha, \beta, \tau) $-locally persevering emulator $ \cH $.
In the classic ES-tree the level of two nodes differs by at most the weight of a path connecting them---modulo some technical conditions that arise for ES-trees of limited depth.
In the monotone ES-tree this is only true for persevering paths (see \Cref{lem:level_difference_decremental_path}).
Before we can show this we need an even simpler property of the monotone ES-tree:
If two nodes are connected by an edge that is not stretched, then their levels differ by at most the weight of the edge connecting them.
Again, in the classic ES-tree this holds for any edge.

\begin{lemma}\label{lem:level_weight_inequality}
Consider any $ 0 \leq i \leq k $ and any $ (x, y) \in E (H_{i}) $. We have
\begin{equation*}
\ell_i (x) \leq \ell_i (y) + w_i (x, y)
\end{equation*}
if $ \ell_i (y) + w_i (x, y) \leq (\alpha + \beta / \tau) \Q + \beta $ and either (a) $ i = 0 $ or (b) $ i \geq 1 $, $ \ell_{i-1} (x) \neq \infty $ and $ (x, y) $ is not stretched.
\end{lemma}

\begin{proof}
Note that no edge in $ H_0 $ is stretched.
Thus, $ (x, y) $ is not stretched for $ i \geq 0 $.
Hence, we either have $ \ell_i (x) \leq \ell_i (y) + w_i (x, y) $, as desired, or $ \ell_i (x) = \infty $.
Thus, we only have to argue that $ \ell_i (x) \neq \infty $.

Assume by contradiction that $ \ell_i (x) = \infty $.
As $ \ell_{i-1} (x) \neq \infty $, the level of $ x $ is not changed while the monotone ES-tree processes the insertions in $ \cH $ caused by the $i$-th deletion in $ \cG $.
Thus, the only possibility for the level to be increased to $ \infty $ is when the monotone ES-tree processes the edge deletions and edge weight increases.
For every node $ v $, let $ \ell' (v) $ and $ w' (u, v) $ denote the level of every node $ v $ and the weight of every edge $ (u, v) $ directly after the level of $ x $ has been increased to $ \infty $.
Since $ \ell' (x) = \infty $ it must be the case that $ \min_z (\ell' (z) + w'(x, z)) > (\alpha + \beta / \tau) \Q + \beta $, and therefore also $ \ell' (y) + w'(x, y) > (\alpha + \beta / \tau) \Q + \beta $.
But since levels and edge weights never decrease, we also have $ \ell' (y) + w' (x, y) \leq \ell_i (y) + w_i (x, y) \leq (\alpha + \beta / \tau) \Q + \beta $, which contradicts the inequality we just derived.
Therefore it cannot be the case that $ \ell' (x) = \infty $.
\end{proof}

\begin{lemma}\label{lem:level_difference_decremental_path}
For every path $ \pi $ from a node $ x $ to a node $ z $ that (1) is persevering up to time~$ i $ and (2) has the property that $ \ell_i (z) + w_i (\pi) \leq (\alpha + \beta / \tau) \Q + \beta $, we have $ \ell_i (x) \leq \ell_i (z) + w_i (\pi) $.
\end{lemma}

\begin{proof}
The proof is by induction on $ i $ and the length of the path $ \pi $.
The claim is clearly true if $ i = 0 $ or the path has length $ 0 $.
Consider now the induction step.
Let $ (x, y) $ denote the first edge on the path.
Let $ \pi' $ denote the subpath of $ \pi $ from $ y $ to $ z $.
Note that $ \ell_i (z) + w_i (\pi') \leq \ell_i (z) + w_i (\pi) \leq (\alpha + \beta / \tau) \Q + \beta $.
Therefore we may apply the induction hypothesis on $ y $ and get that $ \ell_i (y) \leq \ell_i (z) + w_i (\pi') $.
Thus, we get
\begin{equation*}
\ell_i (y) + w_i (x, y) \leq \ell_i (z) + w_i (\pi') + w (x, y) = \ell_i (z) + w_i (\pi) \leq (\alpha + \beta / \tau) \Q + \beta \, .
\end{equation*}
By the definition of persevering paths, every edge $ (u, v) $ on $ \pi $ has always existed in $ \cH $ since the beginning.
Therefore the edge $ (x, y) $ has never been inserted which means that $ (x, y) $ is not stretched by \Cref{obs:simple_observations_monotone_ES}\ref{item: observation two}.
Since levels and edge weights are nondecreasing we have $ \ell_{i-1} (z) + w_{i-1} (\pi) \leq \ell_i (z) + w_i (\pi) \leq (\alpha + \beta / \tau) \Q + \beta $.
By the induction hypothesis for $ i - 1 $ this implies that $ \ell_{i-1} (x) \leq \ell_{i-1} (z) + w_{i-1} (\pi) \neq \infty $.
We therefore may apply \Cref{lem:level_weight_inequality} and get that $ \ell_i (x) \leq \ell_i (y) + w_i (x, y) \leq \ell_i (z) + w_i (\pi) $.
\end{proof}

Using the property above, we would ideally like to do the following:
Split a shortest path from $ x $ to the root $ \treeroot $ into subpaths of length $ \leq \tau $ and replace each subpath by a persevering path such that the length of each subpath and the persevering path by which it is replaced are approximately the same.
Repeated applications of the inequality of \Cref{lem:level_difference_decremental_path} would then allow us to bound the level of $ x $.
However, this approach alone does not work because the definition of a locally persevering emulator does not always guarantee the existence of a persevering path.
Instead of a persevering path, the locally persevering emulator might also provide us with a shortest path of $ G_i $ that is contained in the current emulator $ H_i $.
In principle this is a nice property because a shortest path is even better than an approximate shortest path.
But the problem now is that nodes on this path could be stretched and only for nonstretched nodes can the difference in levels of two nodes be bounded by the weight of the edge between them.
We can resolve this issue by induction on the distance to $ \treeroot $, which allows us to use the contained path only partially.

\begin{lemma}[Correctness]\label{lem:approximation_ES_tree}
For every node $ x $ and every $ 0 \leq i \leq k $, $ \ell_i (x) \geq \dist_{G_i} (x, \treeroot) $, and if $ \dist_{G_i} (x, \treeroot) \leq \Q $, then $ \ell_i (x) \leq (\alpha + \beta / \tau) \cdot \dist_{G_i} (x, \treeroot) + \beta $.
\end{lemma}

\begin{proof}
We start with a proof of the first inequality, $ \ell_i (x) \geq \dist_{G_i} (x, \treeroot) $.
Consider the (weighted) path $ \pi $ from $ x $ to the root $ \treeroot $ in the monotone ES-tree. Recall that the parent of node $v$ is a node $u=\arg\min_{z} (\lev (z) + w (y, z)) $ as in Line~\ref{line: update ell(y) modified} in \Cref{alg:monotone_ES_tree}.
For every edge $ (u, v) $ on this path, where $ v $ is the parent of $ u $, we have $ \ell_i (u) \geq \ell_i (v) + w_i (u, v) $ by \Cref{obs:simple_observations_monotone_ES}\ref{item: observation four}.
By repeated applications of this inequality for every edge on $ \pi $ we get $ \ell_i (x) \geq w_i (\pi) + \ell_i (\treeroot) = w_i (\pi) $ (since the level of the root $ \treeroot $ is always $ 0 $).
Since $ \pi $ is a path in $ H_i $ we have $ w_i (\pi) \geq \dist_{G_i} (x, \treeroot) $ because a locally persevering emulator never underestimates the true distance by definition.

We now prove the second inequality, $ \ell_i (x) \leq (\alpha + \beta / \tau) \cdot \dist_{G_i} (x, \treeroot) + \beta $ if $ \dist_{G_i} (x) \leq \Q $.
The proof is by induction on $ i $ and the distance of $ x $ to $ \treeroot $ in $ G_i $.

The claim is clearly true if $ x $ is the root node $ \treeroot $ itself.
If $ i \geq 1 $, then note that  $ \dist_{G_{i-1}} (x, \treeroot) \leq \dist_{G_{i}} (x, \treeroot) \leq \Q $, and therefore, by the induction hypothesis for $ i-1 $, we have $ \ell_{i-1} (x) \neq \infty $.
Therefore we may apply \Cref{lem:monotone_ES_correctness_stretched_nodes}, which means that the desired inequality holds if $ x $ is stretched.
Thus, from now on we assume that $ x \neq \treeroot $ and that $ x $ is not stretched.
We distinguish two cases.

\textit{Case 1:} Consider first the case that there is a shortest path from $ x $ to $ \treeroot $ in $ G_i $ such that its first edge $ (x, y) $ is contained in $ (H_i, w_i) $.
Note that $ \dist_{G_i} (y, \treeroot) < \dist_{G_i} (x, \treeroot) $.
Therefore we may apply the induction hypothesis, and by doing so we get $ {\ell_i (y) \leq (\alpha + \beta / \tau) \dist_{G_i} (y, \treeroot) + \beta} $.
We now want to argue that $ \ell_i (x) \leq \ell_i (y) + w_i (x, y) $ by applying \Cref{lem:level_weight_inequality}.
The edge $ (x, y) $ is contained in $ (H_i, w_i) $ with weight $ w_i (x, y) = \dist_{G_i} (x, y) $, and thus
\begin{align*}
\ell_i (y) + w_i (x, y) = \ell_i (y) + \dist_{G_i} (x, y) &\leq (\alpha + \beta / \tau) \cdot \dist_{G_i} (y, \treeroot) + \beta + \dist_{G_i} (x, y) \\
 &\leq (\alpha + \beta / \tau) \cdot (\dist_{G_i} (x, y) + \dist_{G_i} (y, \treeroot)) + \beta \\
 &= (\alpha + \beta / \tau) \cdot \dist_{G_i} (x, \treeroot) + \beta && \text{(1)} \\
 &\leq (\alpha + \beta / \tau) \cdot \Q + \beta  \, . && \text{(2)}
\end{align*}
Remember that $ (x, y) $ is not stretched and if $ i \geq 1 $, then $ \ell_{i-1} (x) \neq \infty $ (as argued above).
Using (2) we may now apply \Cref{lem:level_weight_inequality} and, together with (1), get that
\begin{equation*}
\ell_i (x) \leq \ell_i (y) + w_i (x, y) \leq (\alpha + \beta / \tau) \cdot \dist_{G_i} (x, \treeroot) + \beta
\end{equation*}
as desired.

\textit{Case 2:} Consider now the case that for every shortest path from $ x $ to $ \treeroot $ in $ G_i $ its first edge is not contained in $ (H_i, w_i) $.
Define the node $ z $ as follows.
If $ \dist_{G_i} (x, \treeroot) < \tau $, then $ z = \treeroot $.
If $ \dist_{G_i} (x, \treeroot) \geq \tau $, then $ z $ is a node on a shortest path from $ x $ to $ \treeroot $ in $ G_i $ whose distance to $ x $ is $ \tau $, i.e., $ \dist_{G_i} (x, z) = \tau $ and $ \dist_{G_i} (x, \treeroot) = \dist_{G_i} (x, z) + \dist_{G_i} (z, \treeroot) $.
In both cases there is no shortest path from $ x $ to $ z $ in $ G_i $ that is also contained in $ (H_i, w_i) $ because every shortest path from $ x $ to $ z $ can be extended to a shortest path from $ x $ to $ \treeroot $ in $ G_i $ and $ (H_i, w_i) $ does not contain the first edge of such a path.
Since $ \cH $ is an $ (\alpha, \beta, \tau) $-locally persevering emulator, we know that there is a path $ \pi $ from $ x $ to $ z $ in $ (H_i, w_i) $ that is persevering up to time~$ i $ such that $ w_i (\pi) \leq \alpha \dist_{G_i} (x, z) + \beta $.

If $ z = \treeroot $, we have $ \ell_i (z) = 0 $, and therefore we get
\begin{align*}
\ell_i (z) + w_i (\pi) = w_i (\pi) \leq \alpha \dist_{G_i} (x, z) + \beta &\leq (\alpha + \beta / \tau) \cdot \dist_{G_i} (x, \treeroot) + \beta \\
&\leq (\alpha + \beta / \tau) \cdot \Q + \beta
\end{align*}
as desired.
Consider now the case that $ z \neq \treeroot $.
Since $ \dist_{G_i} (z, \treeroot) < \dist_{G_i} (x, \treeroot) $, we may apply the induction hypothesis on $ z $ and get that $ \ell_i (z) \leq (\alpha + \beta / \tau) \cdot \dist_{G_i} (z, \treeroot) + \beta $.
Together with $ \dist_{G_i} (x, z) = \tau $, we get
\begin{align*}
\ell_i (z) + w_i (\pi) &\leq (\alpha + \beta / \tau) \cdot \dist_{G_i} (z, \treeroot) + \beta + \alpha \dist_{G_i} (x, z) + \beta \\
 &= (\alpha + \beta / \tau) \cdot \dist_{G_i} (z, \treeroot) + \beta + \alpha \dist_{G_i} (x, z) + \beta \cdot \dist_{G_i} (x, z) / \tau \\
 &= (\alpha + \beta / \tau) \cdot \dist_{G_i} (z, \treeroot) + \beta + (\alpha + \beta / \tau) \cdot \dist_{G_i} (x, z) \\
 &= (\alpha + \beta / \tau) \cdot (\dist_{G_i} (x, z) + \dist_{G_i} (z, \treeroot)) + \beta \\
 &= (\alpha + \beta / \tau) \cdot \dist_{G_i} (x, \treeroot) + \beta \\
 &\leq (\alpha + \beta / \tau) \cdot \Q + \beta \, .
\end{align*}
The last equation follows from the definition of $ z $.

In both cases we have $ \ell_i (z) + w_i (\pi) \leq (\alpha + \beta / \tau) \cdot \Q + \beta $.
Since $ \pi $ is persevering up to time~$ i $, we may apply \Cref{lem:level_difference_decremental_path} and get the following approximation guarantee:
\begin{equation*}
\ell_i (x) \leq \ell_i (z) + w_i (\pi) \leq (\alpha + \beta / \tau) \cdot \dist_{G_i} (x, \treeroot) + \beta \, . \qedhere
\end{equation*}
\end{proof}

Finally, we provide the running time analysis.
In principle we use the same charging argument as for the classic ES-tree.
We only have to deal with the fact that the degree of a node might change over time in the dynamic emulator.

\begin{lemma}[Running Time]\label{lem:running_time_monotone_ES_tree_with_log}
For $ k $ deletions in $ \cG $, the monotone ES-tree has a total update time of $ O(\phi_k(\cH) \log{n} + | E_k (\cH) | \cdot ((\alpha + \beta/\tau) \Q + \beta) \log{n} ) $, where $ E_k (\cH) $ is the set of all edges ever contained in $ \cH $ up to time $ k $.
\end{lemma}

\begin{proof}
We first bound the time needed for the initialization.
Using Dijkstra's algorithm, the shortest paths tree can be computed in time $ O ( | E (H_0) | + n \log{n}) $, which is $ O ( | E_k (\cH) | \log{n}) $.

We now bound the time for processing all edge deletions in $ \cG $.
Remember that the monotone ES-tree runs on the emulator $ \cH $.
An edge deletion in $ \cG $ could result in several updates in the emulator $ \cH $.
All of these updates have to be processed by the monotone ES-tree with time $ O (\log{n}) $ per update plus the time needed for running the procedure \UpdateLevels.
Therefore the total update time is $ O (\phi_k (\cH) \log{n}) $, where $ \phi_k (\cH) $ is the total number of updates in $ \cH $, plus the cumulated time for updating the levels in procedure \UpdateLevels.

We now bound the running time of the procedure \UpdateLevels.
Here, the well-known level-increase argument works.
We define the \emph{dynamic degree} of a node $ x $ by $ \deg_{\cH} (x) = | \{ (x, y) \mid (x, y) \in E_k (\cH) \} | $.
Clearly, the dynamic degree never underestimates the current degree of a node in the emulator.
We charge time $ O (\deg_{\cH} (x) \log (x)) $ to every level increase of a node $ x $ and time $ O(\log{n}) $ to every update in $ \cH $.

We now argue that this charging covers all costs in the procedure \UpdateLevels.
Consider a node $ x $ that is processed in the while-loop of the procedure \UpdateLevels after some update in the emulator.
Now the following holds:
If the level of $ x $ increases, the monotone ES-tree has to spend time $ O (\deg_{\cH} (x) \log (x)) $ because $ \deg_{\cH} (x) $ bounds the current degree of $ x $ in the emulator.
If the level of $ x $ does not increase, the monotone ES-tree has to spend time $ O(\log{n}) $.
We now only have to argue that the cost of $ O(\log{n}) $ in the second case is already covered by our charging scheme.

There are two possibilities explaining why $ x $ is in the heap.
The first one is that $ x $ is processed directly after the deletion or weight increase of an edge $ (x, y) $.
The second one is that it was put there by one of its neighbors.
In the first situation we can charge the running time of $ O(\log{n}) $ to the weight increase (or delete) operation.
Consider now the second situation: the level of a node $ y $ increases and its neighbor $ x $ is put into the heap for later processing.
Later on $ x $ is processed but its level does not increase.
Then we can charge the running time of $ O(\log{n}) $ to the time $ O (\deg_{\cH} (x) \log{n}) $ that we already charge to $ y $.

Since the monotone ES-tree is only maintained up to depth $ (\alpha + \beta/\tau) \Q + \beta $, at most $ (\alpha + \beta/\tau) \Q + \beta $ level increases are possible for every node.
Thus, the total update time of the monotone ES-tree is $ O (\phi_k (\cH) \log{n} + \sum_{x \in U} \deg_{\cH} (x) ((\alpha + \beta/\tau) \Q + \beta) \log{n}) $.
Since $ \sum_{x \in U} \deg_{\cH} (x) \leq 2 |E_\cH (U)| $, the total update time is $ O (\phi_k (\cH) + E_\cH (U) ((\alpha + \beta/\tau) \Q + \beta) \log{n}) $.
\end{proof}

\paragraph*{Eliminating the $ \log{n} $-factor.}
The factor $ \log{n} $ in the running time of \Cref{lem:running_time_monotone_ES_tree_with_log} comes from using a heap $ Q $ and, for every node $ u $, a heap $ N (u) $.
We now want to avoid using these heaps and only charge $ O (\deg_{\cH} (u)) $ to every level increase of a node $ u $ and time $ O(1) $ to every update in $ \cH $.
King~\cite{King99} explained how to eliminate the $ \log{n} $-factor for the classic ES-tree.
However, we cannot use the same modified data structures as King because of the possibility of insertions and edge weight increases.

First we explain how to avoid the heap $ Q $.
Observe that every time we increase the level of a node, it suffices to increase the level by only $ 1 $.
Thus, instead of a heap for $ Q $ we can also use a simple queue, implemented with a list that allows us to retrieve and remove its first element and to append an element at its end.

Now we explain how to avoid the heap $ N (u) $ of every node $ u $.
Remember that we only want to increase the level of a node $ u $ if there is no neighbor $ v $ of $ u $ such that
\begin{equation}
\ell (v) + w (u, v) \leq \ell (u) \, . \label{eq:parent_inequality}
\end{equation}
Therefore we maintain a counter $ c (u) $ for every node $ u $ such that $ c (u) = | \{ v \mid \ell (v) + w (u, v) \leq \ell (u) \} | $.\footnote{The idea of maintaining this kind of counter has previously been used by Brim et al.~\cite{Brim11} in the context of mean-payoff games.}
If the counters are correctly maintained, we can simply check whether $ c (u) $ is $ 0 $ to determine whether the level of $ u $ has to increase (which replaces Lines~\ref{line:retrieve_min_monotone_ES_tree} and~\ref{line:check_for_level_increase} of \Cref{alg:monotone_ES_tree}).
For a node $ u $ and its neighbor $ v $ the status of Inequality~\eqref{eq:parent_inequality} only changes (i.e., the inequality starts or stops being satisfied) in the following cases:
\begin{itemize}
\item The level of $ u $ or the level of $ v $ increases.
\item The weight of the edge $ (u, v) $ increases.
\item The edge $ (u, v) $ is inserted (thus $ v $ becomes a neighbor of $ u $).
\item The edge $ (u, v) $ is deleted (thus $ v $ stops being a neighbor of $ v $).
\end{itemize}
Note that for two nodes $ u $ and $ v $ we can check whether they satisfy Inequality~\eqref{eq:parent_inequality} in constant time.
Thus, we can efficiently maintain the counters as follows:
\begin{itemize}
\item Every time we update an edge $ (u, v) $ (by an insertion, deletion, or weight increase), we check in constant time whether Inequality~\eqref{eq:parent_inequality} holds before the update and whether it holds after the update.
Then we increase or decrease $ c(v) $ and $ c(u) $ if necessary.
These operations take constant time, which we charge to the update in $ \cH $.
\item Every time $ \ell (u) $ increases, we recompute $ c(u) $.
This takes time $ O (\deg_{\cH} (u)) $.
Furthermore, for every neighbor $ v $ of $ u $, we check in constant time whether Inequality~\eqref{eq:parent_inequality} holds before the update and whether it holds after the update.
Then we increase or decrease $ c(v) $ if necessary.
This takes constant time for every neighbor of $ u $ and thus time $ O (\deg_{\cH} (u)) $ for all of them.
We can charge the running time $ O (\deg_{\cH} (u)) $ to the level increase of $ u $.
\end{itemize}

Having explained how to maintain the counters, the remaining running time analysis is the same as in \Cref{lem:running_time_monotone_ES_tree_with_log}.
The improved running time can therefore be stated as follows.
\begin{lemma}[Improved Running Time]\label{lem:running_time_monotone_ES_tree_without_log}
For $ k $ deletions in $ \cG $, the monotone ES-tree can be implemented with a total update time of $ O(\phi_k(\cH) + | E_k (\cH) | \cdot ((\alpha + \beta/\tau) \Q + \beta) ) $, where $ E_k (\cH) $ is the set of all edges ever contained in $ \cH $ up to time $ k $.
\end{lemma}

Note that the solution proposed above does not allow us to retrieve the parent of every node in the tree in constant time.
This would be desirable because then, for every node $ v $, we could not only get the approximate distance of $ v $ to the root in constant time, but also a path of corresponding or smaller length in time proportional to the length of this path.

We can achieve this property as follows.
For every node $ u $ we maintain a list $ L (u) $ of nodes.
\emph{Every time} a node $ u $ and one of its neighbors $ v $ start to satisfy Inequality~\eqref{eq:parent_inequality}, $ v $ is appended to $ L (u) $.
Note that it is \emph{not} always the case that $ u $ and all nodes $ v $ in the list $ L (u) $ satisfy Inequality~\eqref{eq:parent_inequality}.
We just have the guarantee that they satisfied it at some previous point in time.
However, the converse is true: If $ u $ and its neighbor $ v $ currently satisfy Inequality~\eqref{eq:parent_inequality}, then $ v $ is contained in $ L (u) $.
Using the same argument as above for maintaining the counters, the running time for appending nodes to the lists is paid for by charging $ O(1) $ to every update in $ \cH $ and $ O (\deg_{\cH} (u)) $ to every level increase of a node $ u $.

We can now decide whether the level of a node $ u $ has to increase as follows (this replaces Lines~\ref{line:retrieve_min_monotone_ES_tree} and~\ref{line:check_for_level_increase} of \Cref{alg:monotone_ES_tree}).
Look at the first node $ v $ in the list $ L (u) $.
If $ u $ and $ v $ \emph{still} satisfy Inequality~\eqref{eq:parent_inequality}, the level of $ u $ does not have to increase.
Otherwise, we retrieve and remove the first element from the list until we find a node $ v $ such that $ u $ and $ v $ satisfy Inequality~\eqref{eq:parent_inequality}.
If no such node $ v $ can be found in the list, then the list will be empty after this process and we know that the level of $ u $ has to increase.
Otherwise, the first node in the list $ L (u) $ serves as the parent of $ u $ in the tree.
The constant running time for reading and removing the first node can be charged to the previous appending of this node to $ L (u) $.

Note that the list $ L (u) $ of each node $ u $ might require a lot of space because some nodes might appear several times.
If we want to save space, we can do the following.
For every node $ u $ we maintain a set $ S (u) $ that stores for every neighbor of $ u $ whether it is contained in $ L(u) $.
Every time we add or remove a node from $ L (u) $ we also add or remove it from $ S (u) $.
Before adding a node to $ L (u) $ we additionally check whether it is already contained in $ S(u) $ and thus also in $ L(u) $.
We implement $ S(u) $ with a dynamic dictionary using dynamic perfect hashing~\cite{DietzfelbingerKMHRT94} or cuckoo hashing~\cite{PaghR04}.
This data structure needs time $ O(1) $ for look-ups and expected amortized time $ O(1) $ for insertions and deletions.
Thus, the running time bound of \Cref{lem:running_time_monotone_ES_tree_without_log} will still hold in expectation.
Furthermore, for every node $ u $, the space needed for $ L (u) $ and $ S (u) $ is bounded by $ O(\deg_{\cH} (u)) $.
However, this solution is no longer deterministic.

\subsection{From Approximate SSSP to Approximate APSP}\label{sec:approximate_SSSP_to_approximate_APSP}

In the following, we show how a combination of approximate decremental SSSP data structures can be turned into an approximate decremental APSP data structure.
We follow the ideas of Roditty and Zwick~\cite{RodittyZ12}, who showed how to obtain approximate APSP from \emph{exact} SSSP.
We remark that one can obtain an efficient APSP data structure from this reduction, if the running time of the (approximate) SSSP data structure depends on the distance range that it covers in a specific way.

We first define an approximate version of the center cover data structure and show how such a data structure can be obtained from an approximate decremental SSSP data structure by marginally worsening the approximation guarantee.
We slightly modify the notions of a center cover and a center cover data structure we gave in \Cref{sec:Roditty_Zwick_framework}, where we reviewed the algorithmic framework of Roditty and Zwick~\cite{RodittyZ12}.
The main idea behind their APSP data structure is to maintain $ \log{n} $ instances of center cover data structures such that the instance $ p $ can answer queries for the approximate distance of two nodes $ x $ and $ y $ if the distance between them is in the range from $ 2^p $ to $ 2^{p+1} $.
Arbitrary distance queries can then be answered by performing a binary search over the instances to determine $ p $.
We will follow this approach using approximate instead of exact data structures.

\begin{definition}[Approximate center cover]
Let $ U $ be a set of nodes in a graph $ G $, let $ \q $ be a positive integer, the \emph{cover range}, and let $ \alpha \geq 1 $ and $ \beta \geq 0 $.
We say that a node $ x $ is \emph{$ (\alpha, \beta) $-covered} by a node $ \cen \in U $ in $ G $ if $ \dist_G (x, \cen) \leq \alpha \q + \beta $.
We say that $ U $ is an \emph{$ (\alpha, \beta) $-approximate center cover of $ G $ with parameter $ \q $} if every node $ x $ that is in a connected component of size at least $ \q $ is $ (\alpha, \beta) $-covered by some node $ \cen \in U $ in~$ G $.
\end{definition}

\begin{definition}\label{def:approximate_center_cover}
An \emph{$ (\alpha, \beta) $-approximate center cover data structure} with \emph{cover range parameter $ \q $} and \emph{distance range parameter $ \Q $} for a decremental graph $ \cG = (G_i)_{0 \leq i \leq k} $ maintains, for every $ 0 \leq i \leq k $, a set of \emph{centers} $ \C_i = \{ 1, 2, \ldots, l \} $ and a set of nodes $ U_i = \{ \cen_i^1, \cen_i^2, \ldots, \cen_i^l \} $ such that $ U_i $ is an $ (\alpha, \beta) $-approximate center cover of $ G_i $ with parameter $ \q $.
For every center $ j \in \C_i $ and every $ 0 \leq i \leq k $, we call $ \cen_i^j $ the \emph{location} of center $ j $ in $ G_i $ and for every node $ x $ we say that $ x $ is $ (\alpha, \beta) $-covered by $ j $ in $ G_i $ if $ x $ is $ (\alpha, \beta) $-covered by $ \cen_i^j $ in $ G_i $.
After the $i$-th edge deletion (where $ 0 \leq i \leq k $), the data structure provides the following operations:
\begin{itemize}
\item \Delete{$u$, $v$}: Delete the edge $ (u, v) $ from $ G_i $.
\item \Distance{$j$, $x$}: Return an estimate $ \delta_i (\cen_i^j, x) $ of the distance between the location $ \cen_i^j $ of center $ j $ and the node $ x $ such that $ \delta_i (\cen_i^j, x) \leq \alpha \dist_{G_i} (\cen_i^j, x) + \beta $, provided that $ \dist_{G_i} (\cen_i^j, x) \leq \Q $.
If $ \dist_{G_i} (\cen_i^j, x) > \Q $, then either return $ \delta_i (\cen_i^j, x) = \infty $ or return $ \delta_i (\cen_i^j, x) \leq \alpha \dist_{G_i} (\cen_i^j, x) + \beta $.
\item \FindCenter{$x$}: If $ x $ is in a connected component of size at least $ \q $, then return a center $ j $ (with current location $ \cen_i^j $) such that $ \dist_{G_i} (x, \cen_i^j) \leq \alpha \q + \beta $.
If $ x $ is in a connected component of size less than $ \q $, then either return $ \bot $ or return a center $ j $ such that $ \dist_{G_i} (x, \cen_i^j) \leq \alpha \q + \beta $.
\end{itemize}
The \emph{total update time} is the total time needed to perform all $ k $ delete operations and the initialization, and the \emph{query time} is the worst-case time needed to answer a single distance or find center query.
\end{definition}

We now show how to obtain an approximate center cover data structure that is correct with high probability, which means that, with small probability, the operation \FindCenter{$x$} might return $ \bot $ although $ x $ is in a connected component of size at least $ \q $.

\begin{lemma}[Approximate SSSP implies approximate center cover]\label{lem:approximate_SSSP_to_approximate_center_cover}
Let $ \q $ and $ \Q $ be parameters such that $ \q \leq \Q $.
If there are $ (\alpha, \beta) $-approximate decremental SSSP data structures with distance range parameters $ \q $ and $ \Q $ for some $ \alpha \geq 1 $ and $ \beta \geq 0 $ that have constant query times and total update times of $ T (\q) $ and $ T (\Q) $, respectively (where $ T (\Q) $ is $ \Omega (n) $), then there is an $ (\alpha, \beta) $-approximate center cover data structure that is correct with high probability and has constant query time and an expected total update time of $ O( (T (\Q) n \log{n}) / \q) $.
\end{lemma}

\begin{proof}
Let $ \cG = (G_i)_{0 \leq i \leq k} $ be a decremental graph.
It is well known (see, for example, \cite{UllmanY91} and~\cite{RodittyZ12}) that, by random sampling, we can obtain a set $ U = \{ \cen^1, \cen^2, \ldots, \cen^l \} $ of expected size $ O (n \log{n} / \q) $ that is a center cover of $ G_i $ for every $ i \leq k $ with high probability.
Clearly, every center cover is also an $ (\alpha, \beta) $-approximate center cover.
Thus, $ U $ is an $ (\alpha, \beta) $-approximate center cover of $ G_i $ for every  $ 0 \leq i \leq k $.
Throughout all deletions, the set $ \C = \{ 1, 2, \ldots, l \} $ will serve as the set of centers and each center~$ j $ will always be located at the same node $ \cen^j $.

We use the following data structures:
For every center $ j $, we maintain two $ (\alpha, \beta) $-approximate decremental SSSP data structures with source $ \cen^j $: For the first one we use the parameter $ \q $, and for the second one we use the parameter $ \Q $.
As there are $ O (n \log{n} / \q) $ centers, the total update time for all these SSSP data structures is $ O( T (\Q) (n \log{n}) / \q) $.
For every node $ x $ and every center $ j $, let $ \delta_i (x, \cen^j) $ denote the estimate of the distance between $ x $ and the location of center $ j $ returned by the second SSSP data structure with source $ \cen^j $ after the $i$-th edge deletion.
For every node $ x $ we maintain a set $ S_x $ of centers that cover $ x $ such that (a) if $ \dist_{G_i} (x, \cen^j) \leq \q $, then $ j \in S_x $ and (b) for all $ j \in S_x $, $ \delta_i (x, \cen^j) \leq \alpha \q + \beta $.

The set $ S_x $ can be implemented by using an array of size $ |C| = O((n \log{n}) / \q) $ for every node $ x $.
We initialize $ S_x $ in time $ O((n \log{n}) / \q) $ as follows:
For every center~$ j $, we query $ \delta_0 (x, \cen^j) $ and insert $ j $ into $ S_x $ if $ \delta_0 (x, \cen^j) \leq \alpha \q + \beta $.
Since $ \delta_0 (x, \cen^j) \leq \alpha \dist_{G_0} (x, \cen^j) + \beta $, this includes every center $ j $ such that $ \dist_{G_0} (x, \cen^j) \leq \q $.
To maintain the sets of centers, we do the following after every deletion.
Remember that for every center $ j $, the first SSSP data structure with source $ \cen^j $ returns every node $ x $ such that $ \delta_i (x, \cen^j) \leq \alpha \q + \beta $ and $ \delta_{i+1} (x, \cen^j) > \alpha \q + \beta $.
For every such node $ x $ we remove~$ j $ from $ S_x $.
Note that every center $ j $ with $ \delta_t (x, \cen^j) > \alpha \q + \beta $ (for $ 0 \leq t \leq k $) can safely be removed from $ S_x $ because $ \delta_t (x, \cen^j) > \alpha \q + \beta $ implies $ \dist_{G_t} (x, \cen^j) > \q $ and $ \dist_{G_i} (x, \cen^j) \geq \dist_{G_t} (x, \cen^j) $ for all $ i \geq t $.
We can charge the running time for maintaining the sets of centers to the delete operations in the SSSP data structures.
Thus, this running time is already included in the total update time stated above.
For every node $ x $, no center is ever added to $ S_x $ after the initialization.
Thus, in the array representing $ S_x $, we can maintain a pointer to the leftmost center time proportional to the size of the array, which is $ |C| = O((n \log{n}) / \q) $.

We now show how to perform the operations of an approximate center cover data structure, as specified in \Cref{def:approximate_center_cover}, in constant time.
Let $ i $ be the index of the last deletion.
Given a center $ j $ and a node $ x $, we answer a query for the distance of $ x $ to $ \cen^j $ by returning $ \delta_i (x, \cen^j) $ from the second SSSP data structure of $ \cen^j $, which gives an $ (\alpha, \beta) $-approximation of the true distance.
Given a node $ x $, we answer a query for finding a nearby center by returning any center $ j $ in the set of centers $ S_x $ of $ x $.
If $ S_x $ is empty, we return $ \bot $.
Note that for every center $ j $ in $ S_x $ we know that $ \dist_{G_i} (x, \cen^j) \leq \alpha \q + \beta $, as required, because $ \dist_{G_i} (x, \cen^j) \leq \delta_i (x, \cen^j) $.
If $ x $ is in a connected component of size at least $ \q $, we can ensure that we find a center $ j $ in $ S_x $ because, by our random choice of centers, we have $ \dist_{G_i} (x, \cen^j) \leq \q $ for some center $ j $ with high probability.
If $ \dist_{G_i} (x, \cen^j) \leq \q $, then $ S_x $ contains $ j $.
\end{proof}

We now show why the approximate center cover data structure is useful.
If one can obtain an approximate center cover data structure, then one also obtains an approximate decremental APSP data structure with slightly worse approximation guarantee.
The proof of this observation follows Roditty and Zwick~\cite{RodittyZ12}.
In their algorithm, Roditty and Zwick keep a set of nodes $ U $ (which we call centers) such that every node (that is in a sufficiently large connected component) is ``close'' to some node in $ U $.
To be able to efficiently find a close center for every node, they maintain, for every node, the nearest node in the set of centers.
However, it is sufficient to return \emph{any} center that is close.

\begin{lemma}[Approximate center cover implies approximate APSP]\label{lem:approximate_centers_data_structure_to_APSP}
Assume that for all parameters $ \q $ and $ \Q $ such that $ \q \leq \Q $ there is an $ (\alpha, \beta) $-approximate center cover data structure that has constant query time and a total update time of $ T (\q, \Q) $.
Then, for every $ 0 < \epsilon \leq 1 $, there is an $ (\alpha + 2 \epsilon \alpha^2, 2 \beta + 2 \alpha \beta) $-approximate decremental APSP data structure with $ O(\log{\log n}) $ query time and a total update time of $ \hat{T} = \sum_{p=0}^{\lfloor \log{n} \rfloor} T (\qp, \Qp) $, where $ \qp = \epsilon 2^p $ and $ \Qp = \alpha \epsilon 2^p + \beta + 2^{p+1} $ (for $ 0 \leq p \leq \lfloor \log{n} \rfloor $).

The query time can be reduced to $ O (1) $ if there is an $ (\alpha', \beta') $-approximate decremental APSP data structure for some constants $ \alpha' $ and $ \beta' $ with constant query time and a total update time of $ \hat{T} $.
\end{lemma}

\begin{proof}
The data structure uses $ \lceil \log n \rceil $ many instances where the $ p $-th instance is responsible for the distance range from $ 2^p $ to $ 2^{p+1} $. 
For the $ p $-th instance we maintain a center cover data structure using the parameters $ \qp = \epsilon 2^p $ and $ \Qp = 2^{p+1} + \alpha \epsilon 2^p + \beta $.
For every center $ j $ and every node $ x $, let $ \delta_i^p (\cen^j, x) $ denote the estimate of the distance between $ \cen^j $ and $ x $ provided by the $p$-th center cover data structure.
Let $ \cG = (G_i)_{0 \leq i \leq k} $ be a decremental graph and let $ i $ be the index of the last deletion.

For every instance $ p $, we can compute a distance estimate $ \hat{\delta}_i^p (x, y) $ for all nodes $ x $ and $ y $ as follows.
Using the center cover data structure, we first check whether there is some center $ j $ with location $ \cen^j $ that $ (\alpha, \beta) $-covers $ x $, i.e., $ \dist_{G_i} (x, \cen^j) \leq \alpha \qp + \beta $.
If $ x $ is not $ (\alpha, \beta) $-covered by any center, we set $ \hat{\delta}_i^p (x, y) = \infty $.
Otherwise we query the center cover data structure to get estimates $ \delta_i^p (\cen^j, x) $ and $ \delta_i^p (\cen^j, y) $ of the distances between $ \cen^j $ and $ x $ and between $ \cen^j $ and $ y $, respectively.
(Remember that these distance estimates might be $ \infty $.)
We now set $ \hat{\delta}_i^p (x, y) = \delta_i^p (\cen^j, x) + \delta_i^p (\cen^j, y) $.
Note that, given~$ p $, we can compute $ \hat{\delta}_i^p (x, y) $ in constant time.
The query procedure will rely on three properties of the distance estimate $ \hat{\delta}_i^p (x, y) $.
\begin{enumerate}
\item The distance estimate never underestimates the true distance, i.e., $ \hat{\delta}_i^p (x, y) \geq \dist_{G_i} (x, y) $.
\item If $ \dist_{G_i} (x, y) \geq 2^p $ and $ \hat{\delta}_i^p (x, y) \neq \infty $, then $ \hat{\delta}_i^p (x, y) \leq (\alpha + 2 \alpha^2) \dist_{G_i} (x, y) + 2 \beta + 2 \alpha \beta $.
\item If $ x $ is in a connected component of size at least $ \qp $ and $ \dist_{G_i} (x, y) \leq 2^{p+1} $, then $ \hat{\delta}_i^p (x, y) \neq \infty $.
\end{enumerate}

The first property is clearly true if $ \hat{\delta}_i^p (x, y) = \infty $ and otherwise follows by applying the triangle inequality (note that $ \dist_{G_i} (\cen^j, y) \leq \delta_i^p (\cen^j, y) $ in any case):
\begin{equation*}
\dist_{G_i} (x, y) \leq \dist_{G_i} (\cen^j, x) + \dist_{G_i} (\cen^j, y) \leq \delta_i^p (\cen^j, x) + \delta_i^p (\cen^j, y) = \hat{\delta}_i^p (x, y) \, .
\end{equation*}
Thus, $ \hat{\delta}_i^p (x, y) $ never underestimates the true distance.
For the second property we remark that if $ \hat{\delta}_i^p (x, y) \neq \infty $, it must be the case that we have found a center $ j $ with location $ \cen^j $ that $ (\alpha, \beta) $-covers $ x $.
Therefore $ \dist_{G_i} (x, \cen^j) \leq \alpha \qp + \beta $.
Furthermore it must be the case that $ \delta_i^p (x, \cen^j) \neq \infty $ and $ \delta_i^p (\cen^j, y) \neq \infty $, and therefore $ \delta_i^p (x, \cen^j) \leq \alpha \dist_{G_i} (x, \cen^j) + \beta $ and $ \delta_i^p (\cen^j, y) \leq \alpha \dist_{G_i} (\cen^j, y) + \beta $.
Now simply consider the following chain of inequalities:
\begin{align*}
\hat{\delta}_i^p (x, y) = \delta_i^p (x, \cen^j) + \delta_i^p (\cen^j, y) &\leq \alpha (\dist_{G_i} (\cen^j, x) + \dist_{G_i} (\cen^j, y)) + 2 \beta \\
 &\leq \alpha (\dist_{G_i} (\cen^j, x) + \dist_{G_i} (\cen^j, x) + \dist_{G_i} (x, y)) + 2 \beta \\
 &= \alpha (2 \dist_{G_i} (\cen^j, x) + \dist_{G_i} (x, y)) + 2 \beta \\
 &\leq \alpha (2 \alpha \qp + 2\beta + \dist_{G_i} (x, y)) + 2 \beta \\
 &= \alpha (2 \alpha \epsilon 2^p + 2\beta + \dist_{G_i} (x, y)) + 2 \beta \\
 &\leq \alpha (2 \alpha \epsilon \dist_{G_i} (x, y) + 2\beta + \dist_{G_i} (x, y)) + 2 \beta \\
 &= (\alpha + 2 \epsilon \alpha^2) \dist_{G_i} (x, y) + 2 \beta + 2 \alpha \beta
\end{align*}
We now prove the third property.
If $ x $ is in a component of size at least $ \qp $, then, with high probability, it is covered by some center $ j $ with location $ \cen^j $, and we have 
\begin{equation*}
\dist_{G_i} (\cen^j, x) \leq \alpha \qp + \beta = \alpha \epsilon 2^p + \beta \leq \Qp \, .
\end{equation*}
Therefore we get $ \delta_i^p (\cen^j, x) \leq \alpha \dist_{G_i} (\cen^j, x) + \beta < \infty $.
Furthermore, we have
\begin{equation*}
\dist_{G_i} (\cen^j, y) \leq \dist_{G_i} (\cen^j, x) + \dist_{G_i} (x, y) \leq \alpha \epsilon 2^p + \beta + 2^{p+1} = \Qp \, ,
\end{equation*}
which gives $ \delta_i^p (\cen^j, y) \leq \alpha \dist_{G_i} (\cen^j, y) + \beta < \infty $.
As both of its components are not $ \infty $, the sum $ \hat{\delta}_i^p (x, y) = \delta_i^p (\cen^j, x) + \delta_i^p (\cen^j, y) $ is also not $ \infty $, as desired.

A query time of $ O (\log{n}) $ is immediate as we can simply return the minimum of all distance estimates $ \hat{\delta}_i^p (x, y) $.
A query time of $ O (\log {\log{n}}) $ is possible because of the following idea:
If $ \dist_{G_i} (x, y) \neq \infty $, it is sufficient to find the minimum index $ p $ such that $ \hat{\delta}_i^p (x, y) \neq \infty $.
This minimum index can be found by performing binary search over all $ \log{n} $ possible indices.
Furthermore, the query time can be reduced to $ O(1) $ if there is a second $ (\alpha', \beta') $-approximate decremental APSP data structure with constant query time for some constants $ \alpha' $ and~$ \beta' $.
We first compute the distance estimate $ \delta_i' (x, y) $ of the second data structure for which we know that $ \dist_{G_i} (x, y) \in [ \delta_i' (x, y) / \alpha' - \beta', \delta_i' (x, y) ] $.
Now there is only a constant number of indices $ p $ such that $ \{ 2^p, \ldots, 2^{p+1} \} \cap [ \delta_i' (x, y) / \alpha' - \beta', \delta_i' (x, y) ] \neq \emptyset $.
For every such index we compute $ \hat{\delta}_i^p (x, y) $ and return the minimum distance estimate obtained by this process.
\end{proof}

Finally, we show how to obtain an approximate decremental APSP data structure from an approximate decremental SSSP data structure if the approximation guarantee is of the form $ (\alpha + \epsilon, \beta) $.
In that case we can avoid the worsening of the approximation guarantee of \Cref{lem:approximate_SSSP_to_approximate_center_cover}.

\begin{lemma}\label{lem:approximate_APSP_with_additive_error}
Assume that for some $ \alpha \geq 1 $ and $ \beta \geq 0 $, every $ 0 < \epsilon \leq 1 $, and all $ 0 \leq \Q $ there is an $ (\alpha+\epsilon, \beta) $-approximate decremental SSSP data structure with distance range parameter $ \Q $ that has constant query time and a total update time of $ T' (\Q, \epsilon) $.
Then there is an $ (\alpha+\epsilon, \beta) $-approximate decremental APSP data structure with a query time of $ O (\log{\log{n}}) $ and a total update time of
\begin{equation*}
\hat{T} = \sum_{p=0}^{\lfloor \log{n} \rfloor} (T' (\Qp, \hat{\epsilon}) n \log{n}) / \qp + n T' (\Qhat, \hat{\epsilon})
\end{equation*}
where $ \hat{\epsilon} = \epsilon / (18 \alpha^2) $, $ \Qhat = (4 \alpha + 8 \beta) / \hat{\epsilon} $, $ \qp = \hat{\epsilon} 2^p $, and $ \Qp = \alpha \hat{\epsilon} 2^p + \beta + 2^{p+1} $ (for $ 0 \leq p \leq \lfloor \log{n} \rfloor $).

The query time can be reduced to $ O (1) $ if there is an $ (\alpha', \beta') $-approximate decremental APSP data structure for some constants $ \alpha' $ and~$ \beta' $ with constant query time and a total update time of $ \hat{T} $.
\end{lemma}

\begin{proof}
By combining \Cref{lem:approximate_SSSP_to_approximate_center_cover} with \Cref{lem:approximate_centers_data_structure_to_APSP} the approximate decremental SSSP data structure implies that there is an $ (\hat{\alpha}, \hat{\beta}) $-approximate decremental APSP data structure where $ \hat{\alpha} = (\alpha + \hat{\epsilon}) + 2 \hat{\epsilon} (\alpha + \hat{\epsilon})^2  $ and $ \hat{\beta} = 2 \beta + 2 (\alpha + \hat{\epsilon}) \beta $.
This APSP data structure has a query time of $ O (\log{\log{n}}) $ and a total update time of
\begin{equation*}
\sum_{p=0}^{\lfloor \log{n} \rfloor} ( T' (\Qp, \hat{\epsilon}) (n \log{n}) / \qp \, .
\end{equation*}
By \Cref{lem:approximate_centers_data_structure_to_APSP} the query time can be reduced to $ O(1) $ if, for some constants $ \alpha' $ and~$ \beta' $, there is an $ (\alpha', \beta') $-approximate decremental APSP data structure with constant query time and a total update time of $ \hat{T} $.

The data structure above provides, for every decremental graph $ \cG = (G_i)_{0 \leq i \leq k} $ and all nodes $ x $ and $ y $, a distance estimate $ \delta_i (x, y) $ such that $ \dist_{G_i} (x, y) \leq \delta_i (x, y) \leq \hat{\alpha} \dist_{G_i} (x, y) + \hat{\beta} $ after the $i$-th deletion.
By our choice of $ \hat{\epsilon} = \epsilon / (18 \alpha^2) $ we get
\begin{equation*}
\hat{\alpha} = (\alpha + \hat{\epsilon}) + 2 \hat{\epsilon} (\alpha + \hat{\epsilon})^2  \leq \alpha + \hat{\epsilon} \alpha^2 + 2 \hat{\epsilon} (\alpha + \alpha)^2 = \alpha + 9 \hat{\epsilon} \alpha^2 = \alpha + \epsilon / 2
\end{equation*}
and
\begin{equation*}
\hat{\beta} = 2 \beta + 2 (\alpha + \hat{\epsilon}) \beta \leq 2 \beta + 2 (\alpha + 1) \beta = (2 \alpha + 4) \beta
\end{equation*}
Thus, if $ \dist_{G_i} (x, y) \geq (4 \alpha + 8 \beta) / \epsilon $, then
\begin{align*}
\delta_i (x, y) \leq (\alpha + \epsilon / 2) \dist_{G_i} (x, y) + (2 \alpha + 4) \beta &\leq (\alpha + \epsilon / 2) \dist_{G_i} (x, y) + \epsilon \dist_{G_i} (x, y) / 2 \\
 &= (\alpha + \epsilon) \dist_{G_i} (x, y) \, .
\end{align*}

Additionally, we use a second approximate decremental APSP data structure to deal with distances that are smaller than $ (4 \alpha + 8 \beta) / \epsilon $ (which is less than $ (4 \alpha + 8 \beta) / \hat{\epsilon} $).
For this data structure we simply maintain an $ (\alpha + \hat{\epsilon}, \beta) $-approximate decremental SSSP data structure for every node with distance range parameter $ \Qhat = (4 \alpha + 8 \beta) / \hat{\epsilon} $.
We answer distance queries by returning the minimum of the distance estimates provided by both APSP data structures.
As both APSP data structures never underestimate the true distance, the minimum of both distance estimates gives the desired $ (\alpha + \epsilon, \beta) $-approximation.
\end{proof}

\subsection{Putting Everything Together: $ \tilde O(n^{5/2})$-Total Time Algorithm for $(1+\epsilon, 2)$- and $(2+\epsilon, 0)$-Approximate APSP}\label{sec:putting together for faster algorithms}

In the following we show how the monotone ES-tree of \Cref{thm:monotone ES tree} together with the locally persevering emulator of \Cref{thm:emulator} can be used to obtain $ (1+\epsilon, 2) $- and $ (2+\epsilon, 0) $-approximate decremental APSP data structures with $ \tilde O (n^{5/2} / \epsilon^2) $ total update time.
These results are direct consequences of the previous parts of this section.
We first show how to obtain a $ (1+\epsilon, 2) $-approximate decremental SSSP data structure.
Using \Cref{lem:approximate_APSP_with_additive_error} we then immediately obtain a $ (1+\epsilon, 2) $-approximate decremental APSP data structure.

\begin{corollary}[$ (1 + \epsilon, 2) $-approximate monotone ES-tree]\label{cor:concrete_monotone_ES_tree}
Given the $ (1, 2, \lceil 2 / \epsilon \rceil) $-locally persevering emulator $ \cH $ of \Cref{thm:emulator}, there is a $ (1 + \epsilon, 2) $-approximate decremental SSSP data structure for every distance range parameter $ \Q $ that is correct with high probability, and has constant query time and an expected total update time of $ O ( n^{3/2} \log{n} / \epsilon + n^{3/2} \Q \log{n}) $, where the time for maintaining $ \cH $ is not included.
\end{corollary}

\begin{proof}
Let $ \cG = (G_i)_{0 \leq i \leq k} $ be a decremental graph and let $ \cH $ be the $ (1, 2, \lceil 2 / \epsilon \rceil) $-locally persevering emulator of \Cref{thm:emulator}.
By \Cref{thm:monotone ES tree} there is an approximate decremental SSSP data structure for every source node $ \treeroot $ and every distance range parameter $ \Q $.
Let $ \delta_i (x, \treeroot) $ denote the estimate of the distance between $ x $ and $ \treeroot $ provided after the $i$-th edge deletion in $ \cG $.
By \Cref{thm:monotone ES tree} we have $ \dist_{G_i} (x, c) \leq \delta_i (x, c) $, and furthermore, if $ \dist_{G_i} (x, c) \leq \Q $, then
\begin{equation*}
\delta_i (x, c) \leq (1 + 2 / (\lceil 2 / \epsilon \rceil)) \dist_{G_i} (x, c) + 2 \leq (1 + \epsilon) \dist_{G_i} (x, c) + 2 \, .
\end{equation*}
By \Cref{thm:emulator}, the number of edges ever contained in the emulator is $ | E_k (\cH) | = O (n^{3/2} \log{n}) $ and the total number of updates in $ \cH $ is $ \phi_k (\cH) = O (n^{3/2} \log{n} / \epsilon) $.
Therefore, by \Cref{thm:monotone ES tree}, the total update time of the approximate decremental SSSP data structure is
\begin{multline*}
O ( \phi_k(\cH) + |E_k (\cH) | \cdot ((\alpha + \beta/\tau) \Q + \beta) ) \\
= O ( (n^{3/2} \log{n}) / \epsilon + (n^{3/2} \log{n}) \cdot ((1 + \epsilon) \Q + 2) ) \\
= O ( (n^{3/2} \log{n}) / \epsilon + n^{3/2} \Q \log{n} ) \, . \qedhere
\end{multline*}
\end{proof}

\begin{theorem}[Main result of \Cref{sec:faster}: Randomized $(1+\epsilon, 2)$-approximation with truly-subcubic total update time]\label{thm:main_APSP_approximation}
For every $ 0 < \epsilon \leq 1 $, there is a $ {(1 + \epsilon, 2)} $-approximate decremental APSP data structure with constant query time and an expected total update time of $ O ( (n^{5/2} \log^3{n}) / \epsilon ) $ that is correct with high probability.
\end{theorem}

\begin{proof}
We set $ \hat{\epsilon} = \epsilon / 18 $.
Let $ \cH $ denote the $ (1, 2, 2/\hat{\epsilon}) $-locally persevering emulator of \Cref{thm:emulator}.
The total update time for maintaining $ \cH $ is $ O (m n^{1/2} \log{n} / \epsilon) $.
Since $ m \leq n^2 $ this is within the claimed total update time.
By \Cref{cor:concrete_monotone_ES_tree} we can use $ \cH $ to maintain, for every distance range parameter $ \Q $, a $ (1+\hat{\epsilon}, 2) $-approximate decremental SSSP data structure that has constant query time and a total update time of $ T (\Q) = O ( (n^{3/2} \log{n}) / \epsilon + n^{3/2} \Q \log{n}) $.

Using $ \alpha = 1 $ and $ \beta = 2 $, it follows from \Cref{lem:approximate_APSP_with_additive_error} that there is a $ (1+\epsilon, 2) $-approximate decremental APSP data structure whose total update time is proportional to
\begin{gather*}
\sum_{p=0}^{\lfloor \log{n} \rfloor} (T(\Qp) n \log{n}) / \qp + T (\Qhat) n = \\
\sum_{p=0}^{\lfloor \log{n} \rfloor} ( (n^{3/2} \log{n}) / \hat{\epsilon} + n^{3/2} \Qp \log{n}) (n \log{n}) / \qp + ( (n^{3/2} \log{n}) / \hat{\epsilon} + n^{3/2} \Qhat \log{n}) n
\end{gather*}
where $ \hat{\epsilon} = \epsilon / 18 $, $ \Qhat = 12 / \hat{\epsilon} $, $ \qp = \hat{\epsilon} 2^p $, and $ \Qp = \alpha \hat{\epsilon} 2^p + 2^{p+1} + 2 $ (for $ 0 \leq p \leq \lfloor \log{n} \rfloor $).
Note that $ 1/\hat{\epsilon} = O (1/\epsilon) $, $ \Qhat = O (1/\epsilon) $, and $ \Qp / \qp = O (1/\epsilon) $.
Therefore the total update time is $ O ( (n^{5/2} \log^3{n}) / \epsilon ) $.

The query time of the APSP data structure provided by \Cref{lem:approximate_APSP_with_additive_error} can be reduced to $ O (1) $.
The reason is that Bernstein and Roditty~\cite{BernsteinR11} provide, for example, a $ (5 + \epsilon', 0) $-approximate decremental APSP data structure for some constant $ \epsilon' $.
The total update time of this data structure is
\begin{equation*}
\tilde O \left( n^{2 + 1/3 + O (1/\sqrt{\log{n}})} \right)
\end{equation*}
which is well within $ O (n^{5/2}) $.
\end{proof}

The $ (2 + \epsilon, 0) $-approximate decremental APSP data structure now follows as a corollary.
We simply need the following observation:
If the distance between two nodes is $ 1 $, then we can answer queries for their distance exactly by checking whether they are connected by an edge.

\begin{corollary}[Randomized $(2+\epsilon, 0)$-approximation with truly-subcubic total update time]\label{thm:2-approx}
For every $ 0 < \epsilon \leq 1 $, there is a $ (2 + \epsilon, 0) $-approximate decremental APSP data structure with constant query time and an expected total update time of $ O ( (n^{5/2} \log^3{n}) / \epsilon ) $ that is correct with high probability.
\end{corollary}

\begin{proof}
By using the data structure of \Cref{thm:main_APSP_approximation} we can, after the $i$-th edge deletion in a decremental graph $ \cG = (G_i)_{0 \leq i \leq k} $ and for all nodes $ x $ and $ y $, query for a distance estimate $ \delta_i (x, y) $ in constant time that satisfies
\begin{equation*}
\dist_{G_i} (x, y) \leq \delta_i (x, y) \leq (1 + \epsilon) \dist_{G_i} (x, y) + 2 \, .
\end{equation*}
Note that if $ \dist_{G_i} (x, y) \geq 2 $, then
\begin{equation*}
\delta_i (x, y) \leq (1 + \epsilon) \dist_{G_i} (x, y) + 2 \leq (1 + \epsilon) \dist_{G_i} (x, y) + \dist_{G_i} (x, y) = (2 + \epsilon) \dist_{G_i} (x, y) \, .
\end{equation*}
If $ \dist_{G_i} (x, y) < 2 $, then we actually have $ \dist_{G_i} (x, y) \leq 1 $ because $ G_i $ is an unweighted graph.
A distance of $ 1 $ simply means that there is an edge connecting $ x $ and $ y $ in $ G_i $.
Since the adjacency matrix of $ \cG $ is maintained anyway, we can find out in constant time whether $ \dist_{G_i}(x, y) = 1 $.
By setting, for all nodes $ x $ and $ y $,
\begin{equation*}
\delta_i' (x, y) = \begin{cases}
0 & \text{if $ x = y $,} \\
1 & \text{if $ (x, y) \in E (G_i) $,} \\
\delta_i (x, y) & \text{otherwise,}
\end{cases}
\end{equation*}
we get $ \dist_{G_i} (x, y) \leq \delta_i' (x, y) \leq (2 + \epsilon) \dist_{G_i} (x, y) $.
Clearly, this data structure can answer queries in constant time by returning the distance estimate $ \delta_i' (x, y) $ and has the same total update time as the $ (1+\epsilon, 2) $-approximate decremental APSP data structure, namely, $ O ( (n^{5/2} \log^3{n}) / \epsilon ) $.
\end{proof}

\section{Deterministic Decremental $ (1 + \epsilon) $-Approximate APSP with Total Update Time $ O(mn\log n) $}\label{sec:deterministic}

In this section, we present a deterministic decremental $ (1 + \epsilon) $-approximate APSP algorithm with $ O(mn\log n / \epsilon) $ total update time.

\begin{theorem}[Main result of \Cref{sec:deterministic}: Deterministic $O((mn\log n)/\epsilon)$ total update time]\label{thm:deterministic}
For every $ 0 < \epsilon \leq 1 $, there is a deterministic $ (1 + \epsilon, 0) $-approximate decremental APSP data structure with a total update time of $ O ((m n \log n) / \epsilon) $ and a query time of $ O (\log \log n) $.
\end{theorem}

Using known reductions, we show in \Cref{sec:fully_dynamic} that this decremental algorithm implies a deterministic fully dynamic algorithm with an amortized running time of $ \tilde O (m n / (\epsilon t) $ per update and a query time of $ \tilde O (t) $ for every $ t \leq n $.

The main task in proving \Cref{thm:deterministic} is to design a deterministic version of the {\em center cover data structure} (see \Cref{sec:Roditty_Zwick_framework}) with a total deterministic update time of $ O (m n \Q / \q) $ and constant query time.
Once we have this data structure, \Cref{thm:deterministic} directly follows as a corollary from \Cref{lem:centers_data_structure_to_APSP}.
Note that we cannot use the same idea as in~\cite{RodittyZ12} to reduce the query time from $ O (\log \log n) $ to $ O (1) $.
This would require a \emph{deterministic} $ (\alpha, \beta) $-approximate decremental APSP data structure for some constants $ \alpha $ and $ \beta $ with \emph{constant} query time and a total update time of $ O ((m n \log n) / \epsilon) $.
To the best of our knowledge such a data structure has not yet been developed.

Recall that $ \q $ and $ \Q $ are the {\em coverage range} and {\em distance range} parameters where (a) we want every node (in a connected component of size at least $ \q $) to be within distance of at most $\q$ from some center, and (b) we want to maintain the distance from each center to every node within distance at most $\Q$. 
Roditty and Zwick~\cite{RodittyZ12}, following an argument of Ullman and Yannakakis~\cite{UllmanY91}, observed that making each node a center independently at random with probability $ (a \ln n)/\q$, where $a$ is a constant and $1 \le \q \le n$, gives a set $\C$ of centers such that with probability at least $1 - n^{-(a-1)}$ the conditions of a center cover with parameter $ \q $ are fulfilled by $\C$ in the initial graph and the expected size of $\C$ is $O((n \log{n} )/\q)$. The randomized decremental APSP algorithm of~\cite{RodittyZ12} simply chooses a large enough value of $a$ so that with high probability $\C$ not only fulfills the center cover properties with parameter $ \q $ in the initial graph but continues to fulfill them in all the $O(n^2)$ graphs generated during $O(n^2)$ edge deletions. 
This is only possible because it is assumed that the ``adversary'' that generates the deletions is oblivious, i.e., does not know the location of the centers. The main challenge for the {\em deterministic} algorithm is to {\em dynamically adapt} the location and number of the centers so that (i) the center cover properties with size $ \q $ continue to hold, while the graph is modified, and (ii) the total cost incurred is $ O(m n)$.
Once we have such a data structure, we can use the approach of Roditty and Zwick, as discussed in \Cref{sec:Roditty_Zwick_framework}, to obtain an algorithm for maintaining decremental approximate shortest paths. 

The {\em new feature} of our deterministic center cover data structure is that it sometimes {\em moves} centers to avoid opening too many centers (which are expensive to maintain). As we described in \Cref{sec:intro}, the key technique behind the new data structure is what we call a {\em moving Even--Shiloach tree}. We note that the moving ES-tree is actually a concept rather than a new implementation: we implement it in a straightforward way by building a new ES-tree every time we have to move it. However, analyzing the total update time needs new insights and a careful charging argument. To separate the analysis of the moving ES-tree from the charging argument, we describe the data structure in two parts: 
\begin{enumerate}[label=(\arabic{*})]
\item First, in \Cref{sec:deterministic_dynamic_centers}, we give the {\em moving centers} data structure that can answer \Distance and \FindCenter queries, but needs to be told {\em where} to move a center, when a center has to be moved. This data structure is basically an implementation of {\em several} moving ES-trees.\footnote{Later on we want to use the moving centers data structure, and not directly the moving ES-trees, because we will need an additional operation which is not directly provided by the ES-trees (in particular, the \FindCenter operation defined in \Cref{sec:deterministic_dynamic_centers}).}

\item Then, in \Cref{sec:algorithm_maintaining_centers}, we show how to determine when a center (with a moving ES-tree rooted at it) has to be moved and, if so, where it has to move. 
\end{enumerate}
 
Combining these two parts gives the center cover data structure.
 


\subsection{A Deterministic Moving Centers Data Structure (\texttt{MovingCenter})}\label{sec:deterministic_dynamic_centers}\label{sec:moving centers data structure}

In the following, we design a deterministic data structure, called {\em moving centers data structure}, and analyze its cost in terms of the number of centers opened (\opens) and the {\em moving distance} (\moves).
 When a center is created, it is given a unique identifier $j$. The data structure can handle the following operations. 
\begin{definition}[Moving centers data structure]\label{def:moving center}
A \emph{moving centers data structure} with \emph{cover range parameter $ \q $} and \emph{distance range parameter $ \Q $} for a decremental graph $ \cG = (G_i)_{0 \leq i \leq k} $ maintains, for every $ 0 \leq i \leq k $, a set of \emph{centers} $ \C_i = \{ 1, 2, \ldots, l \} $ and a set of nodes $ U_i = \{ \cen_i^1, \cen_i^2, \ldots, \cen_i^l \} $.
For every center $ j \in \C $ and every $ 0 \leq i \leq k $, we call $ \cen_i^j \in U_i $ the \emph{location} of center $ j $ in $ G_i $.
After the $i$-th edge deletion (where $ 0 \leq i \leq k $), the data structure provides the following operations:
\begin{itemize}
\item \Delete{$u$, $v$}: Delete the edge $ (u, v) $ from $ G_i $.
\item \Open{$x$}: Open a new center at node $ x $ and return the ID of the opened center for later use. 
\item \Move{$j$, $x$}: Move the center $ j $ from its current location $ \cen_i^j $ to node $ x $. 
\item \Distance{$j$, $x$}: Return the distance $ \dist_{G_i} (\cen_i^j, x) $ between the location $ \cen_i^j $ of center $ j $ and the node $ x $, provided that $ \dist_{G_i} (\cen_i^j, x) \leq \Q $.
If $ \dist_{G_i} (\cen_i^j, x) > \Q $, then return $ \infty $.
\item \FindCenter{$x$}: Return a center $ j $ (with location $ \cen_i^j $) such that $ \dist_{G_i} (x, \cen_i^j) \leq \q $.
If no such center exists, return $ \bot $.
\end{itemize}
The \emph{total update time} is the total time needed for performing all the delete, open, and move operations and the initialization.
The \emph{query time} is the worst-case time needed to answer a single distance or find center query.
\end{definition}

The moving centers data structure is a first step toward implementing the center cover data structure:
It  can answer all query operations that are posed to the center cover data structure, but, unlike the center cover data structure, it needs to be told where to place the centers and where to open new centers. This information is determined by the data structure in the next section.

In the rest of \Cref{sec:moving centers data structure} we use the following notation:
The decremental graph $ \cG $ undergoes a sequence of $ k $ edge deletions.
By $G_i$ we denote the graph after the $i$-th deletion (for $ 0 \leq i \leq k $).
Each deletion in the graph is reported to the moving centers data structure by a delete operation.
By $\C_i$ we denote the set of centers at the time of the $i$-th delete operation, and by $\cen_i^j$ we denote the location of center $j \in \C_i$ at the time of the $i$-th delete operation.

\begin{definition}[Moving distance (\moves)]\label{def:moving_distance}
The \emph{total moving distance}, denoted by \moves, is defined as $ \moves = \sum_{0 \leq i < k } \sum_{j \in \C_i} \dist_{G_i} (\cen^j_i, \cen^j_{i+1}) $.
\end{definition}

The main result of this section is that we can maintain a moving centers data structure in $ O (m(\opens  \Q +  \moves)) $ time, as in \Cref{lem:existence_dynamic_centers_ds} below.
The data structure is actually very simple: We maintain an ES-tree of depth at most~$\Q$ at every node for which we open a center and for every node to which we move a center. Note that our algorithm treats an ES-tree at each such node as a new tree, regardless of whether we open a center or move a center there. 
While the algorithm can naively treat each ES-tree as a new one, the analysis cannot: If we do so, we will get a total update time of $ O ((\opens + \nummoves) m \Q) $, where $ \nummoves $ is the total number of move-operations (since maintaining each ES-tree takes $O(m\Q)$ total update time).
%
Instead, we bound the cost incurred by the move-operation based on how far a center is moved, i.e., the moving distance \moves. This argument allows us to replace the unfavorable term $ \nummoves m \Q $ by $ \moves m $.
The deterministic center cover data structure of \Cref{sec:algorithm_maintaining_centers} will generate a sequence of center open and move requests so that \moves = $O(n)$.
For simplifying the analysis, we state the following result under a technical assumption, which will always be fulfilled by the intended use of the moving centers data structure in \Cref{sec:algorithm_maintaining_centers}.\footnote{Without this assumption the total update time will be $ O ((\opens  \Q + \nummoves + \moves) m) $, where $ \opens $ is the number of open-operations, $ \nummoves $ is the number of move operations, and $ \moves $ is the total moving distance.}

\begin{proposition}[Main result of \Cref{sec:deterministic_dynamic_centers}: deterministic moving centers data structure]\label{lem:existence_dynamic_centers_ds}
Let $ \q $ and $ \Q $ be parameters such that $ \q \leq \Q $.
Under the assumption that between two consecutive delete operations there can be at most one open or delete operation for each center, there is a moving centers data structure with a total deterministic update time of $ O ((\opens  \Q + \moves) m) $, where $ \opens $ is the number of open-operations and $ \moves $ is the total moving distance.\footnote{Note that the total moving distance might be $ \infty $ if, after some deletion $ i $, a center $ j $ is moved from $ \cen_i^j $ to $ \cen_{i+1}^j $ such that there is no path between $ \cen_i^j $ and $ \cen_{i+1}^j $ in $ G_i $. In this case our analysis cannot bound the total update time of the moving centers data structure.}
The data structure can answer each query in constant time.
\end{proposition}

\begin{proof}
Our data structure maintains (1) an ES-tree of depth $ \Q $ rooted at every node that currently hosts a center; and
(2) for every node a doubly linked {\em center list} of centers by which it is covered.
Recall that a node is covered by a center if and only if the node is contained in the ES-tree of depth $ \q $ of the center. For every center $ j $ and node $x$ we keep a pointer of the node representing $x$ in the ES-tree of $ j $ to the list element representing $ j $ in the center list of $x$.

The data structure is updated as follows:
Every time we open a center $ j $ at some node $ x $ we build an ES-tree of depth $ \Q $ rooted at $ x $. Additionally we add $ j $ to the center list of all nodes covered by $ j $ and set the pointers from the ES-tree to the center lists. 

When we move a center $j$ from a node $ x $ to another node $ y $ we build an ES-tree of depth $ \Q $ rooted at $ y $ and stop maintaining the ES-tree rooted at $ x $. Additionally we use the pointers from the ES-tree rooted at $ x $ into the center lists to remove $ j $ from all the center lists of the nodes
in the ES-tree rooted $ x $. Then  we add $ j $ to the suitable center lists for all nodes in the ES-tree of $ y $
and add pointers into these lists from the ES-tree of $ y $. 

After deleting an edge we update all ES-trees of depth $ \Q $. If a node $x$ reaches a level larger than $ \Q $ in the ES-tree of $ j $, it is removed from the ES-tree of $ j $ and we use its pointer in the ES-tree to remove $ j $ from $x$'s center list. The total work of this operation is proportional to the amount of time spent updating all the ES-trees.

To answer a distance query for center $j$ and node $x$ we return the distance of $x$ to the root of  the ES-tree of $j$.
To answer a find center query for node $u$ we simply return the first element of the center list of node $u$.
Both query operations take constant worst-case time.

We now bound the running time for maintaining the ES-trees of the centers.
First, we bound the initialization costs.
For each open-operation and each move operation of a center $ j $ we spend time $ O(m) $ for (re)initializing the ES-tree of center $ j $.
This leads to a total running time of $ O (\opens m + \nummoves m) $ for all initializations, where $ \nummoves $ is the total number of move operations.
Note that we can ignore every move operation that does not change the location of any center.
Every other move operation increases the total moving distance by at least $ 1 $.
Therefore we can charge the initialization cost of $ O (m) $ for moving a center to the moving distance, which means that the quantity $ O (\opens m + \nummoves m) $ will be absorbed by $ O ((\opens  \Q + \moves) m) $, the projected total update time.

We are left to bound the time spent for processing the deletions in the ES-trees of centers.
For every center $ j $, we denote by $ T (i, j) $ the running time for processing the $i$-th edge deletion in the ES-tree of center~$ j $.
Furthermore, we denote by $ o_j $ the index of the delete operation before which the center $ j $ has been opened, i.e., center $ j $ was opened before the $ o_j $-th and after the $ (o_j - 1) $-th delete operation.
Remember that the set of centers never shrinks, i.e., $ \C_i \subseteq \C_k $ for every $ 0 \leq i \leq k $.
We will show that $ \sum_{0 < i \leq k} \sum_{j \in \C_i} T (i, j) = O ((\opens \Q + \moves) m) $.

The basic idea is that the time spent up to deletion $ i $ for node $ x $ in the ES-tree of center $ j $ is $ O (\deg_{G_0} (x) \cdot \dist_{G_i} (x, \cen_i^j)) $.
After a move operation the distance of $ x $ to the new root $ \cen_{i+1}^j $ is at most $ \dist_{G_i} (\cen_{i+1}^j, \cen_i^j) $ smaller than the previous distance and, thus, at most $ \sum_{0 \leq i < k} \deg_{G_0} (x) \cdot \dist_{G_i} (\cen_i^j, \cen_{i+1}^j)  $ additional time will be spent updating $ x $ in the ES-tree of center $ j $.

Consider the $ (i+1) $-th edge deletion and let $ j \in \C_{i+1} $ be a center.
By \Cref{cor:ES_running_time_distance_increase}, the total time for processing this deletion in the ES-tree of center $ j $ is
\begin{equation}
T (i+1, j) = \sum_{x \in V} \deg_{G_0} (x) \cdot \left( \min \left( \dist_{G_{i+1}} (x, \cen_{i+1}^j), \Q \right) - \min \left(\dist_{G_i} (x, \cen_{i+1}^j), \Q \right) \right) \label{eq:running_time_center_deletion}
\end{equation}

If $ j $ has already been opened before the $i$-th edge deletion, then, by the triangle inequality, we get $
\dist_{G_i} (x, \cen_i^j) \leq \dist_{G_i} (x, c_{i+1}^j) + \dist_{G_i} (\cen_i^j, c_{i+1}^j) $, which is equivalent to $ \dist_{G_i} (x, c_{i+1}^j) \geq \dist_{G_i} (x, \cen_i^j) - \dist_{G_i} (\cen_i^j, c_{i+1}^j) $ .
It follows that
\begin{align*}
\min \left(\dist_{G_i} (x, \cen_{i+1}^j), \Q \right) &\geq \min \left(\dist_{G_i} (x, \cen_i^j) - \dist_{G_i} (c_{i+1}^j, \cen_i^j), \Q \right) \\
&\geq \min \left(\dist_{G_i} (x, \cen_i^j), \Q \right) - \dist_{G_i} (\cen_i^j, c_{i+1}^j) \, .
\end{align*}
Therefore we get
\begin{align*}
T (i+1, j) &\leq \sum_{x \in V} \deg_{G_0} (x) \cdot \left( \min \left( \dist_{G_{i+1}} (x, \cen_{i+1}^j), \Q \right) - \min \left(\dist_{G_i} (x, \cen_i^j), \Q \right) \right. \\
 &~~~~+ \left. \dist_{G_i} (\cen_i^j, c_{i+1}^j) \right) \\
 &= \sum_{x \in V} \deg_{G_0} (x) \cdot \left( \min \left( \dist_{G_{i+1}} (x, \cen_{i+1}^j), \Q \right) - \min \left(\dist_{G_i} (x, \cen_i^j), \Q \right) \right) \\
 &~~~~+ \sum_{x \in V} \deg_{G_0} (x) \cdot \dist_{G_i} (\cen_i^j, c_{i+1}^j) \\
 &\leq \sum_{x \in V} \deg_{G_0} (x) \cdot \left( \min \left( \dist_{G_{i+1}} (x, \cen_{i+1}^j), \Q \right) - \min \left(\dist_{G_i} (x, \cen_i^j), \Q \right) \right) \\
 &~~~~+ 2 m \cdot \dist_{G_i} (\cen_i^j, c_{i+1}^j) \, .
\end{align*}
Summing up all $ T (i, j) $ for every deletion $ i > o_j $ gives a telescoping sum that results in the following term:
\begin{multline*}
\sum_{o_j < i \leq k} T (i, j) = \sum_{x \in V} \deg_{G_0} (x) \cdot \min \left( \dist_{G_k} (x, \cen_k^j), \Q \right) \\ - \sum_{x \in V} \deg_{G_0} (x) \cdot \min \left(\dist_{G_{o_j}} (x, \cen_{o_j}), \Q \right)
+ \sum_{o_j \leq i < k} 2 m \cdot \dist_{G_i} (\cen_i^j, c_{i+1}^j) \, .
\end{multline*}

Consider now a center $ j $ and the $ o_j $-th edge deletion.
By \eqref{eq:running_time_center_deletion} we can bound the running time $ T (o_j, j) $ as follows:
\begin{equation*}
T (o_j, j) \leq \sum_{x \in V} \deg_{G_0} (x) \cdot \min \left( \dist_{G_{o_j}} (x, \cen_{o_j}^j), \Q \right) \, .
\end{equation*}
Therefore the total time for maintaining the moving ES-tree of center $ j $ over all deletions is
\begin{align*}
\sum_{o_j \leq i \leq k} T (i, j) &= T (o_j, j)  + \sum_{o_j < i \leq k} T (i, j) \\
 & \leq \sum_{x \in V} \deg_{G_0} (x) \cdot \min \left( \dist_{G_k} (x, \cen_k^j), \Q \right) + \sum_{o_j \leq i < k} 2 m \cdot \dist_{G_i} (\cen_i^j, c_{i+1}^j) \\
 & \leq \sum_{x \in V} \deg_{G_0} (x) \cdot \Q + \sum_{o_j \leq i < k} 2 m \cdot \dist_{G_i} (\cen_i^j, c_{i+1}^j) \\
 & \leq 2 m \Q + \sum_{o_j \leq i < k} 2 m \cdot \dist_{G_i} (\cen_i^j, c_{i+1}^j) \, .
\end{align*}
By summing up this quantity over all centers and switching the order of the double sum, we arrive at the following total time:
\begin{align*}
\sum_{0 < i \leq k} \sum_{j \in \C_i} T (i, j) &= \sum_{j \in \C_k} \sum_{o_j \leq i \leq k} T (i, j) \\
 &\leq \sum_{j \in \C_k} 2 m \Q + \sum_{j \in \C_k} \sum_{o_j \leq i < k} 2m \cdot \dist_{G_i} (\cen_i^j, \cen_{i+1}^j) \\
 &= \sum_{j \in \C_k} 2 m \Q + 2 m \cdot \sum_{0 \leq i < k} \sum_{j \in \C_i} \dist_{G_i} (\cen_i^j, \cen_{i+1}^j) \\
 &= 2 \opens m \Q + 2 m \moves
\end{align*}
Therefore the total update time for maintaining the moving centers data structure over all operations is $ O ((\opens \Q + \moves) m) $.
\end{proof}

\subsection{A Deterministic Center Cover Data Structure (\texttt{CenterCover})}\label{sec:algorithm_maintaining_centers}

In this section, we present a deterministic algorithm for maintaining the center cover data structure \CenterCover, as defined in \Cref{def:CenterCover}. That is, for parameters $\q$ and~$\Q$, we show that we can maintain a set of centers with the following two properties.
First, all nodes in a connected component of size at least $\q$ are {\em covered} by some center; i.e., each of them is in distance at most $\q$ to some center.
Second, for every center, the distance to every node up to distance $ \Q $ is maintained.
This section is devoted to proving the following.

\begin{proposition}[Main result of \Cref{sec:algorithm_maintaining_centers}]\label{thm:deterministic RodittyZwick}
For every cover range parameter~$ \q $ and every distance range parameter~$ \Q $ such that $ \q \leq \Q $, there is a center cover data structure with a total \emph{deterministic} update time of $ O (m n \Q / \q) $ and constant query time.
\end{proposition}

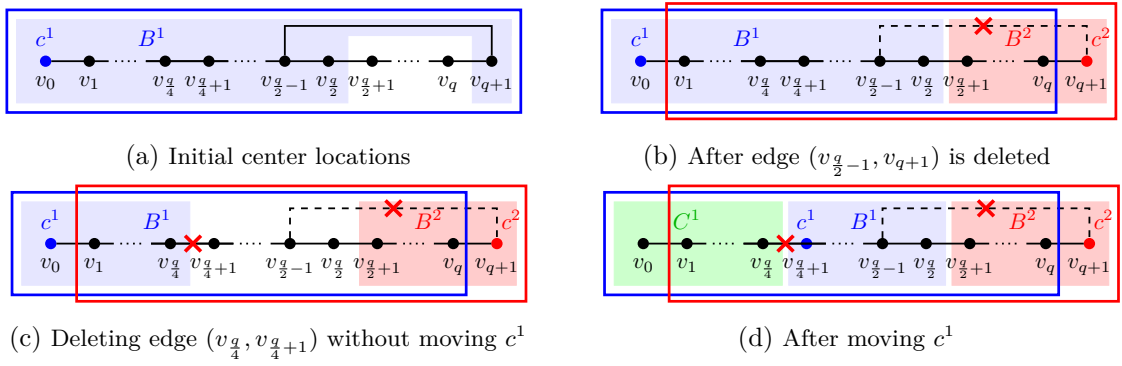
\begin{figure}
\centering
\begin{subfigure}[t]{0.49\linewidth}
\hspace{-1.05cm}%
\scalebox{0.85}{
\pgfdeclarelayer{bbackground}
\pgfdeclarelayer{background}
\pgfdeclarelayer{foreground}     
\pgfsetlayers{bbackground,background,main,foreground}  

\tikzstyle{vertex}=[circle,fill=black,minimum size=5pt,inner sep=0pt,outer sep=0pt]
\tikzstyle{vertexb}=[circle,fill=blue,minimum size=5pt,inner sep=0pt,outer sep=0pt]
\tikzstyle{shortest-path} = [draw,thick]
\tikzstyle{shortest-path-omitted} = [draw,thick,dotted]
\tikzstyle{hop-path-small} = [draw,ultra thick,-,color=blue]
\tikzstyle{hop-path-big} = [draw,ultra thick,-,color=red]

\begin{tikzpicture}[node distance=0.1mm and 5mm, skip loop/.style={to path={-- ++(0,#1)|- (\tikztotarget)}}]

\node[vertexb] (v0) at (0,0) {};
\node[vertex] (v1) [right= of v0] {};

\node[vertex] (vq4) [right=1cm of v1] {};
\node[vertex] (vq4+1) [right= of vq4] {};

\node[vertex] (vq2-1) [right=1cm of vq4+1] {};
\node[vertex] (vq2) [right= of vq2-1] {};
\node[vertex] (vq2+1) [right= of vq2] {};

\node[vertex] (vq) [right=1cm of vq2+1] {};
\node[vertex] (vq+1) [right=of vq] {};

\node (c1) [above= of v0] {\small \textcolor{blue}{$c^1$}};
\node [right=1cm of c1] {\small \textcolor{blue}{$B^1$}};

\node [below= of v0] {\small $v_0$};
\node [below= of v1] {\small $v_1$};
\node [below= of vq4] {\small $v_{\frac{q}{4}}$};
\node [below= of vq4+1] {\small $v_{\frac{q}{4}+1}$};
\node [below= of vq2-1] {\small $v_{\frac{q}{2}-1}$};
\node [below= of vq2] {\small $v_{\frac{q}{2}}$};
\node [below= of vq2+1] {\small $v_{\frac{q}{2}+1}$};
\node [below= of vq] {\small $v_q$};
\node [below= of vq+1] {\small $v_{q+1}$};

\path[shortest-path] (v0) -- ($(v1) + (0.3,0)$);
\path[shortest-path-omitted] ($(v1) + (0.4,0)$) -- ($(vq4) + (-0.4,0)$);
\path[shortest-path] ($(vq4) + (-0.3,0)$) -- ($(vq4+1) + (0.3,0)$);
\path[shortest-path-omitted] ($(vq4+1) + (0.4,0)$) -- ($(vq2-1) + (-0.4,0)$);
\path[shortest-path] ($(vq2-1) + (-0.3,0)$) -- ($(vq2+1) + (0.3,0)$);
\path[shortest-path] ($(vq4) + (-0.3,0)$) -- ($(vq4+1) + (0.3,0)$);
\path[shortest-path-omitted] ($(vq2+1) + (0.4,0)$) -- ($(vq) + (-0.4,0)$);
\path[shortest-path] ($(vq) + (-0.3,0)$) -- (vq+1);
\path[shortest-path] (vq2-1) -- ($(vq2-1) + (0,0.55)$) -- ($(vq+1) + (0,0.55)$) -- (vq+1);

\begin{pgfonlayer}{foreground}
	\draw [very thick, draw=blue] ($(v0) + (-0.6,-0.8)$) rectangle ($(vq+1) + (0.45,0.8)$);
\end{pgfonlayer}

\begin{pgfonlayer}{background}    
	\draw [draw=blue!10,fill=blue!10] ($(v0) + (-0.45,-0.65)$) rectangle ($(vq2) + (0.3,0.65)$);
	\draw [draw=blue!10,fill=blue!10] ($(vq2) + (0.3,0.4)$) rectangle ($(vq+1) + (0.3,0.65)$);
	\draw [draw=blue!10,fill=blue!10] ($(vq+1) + (-0.3,-0.65)$) rectangle ($(vq+1) + (0.3,0.65)$);
\end{pgfonlayer}

\begin{pgfonlayer}{bbackground}    
	\draw [draw=white] (-0.75,-1) rectangle (7.5,1);
\end{pgfonlayer}

\end{tikzpicture}
}
\caption{\footnotesize Initial center locations}
\label{fig:deterministic1}
\end{subfigure}
\begin{subfigure}[t]{0.49\linewidth}
\hspace{-1.05cm}%
\scalebox{0.85}{
\pgfdeclarelayer{bbackground}
\pgfdeclarelayer{background}
\pgfdeclarelayer{foreground}     
\pgfsetlayers{bbackground,background,main,foreground}  

\tikzstyle{vertex}=[circle,fill=black,minimum size=5pt,inner sep=0pt,outer sep=0pt]
\tikzstyle{vertexb}=[circle,fill=blue,minimum size=5pt,inner sep=0pt,outer sep=0pt]
\tikzstyle{vertexr}=[circle,fill=red,minimum size=5pt,inner sep=0pt,outer sep=0pt]
\tikzstyle{shortest-path} = [draw,thick]
\tikzstyle{shortest-path-omitted} = [draw,thick,dotted]
\tikzstyle{hop-path-small} = [draw,ultra thick,-,color=blue]
\tikzstyle{hop-path-big} = [draw,ultra thick,-,color=red]

\tikzset{cross/.style={cross out, draw=black, ultra thick, minimum size=5*(#1-\pgflinewidth), inner sep=0pt, outer sep=0pt}, 
cross/.default={3pt}}

\begin{tikzpicture}[node distance=0.1mm and 5mm, skip loop/.style={to path={-- ++(0,#1)|- (\tikztotarget)}}]

\node[vertexb] (v0) at (0,0) {};
\node[vertex] (v1) [right= of v0] {};

\node[vertex] (vq4) [right=1cm of v1] {};
\node[vertex] (vq4+1) [right= of vq4] {};

\node[vertex] (vq2-1) [right=1cm of vq4+1] {};
\node[vertex] (vq2) [right= of vq2-1] {};
\node[vertex] (vq2+1) [right= of vq2] {};

\node[vertex] (vq) [right=1cm of vq2+1] {};
\node[vertexr] (vq+1) [right=of vq] {};

\node (c1) [above= of v0] {\small \textcolor{blue}{$c^1$}};
\node [right=1cm of c1] {\small \textcolor{blue}{$B^1$}};

\node (c2) [right=6.55cm of c1] {\small \textcolor{red}{$c^2$}};
\node [right=5.2cm of c1] {\small \textcolor{red}{$B^2$}};

\node [below= of v0] {\small $v_0$};
\node [below= of v1] {\small $v_1$};
\node [below= of vq4] {\small $v_{\frac{q}{4}}$};
\node [below= of vq4+1] {\small $v_{\frac{q}{4}+1}$};
\node [below= of vq2-1] {\small $v_{\frac{q}{2}-1}$};
\node [below= of vq2] {\small $v_{\frac{q}{2}}$};
\node [below= of vq2+1] {\small $v_{\frac{q}{2}+1}$};
\node [below= of vq] {\small $v_q$};
\node [below= of vq+1] {\small $v_{q+1}$};

\path[shortest-path] (v0) -- ($(v1) + (0.3,0)$);
\path[shortest-path-omitted] ($(v1) + (0.4,0)$) -- ($(vq4) + (-0.4,0)$);
\path[shortest-path] ($(vq4) + (-0.3,0)$) -- ($(vq4+1) + (0.3,0)$);
\path[shortest-path-omitted] ($(vq4+1) + (0.4,0)$) -- ($(vq2-1) + (-0.4,0)$);
\path[shortest-path] ($(vq2-1) + (-0.3,0)$) -- ($(vq2+1) + (0.3,0)$);
\path[shortest-path] ($(vq4) + (-0.3,0)$) -- ($(vq4+1) + (0.3,0)$);
\path[shortest-path-omitted] ($(vq2+1) + (0.4,0)$) -- ($(vq) + (-0.4,0)$);
\path[shortest-path] ($(vq) + (-0.3,0)$) -- (vq+1);
\path[shortest-path,dashed] (vq2-1) -- ($(vq2-1) + (0,0.55)$) -- ($(vq+1) + (0,0.55)$) node [align=center,midway,cross,red,solid] {} -- (vq+1);

\begin{pgfonlayer}{background}    
	\draw [draw=blue!10,fill=blue!10] ($(v0) + (-0.45,-0.65)$) rectangle ($(vq2) + (0.3,0.65)$);
	\draw [draw=red!20,fill=red!20] ($(vq2) + (0.4,-0.65)$) rectangle ($(vq+1) + (0.3,0.65)$);
\end{pgfonlayer}

\begin{pgfonlayer}{foreground}
	\draw [very thick, draw=blue] ($(v0) + (-0.6,-0.8)$) rectangle ($(vq) + (0.2,0.8)$);
	\draw [very thick, draw=red] ($(v0) + (0.4,-0.9)$) rectangle ($(vq+1) + (0.45,0.9)$);
\end{pgfonlayer}

\begin{pgfonlayer}{bbackground}    
	\draw [draw=white] (-0.75,-1) rectangle (7.5,1);
\end{pgfonlayer}

\end{tikzpicture}
}
\caption{\footnotesize After edge $(v_{\frac{q}{2}-1}, v_{q+1})$ is deleted}
\label{fig:deterministic2}
\end{subfigure}
\begin{subfigure}[t]{0.49\linewidth}
\hspace{-1.2cm}%
\scalebox{0.85}{
\pgfdeclarelayer{bbackground}
\pgfdeclarelayer{background}
\pgfdeclarelayer{foreground}     
\pgfsetlayers{bbackground,background,main,foreground}  

\tikzstyle{vertex}=[circle,fill=black,minimum size=5pt,inner sep=0pt,outer sep=0pt]
\tikzstyle{vertexb}=[circle,fill=blue,minimum size=5pt,inner sep=0pt,outer sep=0pt]
\tikzstyle{vertexr}=[circle,fill=red,minimum size=5pt,inner sep=0pt,outer sep=0pt]
\tikzstyle{shortest-path} = [draw,thick]
\tikzstyle{shortest-path-omitted} = [draw,thick,dotted]
\tikzstyle{hop-path-small} = [draw,ultra thick,-,color=blue]
\tikzstyle{hop-path-big} = [draw,ultra thick,-,color=red]

\tikzset{cross/.style={cross out, draw=black, ultra thick, minimum size=5*(#1-\pgflinewidth), inner sep=0pt, outer sep=0pt}, 
cross/.default={3pt}}

\begin{tikzpicture}[node distance=0.1mm and 5mm, skip loop/.style={to path={-- ++(0,#1)|- (\tikztotarget)}}]

\node[vertexb] (v0) at (0,0) {};
\node[vertex] (v1) [right= of v0] {};

\node[vertex] (vq4) [right=1cm of v1] {};
\node[vertex] (vq4+1) [right= of vq4] {};

\node[vertex] (vq2-1) [right=1cm of vq4+1] {};
\node[vertex] (vq2) [right= of vq2-1] {};
\node[vertex] (vq2+1) [right= of vq2] {};

\node[vertex] (vq) [right=1cm of vq2+1] {};
\node[vertexr] (vq+1) [right=of vq] {};

\node (c1) [above= of v0] {\small \textcolor{blue}{$c^1$}};
\node [right=1cm of c1] {\small \textcolor{blue}{$B^1$}};

\node (c2) [right=6.55cm of c1] {\small \textcolor{red}{$c^2$}};
\node [right=5.2cm of c1] {\small \textcolor{red}{$B^2$}};

\node [below= of v0] {\small $v_0$};
\node [below= of v1] {\small $v_1$};
\node [below= of vq4] {\small $v_{\frac{q}{4}}$};
\node [below= of vq4+1] {\small $v_{\frac{q}{4}+1}$};
\node [below= of vq2-1] {\small $v_{\frac{q}{2}-1}$};
\node [below= of vq2] {\small $v_{\frac{q}{2}}$};
\node [below= of vq2+1] {\small $v_{\frac{q}{2}+1}$};
\node [below= of vq] {\small $v_q$};
\node [below= of vq+1] {\small $v_{q+1}$};

\path[shortest-path] (v0) -- ($(v1) + (0.3,0)$);
\path[shortest-path-omitted] ($(v1) + (0.4,0)$) -- ($(vq4) + (-0.4,0)$);
\path[shortest-path] ($(vq4) + (-0.3,0)$) -- ($(vq4+1) + (0.3,0)$);
\path[shortest-path-omitted] ($(vq4+1) + (0.4,0)$) -- ($(vq2-1) + (-0.4,0)$);
\path[shortest-path] ($(vq2-1) + (-0.3,0)$) -- ($(vq2+1) + (0.3,0)$);
\path[shortest-path] ($(vq4) + (-0.3,0)$) -- ($(vq4+1) + (0.3,0)$);
\path[shortest-path-omitted] ($(vq2+1) + (0.4,0)$) -- ($(vq) + (-0.4,0)$);
\path[shortest-path] ($(vq) + (-0.3,0)$) -- (vq+1);
\path[shortest-path,dashed] (vq2-1) -- ($(vq2-1) + (0,0.55)$) -- ($(vq+1) + (0,0.55)$) node [align=center,midway,cross,red,solid] {} -- (vq+1);

\begin{pgfonlayer}{background}    
	\draw [draw=blue!10,fill=blue!10] ($(v0) + (-0.45,-0.65)$) rectangle ($(vq4) + (0.3,0.65)$);
	\draw [draw=red!20,fill=red!20] ($(vq2) + (0.4,-0.65)$) rectangle ($(vq+1) + (0.3,0.65)$);
\end{pgfonlayer}

\begin{pgfonlayer}{foreground}
	\draw [very thick, draw=blue] ($(v0) + (-0.6,-0.8)$) rectangle ($(vq) + (0.2,0.8)$);
	\draw [very thick, draw=red] ($(v0) + (0.4,-0.9)$) rectangle ($(vq+1) + (0.45,0.9)$);
	
	\node [cross,red] at ($(vq4) + (0.35,0)$) {};
\end{pgfonlayer}

\begin{pgfonlayer}{bbackground}    
	\draw [draw=white] (-0.75,-1) rectangle (7.5,1);
\end{pgfonlayer}

\end{tikzpicture}
}
\caption{\footnotesize Deleting edge $(v_{\frac{q}{4}}, v_{\frac{q}{4}+1})$ without moving $\cen^1$}
\label{fig:deterministic3}
\end{subfigure}
\begin{subfigure}[t]{0.49\linewidth}
\hspace{-1.125cm}%
\scalebox{0.85}{
\colorlet{darkgreen}{green!70!black}

\pgfdeclarelayer{bbackground}
\pgfdeclarelayer{background}
\pgfdeclarelayer{foreground}     
\pgfsetlayers{bbackground,background,main,foreground}  

\tikzstyle{vertex}=[circle,fill=black,minimum size=5pt,inner sep=0pt,outer sep=0pt]
\tikzstyle{vertexb}=[circle,fill=blue,minimum size=5pt,inner sep=0pt,outer sep=0pt]
\tikzstyle{vertexr}=[circle,fill=red,minimum size=5pt,inner sep=0pt,outer sep=0pt]
\tikzstyle{shortest-path} = [draw,thick]
\tikzstyle{shortest-path-omitted} = [draw,thick,dotted]
\tikzstyle{hop-path-small} = [draw,ultra thick,-,color=blue]
\tikzstyle{hop-path-big} = [draw,ultra thick,-,color=red]

\tikzset{cross/.style={cross out, draw=black, ultra thick, minimum size=5*(#1-\pgflinewidth), inner sep=0pt, outer sep=0pt}, 
cross/.default={3pt}}

\begin{tikzpicture}[node distance=0.1mm and 5mm, skip loop/.style={to path={-- ++(0,#1)|- (\tikztotarget)}}]

\node[vertex] (v0) at (0,0) {};
\node[vertex] (v1) [right= of v0] {};

\node[vertex] (vq4) [right=1cm of v1] {};
\node[vertexb] (vq4+1) [right= of vq4] {};

\node[vertex] (vq2-1) [right=1cm of vq4+1] {};
\node[vertex] (vq2) [right= of vq2-1] {};
\node[vertex] (vq2+1) [right= of vq2] {};

\node[vertex] (vq) [right=1cm of vq2+1] {};
\node[vertexr] (vq+1) [right=of vq] {};

\node [above= of v1] {\small \textcolor{darkgreen}{$C^1$}};

\node (c1) [above= of vq4+1] {\small \textcolor{blue}{$c^1$}};
\node [right=0.25cm of c1] {\small \textcolor{blue}{$B^1$}};

\node (c2) [right=4.02cm of c1] {\small \textcolor{red}{$c^2$}};
\node [right=2.7cm of c1] {\small \textcolor{red}{$B^2$}};

\node [below= of v0] {\small $v_0$};
\node [below= of v1] {\small $v_1$};
\node [below= of vq4] {\small $v_{\frac{q}{4}}$};
\node [below= of vq4+1] {\small $v_{\frac{q}{4}+1}$};
\node [below= of vq2-1] {\small $v_{\frac{q}{2}-1}$};
\node [below= of vq2] {\small $v_{\frac{q}{2}}$};
\node [below= of vq2+1] {\small $v_{\frac{q}{2}+1}$};
\node [below= of vq] {\small $v_q$};
\node [below= of vq+1] {\small $v_{q+1}$};

\path[shortest-path] (v0) -- ($(v1) + (0.3,0)$);
\path[shortest-path-omitted] ($(v1) + (0.4,0)$) -- ($(vq4) + (-0.4,0)$);
\path[shortest-path] ($(vq4) + (-0.3,0)$) -- ($(vq4+1) + (0.3,0)$);
\path[shortest-path-omitted] ($(vq4+1) + (0.4,0)$) -- ($(vq2-1) + (-0.4,0)$);
\path[shortest-path] ($(vq2-1) + (-0.3,0)$) -- ($(vq2+1) + (0.3,0)$);
\path[shortest-path] ($(vq4) + (-0.3,0)$) -- ($(vq4+1) + (0.3,0)$);
\path[shortest-path-omitted] ($(vq2+1) + (0.4,0)$) -- ($(vq) + (-0.4,0)$);
\path[shortest-path] ($(vq) + (-0.3,0)$) -- (vq+1);
\path[shortest-path,dashed] (vq2-1) -- ($(vq2-1) + (0,0.55)$) -- ($(vq+1) + (0,0.55)$) node [align=center,midway,cross,red,solid] {} -- (vq+1);

\begin{pgfonlayer}{background}    
	\draw [draw=green!20,fill=green!20] ($(v0) + (-0.45,-0.65)$) rectangle ($(vq4) + (0.3,0.65)$);
	\draw [draw=blue!10,fill=blue!10] ($(vq4) + (0.4,0.65)$) rectangle ($(vq2) + (0.3,-0.65)$);
	\draw [draw=red!20,fill=red!20] ($(vq2) + (0.4,-0.65)$) rectangle ($(vq+1) + (0.3,0.65)$);
\end{pgfonlayer}

\begin{pgfonlayer}{foreground}
	\draw [very thick, draw=blue] ($(v0) + (-0.6,-0.8)$) rectangle ($(vq) + (0.2,0.8)$);
	\draw [very thick, draw=red] ($(v0) + (0.4,-0.9)$) rectangle ($(vq+1) + (0.45,0.9)$);
	
	\node [cross,red] at ($(vq4) + (0.35,0)$) {};
\end{pgfonlayer}

\begin{pgfonlayer}{bbackground}    
	\draw [draw=white] (-0.75,-1) rectangle (7.5,1);
\end{pgfonlayer}

\end{tikzpicture}
}
\caption{\footnotesize After moving $\cen^1$}
\label{fig:deterministic4}
\end{subfigure}
\caption{\small Example of our algorithm for maintaining the center cover data structure using the moving centers data structure, as in \Cref{thm:deterministic RodittyZwick}. We use $q=\q$ and, for any $j$, we let $\cen^j$ denote the location of center $ j $. Boxes filled with colors show sets $\ball^j$ and $\T^j$.
(\subref{fig:deterministic1}) shows a possible initial location of center $\cen^1$. This makes $\ball^1=\{v_0, \ldots, v_{q/2}\} \cup \{v_{q+1}\}$ and $\T^1=\emptyset$. All nodes are covered by center $\cen^1$.
(\subref{fig:deterministic2}) shows what our algorithm does when edge $(v_{q/2-1}, v_{q+1})$ is deleted. In this case, $v_{q+1}$ is not covered by $\cen^1$ anymore, so we open a center $\cen^2$ at $v_{q+1}$.
(\subref{fig:deterministic3}) shows what $\ball^1$ will look like after edge $(v_{q/4}, v_{q/4+1})$ is deleted, if we do not move center $\cen^1$. In particular, $|\ball^1\cup \T^1|<q/2$.
(\subref{fig:deterministic4}) shows what our algorithm will do after edge $(v_{q/4}, v_{q/4+1})$ is deleted  to maintain the largeness property (i.e., to make sure that $|\ball^1 \cup \T^1|\geq q/2$): it moves nodes $v_0, \ldots, v_{q/4}$ from $\ball^1$ to $\T^1$ and moves the first center from $v_0$ to $v_{q/4+1}$.}
\label{fig:deterministic}
\end{figure}

\subsubsection{High-Level Ideas}\label{sec:deterministic_algorithm_ideas}

Our algorithm will internally use the moving centers data structure from \Cref{sec:moving centers data structure} (called  \MovingCenter). It has to determine how to open and move centers in a way that ensures that at any time every node in a connected component of size at least $ \q $ is covered by some center,;i.e., its distance to the nearest center is at most $\q$.
At a high level, our algorithm is very simple (see \Cref{fig:deterministic} for an example; note that $q=\q$): For each center $j$, it maintains two sets $\ball^j$ and $\T^j$, where $\ball^j$ is always defined to be the set of nodes whose distance to center $j$ is at most $ \q-|\T^j| $. Initially, the algorithm sets $\T^j=\emptyset$ and chooses a set of centers such that all sets $\ball^j$ are disjoint (see \Cref{fig:deterministic1}). The sets $\T^j$ will never decrease during the algorithm. After an edge deletion, if some node in a large connected component (size $ \geq \q $) that is no longer covered by any center (e.g., $v_{q+1}$ in \Cref{fig:deterministic2}), then the algorithm simply opens a new center at that node. However, before doing so it has to check whether $|\ball^j\cup \T^j|<\q/2$ for some existing center $j$. (For example, after edge $(v_{q/4}, v_{q/4+1})$ is deleted as in \Cref{fig:deterministic3}, $|\ball^1\cup \T^1| = (q/4+1) < q/2 = \q/2$.) If this is the case for center $j$, it will add all nodes of $\ball^j$ to $\T^j$ and move the center $j$ to the end-node of the deleted edge that is in a different connected component than the old location of $j$. As we will show, the nodes in $\ball^j$ at the new location are {\em not} contained in $\ball^{j'}$ for any center $j' \ne j$; i.e., the invariant that all sets $\ball^j$ are disjoint remains valid. For example, in \Cref{fig:deterministic4}, the algorithm puts nodes $v_0, \ldots, v_{q/4}$ to $\T^1$ and moves $\cen^1$ to node $v_{q/4+1}$, which is the end-node of the deleted edge $(v_{q/4}, v_{q/4+1})$ that is in a connected component different from center $\cen^1$.

We now give the intuition behind this algorithm and its analysis before going into detail. Recall from \Cref{lem:existence_dynamic_centers_ds} that opening and maintaining a center together costs $O(m\q)$ time in total, and a move-operation incurs a total time of $O(m)$ per one unit moving distance. So, to get the desired $O(mn\Q / \q)$ total time bound, we will make sure that our algorithm uses a limited number of open-operations and a limited moving distance; in particular, we will make sure that
\begin{align*}
\opens = O(n/\q) ~~~\mbox{and}~~~ \moves = O(n).
\end{align*}

To guarantee that we open at most $O(n/\q)$ centers, we imagine that each node holds a coin at the beginning of the algorithm, which it can give to at most one center during the algorithm, and we require that each center must receive at least $\q/2$ coins from some nodes in the end. Clearly, this will automatically ensure that at most $2n/\q$ centers will be opened. Since the graph keeps changing, it is hard to say which node should give a coin to which center at the beginning. Instead, our algorithm will maintain two sets for each center $j$: the set $\ball^j$ of {\em borrowed} nodes from which center $j$ has borrowed coins that it might have to return, and the set $\T^j$ of {\em collected} nodes from which center $j$ has collected coins that it will never return. After all edge deletions, $j$ will hold the coins of all  nodes in $\ball^j \cup \T^j$. Our algorithm will maintain $\ball^j\cup \T^j$ with two properties:
\begin{enumerate}
\item (Largeness.) $|\ball^j\cup \T^j|\geq \q/2$ at any time (so that $j$ gets enough coins in the end).
\item (Disjointness.) $\ball^j\cup \T^j$ is disjoint from $\ball^{j'}\cup \T^{j'}$ for all centers $j\neq j'$ (so that no node gives a coin to more than one center).
\end{enumerate}

These two properties easily imply that every center will get at least $\q/2$ coins in the end---center $j$ simply collects coins from the nodes in $\ball^j\cup \T^j$; consequently, they guarantee that $\opens=O(n/\q)$, as desired.
Note that $ \ball^j $ and $\T^j$ are only introduced for the analysis; our algorithm does {\em not} need to maintain them explicitly.
If the location of center $j$ is moved from $x$ to $y$, then we say for every node $u$ on a shortest path between $x$ to $y$ that the center has been {\em moved through} $u$.
To guarantee that the total moving distance is $O(n)$, we need one more property: 
\begin{enumerate}[resume]
\item (Confinement.) The location of center $j$ is moved {\em only through nodes that are added to $\T^j$}.
\end{enumerate}
By the disjointness property, no two centers are moved along the same node if the confinement property is satisfied. So, the total moving distance will be $\moves=O(n)$, as desired.  

It is left to check whether the algorithm we have sketched earlier satisfies all three properties above. The largeness property can be guaranteed using the fact that after every edge deletion, the algorithm will move every center $j$ such that $|\ball^j\cup \T^j|< \q/2$ to a new node; the only nonobvious property we have to prove is that $\ball^j$ will be large enough after the move, and the key to this proof is the fact that the connected component containing the new location of center $j$ has size at least $\q/2-|\T^j|$.
For the disjointness property, we will show two further properties. 
\begin{enumerate}[label=(P\arabic{*})]
\item (Initial-disjointness) When we open a center $j$, $\ball^j$ is disjoint from $\ball^{j'}\cup \T^{j'}$ for all other centers $j'$. \label{inv:disj inv 2}
\item (Shrinking) We never add any node to $\ball^j\cup \T^j$. (For example, $\ball^1\cup \T^1$ in \Cref{fig:deterministic1} is a subset of $\ball^1\cup \T^1$ in \Cref{fig:deterministic4}.) \label{inv:disj inv 1}
\end{enumerate}
These two properties are sufficient to guarantee the disjointness property because if two sets $\ball^j\cup \T^j$ and $\ball^{j'}\cup \T^{j'}$ are disjoint at the beginning (by Property \ref{inv:disj inv 2}), they will remain disjoint if we never add a node to them (by Property \ref{inv:disj inv 1}). The shrinking property (Property \ref{inv:disj inv 1}) can be checked simply by observing the behavior of the algorithm (see \Cref{lem:inter_disjointness_subset_property} for details). To show the initial-disjointness property (Property~\ref{inv:disj inv 2}), we use the fact that $j$ is of distance at least $\q$ from other centers when $j$ is opened, which implies that $\ball^j\cap \ball^{j'}=\emptyset$. Additionally, we will prove that $\T^{j'}$ contains only nodes in connected components of size less than $\q$, whereas any new center $ j $ is opened in a connected component of size at least $\q$. This implies that $\ball^j\cap \T^{j'}=\emptyset$ when $j$ is opened.

Finally, for the confinement property, just observe that before the algorithm moves a center $j$, it puts all nodes in the connected component containing the center $j$ to $\T^j$ and moves $j$ to a node outside of this connected component.
For example, in \Cref{fig:deterministic4} the algorithm puts nodes $v_1, \ldots, v_{q/4}$ to $\T^1$ before moving the first center through  $v_1, \ldots, v_{q/4}$  to $v_{q/4+1}$.

\subsubsection{Algorithm Description}

\begin{algorithm2e}
\thisfloatpagestyle{empty} 
\caption{\texttt{CenterCover} (Deterministic Center Cover Data Structure)}
\label{alg:deterministic_algorithm}

\tcp{\textrm{Given a decremental graph $\cG = (G_i)_{0 \leq i \leq k} $ and integers $\q$ and $\Q$, this data structure maintains a set of centers such that every node (that is in a connected component of size at least $ \q $) has distance at most $\q$ to at least one center and we can query the distance between a center and a node if their distance is at most $\Q$ (otherwise, we will get $\infty$ in return). It allows four operations: \Initialize, \Delete, \FindCenter and \Distance, as defined in \Cref{def:CenterCover}.}}

\BlankLine

\Procedure{\CenterCoverGreedyOpen{}}{
	Let $ G_i $ denote the current graph\;

	\ForEach{node $ x $}{
	    	\tcp{\textrm{The if-statement checks if $x$ is \emph{not} covered by a center and the size of the connected component containing it is larger than $\q$. See \Cref{rem:detail of initialization} for the implementation detail.}}
		\If{\FindCenter{$x$} $ = \bot $ \KwAnd $ |\comp_{G_{i}} (x)| \geq \q $}{\label{line:component size}
			\tcp{\textrm{tell moving centers data structure to open new center at $ x $. Let $j$ be the index of this center.}}
			$ j \gets $ \MCopen{$x$}\;
			Set $ \T^j \gets \emptyset $, $ \radius^j \gets \q / 2 $, and $ \cen^j \gets x $\; \label{line:initial T r c}
		}
	}
}


\tcp{\textrm{Parameters: Initial version $G_0$ of decremental graph $\cG$, integers $\q$ and $\Q$.}}
\Procedure{\CenterCoverInitialize{$G_0$, $\q$, $\Q$}}{ 
	\tcp{\textrm{Initialize the moving centers data structure (see \Cref{def:moving center})}}
	\MCinitialize{$G_0$, $\q$, $\Q$}\;
	\GreedyOpen{}\;
}


\Procedure(\tcp*[f]{\textrm{Parameter: Node $v$.}}){\CenterCoverFindCenter{$v$}}{
	\Return \MCfindCenter{$v$}\;
}


\tcp{\textrm{Parameters: Center index $j$ and node $v$.}}
\Procedure{\CenterCoverDistance{$j$, $v$}}{
	\Return \MCdistance{$j$, $v$}\;
}


\tcp{\textrm{Parameter: $ (i+1) $-th deleted edge $ (u, v) $.}}
\Procedure{\CenterCoverDelete{$ u, v $}}{
	Let $ G_i $ denote the graph before deleting $ (u, v) $ and let $ G_{i+1} $ denote the graph afterwards.\;
	\tcp{\textrm{Find a center $j$ for which the connected component containing it becomes smaller than $\radius^j$. See \Cref{rem:detail of finding bad component} for how to find such a center $j$. (Actually, there will be at most one such center, see \Cref{lem:at_most_one_center_moved}.)}}
	Find a center $ j $ such that $ | \comp_{G_{i+1}} (\cen^j) | < \radius^j $.\; \label{line:find bad component}

	\If{such a center $j$ exists}{
		\tcp{\textrm{Move $j$ to either $u$ or $v$ depending on who is in a different connected component than $ \cen^j $.}} 
		\leIf{$ u $ and $ \cen^j $ are {\em not} connected in $ G_{i+1} $}{$ y \gets u $}{$ y \gets v $} \label{line:define y} 
		Set $ \T^j \gets \T^j \cup \comp_{G_{i+1}} (\cen^j) $, $ \radius^j \gets \radius^j - | \comp_{G_{i+1}} (\cen^j) | $, and $ \cen^j \gets y $\; \label{line:update variables after move}
		\MCmove{$j$, $y$}\; \label{line:notify moving centers move}
	}
	\tcp{\textrm{Report edge deletion to moving centers data structure (\Cref{def:moving center}).}}
	\MCdelete{$u$, $v$}\; \label{line:notify moving centers delete}
	\GreedyOpen{}\;
}

\end{algorithm2e}

Our algorithm is outlined in \Cref{alg:deterministic_algorithm}.
For each center $j$, the algorithm maintains its location $\cen^j$, which could change over time since centers can be moved. In addition, it also maintains the set~$\T^j$ and the number~$\radius^j$, which are set to $\emptyset$ and $\q/2$, respectively, when center $j$ is opened.
The intended value of $\radius^j$ is $\radius^j=\q/2-|\T^j|$, and the algorithm always updates $ \radius^j $ in a way that this is ensured.
The algorithm also uses the moving centers data structure (denoted by \MovingCenter and explained in \Cref{sec:moving centers data structure}) to maintain the distance between each center $j$ to other nodes in the graph, up to distance $\Q$. This helps us to implement \CenterCoverFindCenter and \CenterCoverDistance queries in a straightforward way: the algorithm just invokes the same queries from the moving centers data structure. 

Initially, on $G_0$ (i.e., before the graph changes), our algorithm initializes the moving centers data structure by opening centers in a greedy manner: as long as there is a node $x$ that is not covered by any center, it opens a center at $x$.
This process will also be used every time an edge is deleted, to make sure that every node remains covered by a center. Procedure \CenterCoverGreedyOpen proceeds as follows. For every node $x$, it checks whether $x$ is {\em not covered}; this is the case if \CenterCoverFindCenter{$x$} returns $\bot$ and the size of the connected component containing $x$ is at least $\q$ (we refer the reader to \Cref{rem:detail of initialization} for how to compute the size of this component). If $x$ is not covered, the algorithm opens a center at $x$, stores the index $ j $ of this new center, and initializes the values of $\T^j$, $\radius^j$, and $\cen^j$, as in Line~\ref{line:initial T r c}. This completes the \GreedyOpen procedure. 

The main work of \Cref{alg:deterministic_algorithm} lies in the \Delete operation, since it has to make sure that all nodes are still covered by some centers after the deletion. Procedure \CenterCoverDelete proceeds as follows. Let us assume that the $ (i+1) $-th edge $(u, v)$ is deleted from $G_i$, and let $G_{i+1}$ denote the resulting graph.  First, the procedure checks whether there is any center $j$ that is in a large component in $G_i$ and in a small connected component in $G_{i+1}$; i.e., the size of the connected component of $ \cen^j $ is at least $\radius^j$ in $G_i$ and less than $\radius^j$ in $G_{i+1}$ (see Line~\ref{line:find bad component} of \Cref{alg:deterministic_algorithm}).
Next, if such a center $j$ in a small connected component exists (in fact, we will show that there exists at most one such $j$; see Lemma~\ref{lem:at_most_one_center_moved}), we will {\em move} $j$ to a different component and update the values of $\T^j$, $\radius^j$, and $\cen^j$. It is crucial in our analysis that $j$ must be moved carefully. In particular, we will move $j$ to either $u$ or $v$, depending on which node is in a {\em different} component from $\cen^j$, the current location of $j$. (Note that one of $u$ and $v$ will be in the same connected component as $j$ and the other will be in a different component.) We use a variable $y\in \{u, v\}$ to refer to the new location to which we move center $j$ (see Line~\ref{line:define y}). We then update the values of $\T^j$, $\radius^j$, and $\cen^j$. In particular, we put {\em all} nodes in the connected component that previously contained center $j$ (before we move it to $y$) into $\T^j$ and update $\radius^j$ to $\q/2-|\T^j|$ and $\cen^j$ to $y$. 
Then we report the move of center $ j $ to $ y $ to the moving centers data structure.
Afterward, we report the deletion of the edge $ (u, v) $ to the moving centers data structure so that it updates the distances between centers and nodes to the new distances in $G_{i+1}$.
Finally, we execute the \CenterCoverGreedyOpen procedure to make sure that every node remains covered: if there is a node $ x $ that is not covered, we open a center at $ x $. This completes the deletion operation.

\subsubsection{Analysis}

The correctness of \Cref{alg:deterministic_algorithm} is immediate.
As the procedure \GreedyOpen is called after every edge deletion, every node in a connected component of size at least $ \q $ will always be covered.
In the following we analyze the running time of Algorithm~\ref{alg:deterministic_algorithm}.

Our main task is to bound the running time of the moving centers data structure internally used by the algorithm.
In particular we want to use the running time bound stated in \Cref{lem:existence_dynamic_centers_ds} which requires us to bound the number $ \opens $ of open-operations performed by the algorithm and the total moving distance $ \moves $.
As outlined in \Cref{sec:deterministic_algorithm_ideas} we assign to each center $ j $ the set $ \ball^j \cup \T^j $.
The set $ \T^j $ contains all nodes of connected components in which the center $ j $ once was located, as shown in \Cref{alg:deterministic_algorithm}.
The set $ \ball^j $ is the set of all nodes that are at distance at most $ \radius^j $ from the center $ j $ in the current graph.
We first show that the sets $ \ball^j \cup \T^j $ fulfill two properties: disjointness and largeness.
Disjointness says that for all centers $ j \neq j' $ the sets $ \ball^j \cup \T^j $ and $ \ball^{j'} \cup \T^{j'} $ are disjoint.
Largeness says that the set $ \ball^j \cup \T^j $ has size at least $ \q / 2 $ for each center $ j $.
Using these two properties, we will prove that there are at most $ \opens = O(n / \q) $ open-operations and that the total moving distance is $ \moves = O(n) $.
These bounds will then allow us to obtain a total update time of $ O (m n \Q / \q) $ for the moving centers data structure used by \Cref{alg:deterministic_algorithm}.
Afterward we will show that all other operations of the algorithm can also be carried out within this total update time.
To make our arguments precise we will use the following notation.
\begin{definition}
Let $ \cG = (G_i)_{0 \leq i \leq k} $ be a decremental graph for which Algorithm~\ref{alg:deterministic_algorithm} maintains a center cover data structure.
The graph undergoes a sequence of $ k $ deletions, and by $ G_i $ we denote the graph after the $i$-th deletion.
We use the following notation:
\begin{itemize}
\item For all nodes $ x $ and $ y $, we denote by $ \dist_i (x, y) = \dist_{G_i} (x, y) $ the distance between $ x $ and $ y $ in the graph $ G_i $.
\item For every node $ x $, we denote by $ \comp_i (x) = \comp_{G_i} (x) $ the nodes in the connected component of $ x $ in the graph $ G_i $.
\item For every center $ j $, we denote by $ \cen_i^j $, $ \radius_i^j $, and $ \T_i^j $ the values of $ \cen^j $, $ \radius^j $, and $ \T^j $ after the algorithm has processed the $i$-th deletion, respectively (equivalently: the values before the $ (i+1) $-th deletion).
\item For every center $ j $, we define the set $ \ball_i^j $ by $ \ball_i^j = \{ x \in V \mid \dist_i (\cen_i^j, x) \leq \radius_i^j \} $; i.e., $ \ball_i^j $ is the set of all nodes that are within distance $ \radius_i^j $ to the location $ \cen_i^j $ of center $ j $ in the graph $ G_i $.
\end{itemize}
\end{definition}

\paragraph*{Preliminary Observations.}
We first state some simple observations that will be helpful later on.
\begin{observation}\label{lem:size_component_ball} 
Let $ x $ be a node, let $ i \leq k $, and let $ B' $ be the set $ B' = \{ y \in V \mid \dist_i (x, y) \leq r \} $ for some integer $ r $.
If $ | \comp_i (x) | < r $, then $ \comp_i (x) = B' $.
Furthermore, $ | \comp_i (x) | < r $ if and only if $ | B' | < r $.
\end{observation}

\begin{proof}
Clearly $ B' \subseteq \comp_i (x) $, and thus $ | B' | \leq | \comp_i (x) | $.
Therefore, if $ | \comp_i (x) | < r $, also $ | B' | < r $.
Now assume that $ | B' | < r $.
We first show that $ \comp_i (x) \subseteq B' $.
Let $ y $ be a node in $ \comp_i (x) $ and assume by contradiction that $ \dist_i (x, y) > r $.
Since $ y \in \comp_i (x) $, $ x $ and $ y $ are connected, and therefore the shortest path from $ x $ to $ y $ has to contain some node $ z $ such that $ \dist_i (x, z) = r $.
The shortest path $ \pi $ from $ x $ to $ z $ contains $ \dist_i (x, z) = r $ edges and $ r + 1 $ nodes.
For every node~$ z' $ on~$ \pi $ we have $ \dist_i (x, z') \leq r $, and thus $ \pi \subseteq B' $.
Since $ | \pi | = r + 1 $, we get $ | B' | \geq r+1 $, which contradicts our assumption.
Therefore $ \dist_i (x, y) \leq r $, which means that $ \comp_i (x) \subseteq B' $.
Now $ | \comp_i (x) | \leq | B' | < r $, as desired.
\end{proof}

\begin{observation}\label{lem:radius_formula}
For every center $ j $ and every $ i \leq k $, $ \radius_i^j = \q/2 - | \T_i^j | $.
\end{observation}

\begin{proof}
When the center $ j $ is opened the algorithm sets $ \radius_i^j = \q/2 $ and $ \T_i^j = \emptyset $.
Therefore $ \radius_i^j = \q/2 - | \T_i^j | $ trivially holds.
Afterward the algorithm only modifies $ \radius^j $ and $ \T^j $ when a center is moved.
Since $ \radius^j $ is increased by exactly the amount by which $ | \T^j | $ is decreased, the equation remains true.
\end{proof}

\begin{observation}\label{lem:special_set_small_component}
For every center $ j $ and every $ i \leq k $, $ | \comp_i (x) | < \q $ for every node $ x \in \T_i^j $.
\end{observation}

\begin{proof}
For every node $ x $ that is put into $ \T_i^j $ after the $i$-th edge deletion, we have $ | \comp_i (x) | < \radius_i^j $.
Since the size of the connected component of $ x $ never increases in a decremental graph and $ \radius_i^j \leq \q $ for all $ i \leq k $ by \Cref{lem:radius_formula}, the claim is true.
\end{proof}

\begin{observation}\label{lem:intra_disjointness}
For every center $ j $ and every $ i \leq k $, the sets $ \ball_i^j $ and $ \T_i^j  $ are disjoint.
\end{observation}

\begin{proof}
The set $ \T_i^j $ only contains nodes in connected components from which the center $ j $ has been moved away, i.e., that do not contain $ \cen_i^j $.
No center will ever be moved back into such a connected component.
Since $ \ball_i^j \subseteq \comp_i (\cen_i^j) $, we conclude that $ \ball_i^j $ and $ \T_i^j  $ are disjoint.
\end{proof}

\paragraph*{Disjointness Property.}
We now want to prove the disjointness property.
We will proceed as follows:
First we show that, for every center $ j $ that is opened, the set $ \ball^j \cup \T^j $ is disjoint from the set $ \ball^{j'} \cup \T^{j'} $ of every other existing center $ j' $.
Afterward we show that the algorithm never adds any nodes to $ \ball^j \cup \T^j $.
These two facts will imply that all the sets $ \ball^j \cup \T^j $ are disjoint.

\begin{lemma}[Initial disjointness]\label{lem:inter_disjointness_open}
When the algorithm opens a center $ j $ after the $i$-th edge deletion, the set $ \ball_i^j \cup \T_i^j $ is disjoint from the set $ \ball_i^{j'} \cup \T_i^{j'} $ for every other center $ j \neq j' $.
\end{lemma}

\begin{proof}
Let $ j $ be the center that is opened, and let $ j' \neq j $ be an existing center.
The algorithm sets $ \T_i^j = \emptyset $, and therefore we only have to argue that $ \ball_i^j $ and $ \ball_i^{j'} \cup \T_i^{j'} $ are disjoint.
Note that $ \cen_i^j $ is in a connected component of size at least $ \q $, because otherwise the algorithm would not have opened a center at $ \cen_i^j $.
Observe that the set $ \ball_i^j $ is contained in the connected component of $ \cen_i^j $.
By \Cref{lem:special_set_small_component} all nodes of $ \T_i^{j'} $ are in a connected component of size less than $ \q $, and therefore  $ \ball_i^j \cap \T_i^{j'} = \emptyset $.
We now argue that $ \ball_i^j \cap \ball_i^{j'} = \emptyset $.
Suppose that there is some node $ x $ contained in both $ \ball_i^j $ and $ \ball_i^{j'} $.
By the definition of $ \ball_i^j $ and $ \ball_i^{j'} $ we get $ \dist_i (\cen_i^j, x) \leq \radius_i^j = \q/2 - | \T_i^j | \leq \q/2 $ as well as  $ \dist_i (\cen_i^{j'}, x) \leq \q/2 $.
By the triangle inequality we get
\begin{equation*}
\dist_i(\cen_i^j, \cen_i^{j'}) \leq \dist_i (\cen_i^j, x) + \dist_i (x, \cen_i^{j'}) \leq \q/2 + \q/2 = \q \, .
\end{equation*}
But then $ \cen_i^j $ is covered by $ \cen_i^{j'} $.
This means that the algorithm would not have opened a new center at $ \cen_i^j $, which contradicts our assumption.
\end{proof}

\begin{lemma}[Shrinking property]\label{lem:inter_disjointness_subset_property}
For every center $ j $ and every $ i < k $, we have $ \ball_i^j \cup \T_i^j \subseteq \ball_{i+1}^j \cup \T_{i+1}^j $.
\end{lemma}

\begin{proof}
Let $ (u, v) $ be the $ (i+1) $-th deleted edge.
We only have to argue that the claim holds for centers that the algorithm has already opened before this deletion.
If the algorithm does not move $ j $, then the values of $ \T^j $, $ \radius^j $, and $ \cen^j $ are not changed at all, and thus $ \T_{i+1}^j = \T_i^j $, $ \radius_{i+1}^j = \radius_i^j $, and $ \cen_{i+1}^j = \cen_i^j $.
Furthermore, since distances never decrease in a decremental graph we also have $ \ball_{i+1}^j \subseteq \ball_i^j $ and the claim follows.

Now consider the case that the algorithm moves the center $ j $ from $ x = \cen_i^j $ to $ \cen_{i+1}^j $, where either $ \cen_{i+1}^j = u $ or $ \cen_{i+1}^j = v $.
Assume without loss of generality that $ \cen_{i+1}^j = v $.
To simplify notation, let $ A $ denote the set $ A = \comp_{i+1} (x) $.
The fact that the algorithm moves the center implies that $ | A | < \radius_i^j $.
Note that the algorithm sets $ \T_{i+1}^j = \T_i^j \cup A $ and $ \radius_{i+1}^j = \radius_i^j - |A| $.

The observation needed for proving the shrinking property is $ \ball_{i+1}^j \cup A \subseteq \ball_i^j $.
From this observation we get $ \ball_{i+1}^j \cup \T_{i+1}^j = \ball_{i+1}^j \cup \T_i^j \cup A \subseteq \ball_i^j \cup \T_i^j $, as desired.
We first prove  $ A \subseteq \ball_i^j $ and then $ \ball_{i+1}^j \subseteq \ball_i^j $.
Let $ B' $ be the set $ B' = \{ z \in V \mid \dist_{i+1} (x, z) \leq \radius_i^j \} $.
Since $ | A | < \radius_i^j $ we get $  A \subseteq B' $ by \Cref{lem:size_component_ball}, and since $ \dist_i (x, z) \leq \dist_{i+1} (x, z) $ for every node $ z $ we have $ B' \subseteq \ball_i^j $.
Now $ A \subseteq B' $ and $ B' \subseteq \ball_i^j $, and we may conclude that $ A \subseteq \ball_i^j $.

Finally, we prove that $ \ball_{i+1}^j \subseteq \ball_i^j $.
Since we move the center $ j $ from $ x $ to $ v $ it must be the case, by the way the algorithm works, that $ | A | = | \comp_{i+1} (x) | < | \comp_i (x) | $ and that $ v $ is not connected to $ x $ in $ G_{i+1} $.
This can only happen if $ v $ is connected to $ x $ in $ G_i $.
Let $ z $ be a node in $ \ball_{i+1}^j $, which means that $ \dist_{i+1} (v, z) \leq \radius_{i+1}^j $.
Consider a shortest path $ \pi $ from $ x $ to $ v $ in $ G_i $ consisting of $ \dist_i (x, v) $ many edges.
Every edge on $ \pi $ except for the last one (which is $ (u, v) $) is also contained in $ G_{i+1} $, and therefore all nodes on $ \pi $ except for $ v $ are contained in $ A $.
Therefore we get $ | A | \geq | \pi \setminus \{v\} | = \dist_i (x, v) $.
We now get $ z \in \ball_i^j $ by observing that $ \dist_i (x, z) \leq \radius_i^j $, which can be seen from the following chain of inequalities:
\begin{equation*}
\dist_i (x, z) \leq \dist_i (x, v) + \dist_i (v, z) \leq \dist_i (x, v) + \dist_{i+1} (v, z) \leq | A | + \radius_{i+1}^j = \radius_i^j \, . \qedhere
\end{equation*}
\end{proof}

\begin{lemma}[Disjointness]\label{lem:disjointness_property}
Algorithm~\ref{alg:deterministic_algorithm} maintains the following invariant: For all centers $ j \neq j' $ and every $ i \leq k $, $ \ball_i^j \cup \T_i^j $ is disjoint from $ \ball_i^{j'} \cup \T_i^{j'} $.
\end{lemma}

\begin{proof}
By \Cref{lem:inter_disjointness_open} the invariant holds after the initialization.
Now consider the $ (i+1) $-th edge deletion.
Let $ j \neq j' $ be two different existing centers.
By the induction hypothesis $ \ball_i^j \cup \T_i^j $ and $ \ball_i^{j'} \cup \T_i^{j'} $ are disjoint.
Since $ \ball_{i+1}^j \cup \T_{i+1}^j \subseteq \ball_i^j \cup \T_i^j $ and $ \ball_{i+1}^{j'} \cup \T_{i+1}^{j'} \subseteq \ball_i^{j'} \cup \T_i^{j'} $ by \Cref{lem:inter_disjointness_subset_property}, also $ \ball_{i+1}^j \cup \T_{i+1}^j $ and $ \ball_{i+1}^{j'} \cup \T_{i+1}^{j'} $ are disjoint.
Now let $ j $ be an existing center and let $ j' $ be a center that is opened in the procedure \CenterCoverGreedyOpen (called at the end of the procedure \CenterCoverDelete).
By \Cref{lem:inter_disjointness_open} we also have that $ \ball_{i+1}^j \cup \T_{i+1}^j $  and $ \ball_{i+1}^{j'} \cup \T_{i+1}^{j'} $ are disjoint.
This shows that, for \emph{all} centers $ j $ and $ j' $ such that $ j \neq j' $, $ \ball_{i+1}^j \cup \T_{i+1}^j $ and $ \ball_{i+1}^{j'} \cup \T_{i+1}^{j'} $ are disjoint.
\end{proof}

\paragraph*{Largeness Property.}
We now want to prove the largeness property which states that for every center $ j $ the size of the set $ \ball^j \cup \T^j $ is always at least $ \q / 2 $.
The largeness property will follow from the invariant $ | \ball^j | \geq \radius^j $.
Before we can prove this invariant we have to argue that our algorithm really moves every center $ j $ that fulfills the ``moving condition'' $ | \comp_i (\cen^j) | \geq \radius^j $ and $ | \comp_{i+1} (\cen^j) | < \radius^j $.
Remember that the algorithm only moves \emph{one} such center after each deletion.
We show that there actually is at most one center fulfilling the moving condition, and therefore it is not necessary that the algorithm also moves any other center.

\begin{observation}\label{lem:center_too_small_consequence}
Let $ (u, v) $ be the $ (i+1) $-th deleted edge.
If $ | \comp_i (\cen_i^j) | \geq \radius_i^j $ and $ | \comp_{i+1} (\cen_i^j) | < \radius_i^j $, then $ u \in \ball_i^j $ and $ v \in \ball_i^j $.
\end{observation}

\begin{proof}
Suppose that $ \dist_i (u, \cen_i^j) > \radius_i^j $.
Let $ \pi $ be a shortest path from $ \cen_i^j $ to $ u $ in $ G_i $ consisting of $ \dist_i (u, \cen_i^j) > \radius_i^j $ many edges and thus at least  $ \radius_i^j + 1 $ nodes.
The edge $ (u, v) $ can only appear as the last edge on the shortest path $ \pi $.
Therefore, after deleting it, there are still $ \radius_i^j $ nodes connected to $ \cen_i^j $, which contradicts the assumption that $ \comp_{i+1} (\cen_i^j) < \radius_i^j $.
Thus, $ \dist_i (u, \cen_i^j) \leq \radius_i^j $ which means that $ u \in \ball_i^j $.
Since the edge $ (u, v) $ is undirected the same argument works for $ v $.
\end{proof}

\begin{lemma}[Uniqueness of center to move]\label{lem:at_most_one_center_moved}
Let $ (u, v) $ be the $ (i+1) $-th deleted edge.
If $ | \ball_i^j | \geq \radius_i^j $ for every center $ j $, then there is at most one center $ j $ such that $ | \comp_{i+1} (\cen_i^j) | < \radius_i^j $ and, in $ G_{i+1} $, either $ u $ is connected to $ \cen_i^j $ (and $ v $ is disconnected from $ \cen_i^j $) or $ v $ is connected to $ \cen_i^j $ (and $ u $ is disconnected from $ \cen_i^j $).
\end{lemma}

\begin{proof}
Let $ j $ be a center such that $ | \comp_{i+1} (\cen_i^j) | < \radius_i^j $.
As $ | \ball_i^j | \geq \radius_i^j $, we also have $ | \comp_i (\cen_i^j) | \geq \radius_i^j $ by \Cref{lem:size_component_ball}.
The size of the connected component of $ \cen_i^j $ can only decrease if the deletion of $ (u, v) $ disconnects at least one node from $ \comp_i (\cen^j) $.
For this to happen, $ u $ and $ v $ must be connected to $ \cen_i^j $ in $ G_i $, and furthermore one of these nodes (either $ u $ or $ v $) must be disconnected from $ \cen_i^j $ in $ G_{i+1} $ while the other node stays connected to $ \cen_i^j $.

Now suppose that there are two centers $ j \neq j' $ such that $ | \comp_i (\cen_i^j) | \geq \radius_i^j $ and $ | \comp_{i+1} (\cen_i^j) | < \radius_i^j $, and $ | \comp_i (\cen_i^{j'}) | \geq \radius_i^{j'} $ and $ | \comp_{i+1} (\cen_i^{j'}) | < \radius_i^{j'} $.
By \Cref{lem:center_too_small_consequence}, we get $ u \in \ball_i^j $ and $ u \in \ball_i^{j'} $, which contradicts the disjointness property of \Cref{lem:disjointness_property}.
We conclude that there cannot be two such centers $ j \neq j' $.
\end{proof}

\begin{lemma}\label{lem:size_balls_invariant}
For every center $ j $ and every $ i \leq k $, \Cref{alg:deterministic_algorithm} maintains the invariant $ | \ball_i^j | \geq \radius_i^j $.
\end{lemma}

\begin{proof}
We first argue that the invariant holds for every center $ j $ that we open at some node $ x $ in the greedy open procedure after the $i$-th deletion.
The algorithm only opens the center if $ x $ is in a connected component of size at least $ \q $.
Since $ \radius_i^j = \q/2 - | \T_i^j | \leq \q $ (\Cref{lem:radius_formula}) we have $ | \comp_i (x) | \geq \radius_i^j $.
Therefore we get $ | \ball_i^j | \geq \radius_i^j $ by \Cref{lem:size_component_ball}.

We now show that the invariant is maintained for all centers that have already been opened before we delete the $ (i+1) $-th edge $ (u, v) $.
Consider first the case that $ | \comp_{i+1} (\cen_i^j) | \geq \radius_i^j $.
In that case the center $ j $ will not be moved and we have $ \T_{i+1}^j = \T_i^j $, $ \ball_{i+1}^j = \ball_i^j $, and $ \radius_{i+1}^j = \radius_i^j $.
Since $ | \comp_{i+1} (\cen_{i+1}^j) | \geq \radius_{i+1}^j $ we get $ | \ball_{i+1}^j | \geq \radius_{i+1}^j $, as desired by \Cref{lem:size_component_ball}.

Now consider the case that $ | \comp_{i+1} (\cen_i^j) | < \radius_i^j $.
Since the invariant holds for $ i $, \Cref{lem:at_most_one_center_moved} applies, and thus we can be sure that the algorithm will move center $ j $ from node $ x $ to node $ y $ (where either $ y = u $ or $ y = v $).
Remember that we have $ \cen_i^j = x $, $ \cen_{i+1}^j = y $, and $ \radius_{i+1}^j = \radius_i^j - | \comp_{i+1} (x) | $ in that case.
Since $ x $ and $ y $ were connected in $ G_i $ but are no longer connected in $ G_{i+1} $ we get $ \comp_{i+1} (y) = \comp_i (x) \setminus \comp_{i+1} (x) $.
Due to $ \comp_{i+1} (x) \subseteq \comp_i (x) $ it follows that
\begin{equation*}
| \comp_{i+1} (y) | = | \comp_i (x)  | - | \comp_{i+1} (x) | \geq \radius_i^j - | \comp_{i+1} (x) | = \radius_{i+1}^j \, .
\end{equation*}
By \Cref{lem:size_component_ball}, the fact that $ | \comp_{i+1} (y) | \geq \radius_{i+1}^j $ implies that $ | \ball_{i+1}^j | \geq \radius_{i+1}^j $ as desired.
\end{proof}

\begin{lemma}[Largeness]\label{lem:largeness_property}
For every center $ j $ and every $ i \leq k $, \Cref{alg:deterministic_algorithm} maintains the invariant $ | \ball_i^j \cup \T_i^j | \geq \q/2 $.
\end{lemma}

\begin{proof}
By \Cref{lem:intra_disjointness}, $ \ball_i^j $ and $ \T_i^j $ are disjoint, and by \Cref{lem:radius_formula} we have $ \radius_i^j = \q/2 - | \T_i^j | $.
By \Cref{lem:size_balls_invariant} we have $ | \ball_i^j | \geq \radius_i^j $.
Therefore we get the desired bound as follows:
\begin{equation*}
| \ball_i^j \cup \T_i^j | = | \ball_i^j | + | \T_i^j | \geq \radius_i^j + | \T_i^j | = \q/2 - | \T_i^j | + | \T_i^j | = \q/2
\end{equation*}
where the inequality above follows from \Cref{lem:size_balls_invariant}.
\end{proof}

\paragraph*{Bounding the Number of Open-Operations.}
Now that we have established the disjointness and the largeness property for the sets $ \ball^j \cup \T^j $ of every center $ j $, we can bound the number of open-operations by $ \opens = O (n / \q) $.
This will be useful for our goal of limiting the total update time of the moving centers data structure to $ O (m n \Q / \q) $.

\begin{lemma}[Number of open-operations]\label{lem:number_of_open_operations}
Over all edge deletions, \Cref{alg:deterministic_algorithm} performs $ O (n / \q) $ open-operations in its internal moving centers data structure.
\end{lemma}

\begin{proof}
Let $ \C_k $ denote the set of centers after all $ k $ deletions.
Note that moving a center does not change the number of centers.
Therefore, the size of $ \C_k $ is equal to the total number of centers opened.
Due to the disjointness property (\Cref{lem:disjointness_property}) the sets $ \ball_k^j \cup \T_k^j $ after all $ k $ edge deletions are disjoint for all centers $ j $.
When we sum up over all these sets we do not count any node twice.
Therefore we get
\begin{equation*}
\sum_{j \in \C_k} | \ball_k^j \cup \T_k^j | = \left| \bigcup_{j \in \C_k} (\ball_k^j \cup  \T_k^j) \right| \leq n
\end{equation*}
By the largeness property (\Cref{lem:largeness_property}) every set $\ball_k^j \cup \T_k^j $ has size at least $ \q/2 $, i.e., $ | \ball_k^j \cup \T_k^j | \geq \q/2 $.
We now combine both inequalities and get
\begin{align*}
n &\geq \sum_{j \in \C} | \ball_k^j \cup \T_k^j | \geq \sum_{j \in \C} \q / 2 = |C| \q / 2
\end{align*}
which gives $ |C| \leq 2n / \q $, as desired.
\end{proof}

\paragraph*{Bounding the Total Moving Distance.}
Finally, we prove that the total moving distance of the moving centers data structure used by our algorithm is $ O (n) $.
For this proof we will use a property of the algorithm that we call \emph{confinement}: Every center $j$ will be moved only through nodes that are added to $\T^j$.

\begin{lemma}[Confinement]\label{lem:confinement}
For every move of center $ j $ from $ \cen_i^j $ to $ \cen_{i+1}^j $ after the $ (i+1) $-th edge deletion, let $ \pi_i^j $ be the set of nodes on a shortest path from $ \cen_i^j $ to $ \cen_{i+1}^j $ in $ G_i $.
Then, for every center $ j $ and every $ 0 \leq i < k $, $ \pi_i^j \setminus \{ \cen_{i+1}^j \} \subseteq \T_{i+1}^j \setminus \T_i^j $.
\end{lemma}

\begin{proof}
Let $ (u, v) $ be the $ (i+1) $-th deleted edge.
Consider the situation that the algorithm moves some center $ j $ from $ \cen_i^j $ to $ \cen_{i+1}^j $.
By the rules of the algorithm for moving centers we have either $ \cen_{i+1}^j = u $ or $ \cen_{i+1}^j = v $.
Due to \Cref{lem:center_too_small_consequence} we have $ \cen_{i+1}^j \in \ball_i^j $, which means that $ \dist_i (\cen_i^j, \cen_{i+1}^j) \leq \radius_i^j $.

Now let $ \pi_i^j $ be a shortest path from $ \cen_i^j $ to $ \cen_{i+1}^j $ in $ G_i $.
All nodes in $ \pi_i^j $, except for $ \cen_{i+1}^j $, are connected to $ \cen_i^j $ in $ G_{i+1} $ since the edge $ (u, v) $ only appears as the last edge on the shortest path due to $ \cen_{i+1}^j = u $ or $ \cen_{i+1}^j = v $.
Therefore we have $ \pi_i^j \setminus \{ \cen_{i+1}^j \} \subseteq \comp_{i+1} (\cen_i^j) $.
Since $ \T_{i+1}^j = \T_i^j \cup \comp_{i+1} (\cen_i^j) $, we get $ \pi_i^j \setminus \{ \cen_{i+1}^j \} \subseteq \T_{i+1}^j $.
We also have $ \pi_i^j \setminus \{ \cen_{i+1}^j \} \subseteq \ball_i^j $ because $ \dist_i (\cen_i^j, \cen_{i+1}^j) \leq \radius_i^j $.
Since $ \ball_i^j $ and $ \T_i^j $ are disjoint (\Cref{lem:intra_disjointness}), also $ \pi_i^j \setminus \{ \cen_{i+1}^j \} $ and $ \T_i^j $ are disjoint.
It therefore follows that $ \pi_i^j \setminus \{ \cen_{i+1}^j \} \subseteq \T_{i+1}^j \setminus \T_i^j $.
\end{proof}

\begin{lemma}[Total moving distance]\label{lem:number_of_move_operations}
The total moving distance of the moving centers data structure used by \Cref{alg:deterministic_algorithm} is $ \moves = O (n) $.
\end{lemma}

\begin{proof}
We let $ \C_k $ denote the set of centers after the algorithm has processed all deletions.
Furthermore, we denote by $ o_j $ the index of the edge deletion before which the center $ j $ has been opened; i.e., center $ j $ was opened before the $ o_j $-th and after the $ (o_j - 1) $-th deletion.

Consider the situation that the algorithm moves a center $ j $ from $ \cen_i^j $ to $ \cen_{i+1}^j $ after the $ (i+1) $-th edge deletion and let $ \pi_i^j $ be a shortest path from $ \cen_i^j $ to $ \cen_{i+1}^j $ in $ G_i $ as in \Cref{lem:confinement}.
The shortest path $ \pi_i^j $ consists of $ \dist_i (\cen_i^j, \cen_{i+1}^j) $ many edges and $ \dist_i (\cen_i^j, \cen_{i+1}^j) + 1 $ many nodes.
Therefore we get $ \dist_i (\cen_i^j, \cen_{i+1}^j) = | \pi_i^j \setminus \{ \cen_{i+1}^j \} | $.
By \Cref{lem:confinement} we have $ \pi_i^j \setminus \{ \cen_{i+1}^j \} \subseteq \T_{i+1}^j \setminus \T_i^j $ for every center $ j $ and every $ 0 \leq i < k $.
The value of the set $ \T^j $ after all edge deletions is given by $ \T_k^j $ for every center $ j $.
By the disjointness property (\Cref{lem:disjointness_property}) we have $ \sum_{j \in \C} | \T_k^j | = \left| \bigcup_{j \in \C} \T_k^j \right| $.
We now obtain the bound $ \moves \leq n $ as follows:
\begin{align*}
\moves &= \sum_{j \in \C_k} \sum_{o_j \leq i < k} \dist_i (\cen_i^j, \cen_{i+1}^j) = \sum_{j \in \C_k} \sum_{o_j \leq i < k} | \pi_i^j \setminus \{ \cen_{i+1}^j \} | \\
 &\leq \sum_{j \in \C_k} \sum_{o_j \leq i < k} | \T_{i+1}^j \setminus \T_i^j| = \sum_{j \in \C_k} | \T_k^j | = \left| \bigcup_{j \in \C} \T_k^j \right| \leq n \, . \qedhere
\end{align*}
\end{proof}

\paragraph*{Implementation Details.}
Before we finish this section we clarify two implementation details of \Cref{alg:deterministic_algorithm} and argue that they can be carried out within the total update time of $ O (m n \Q / \q) $.

There are two places in the algorithm where we have to compute the sizes of connected components.
First, in the procedure \GreedyOpen, we have to check for every node that is not covered by any center whether it is in a connected component of size at least $ \q $.
Second, in the procedure \Delete, we have to check whether the size of the connected component of some center $ j $ drops below $ \radius^j $.
So far we have not explained explicitly how to carry out these steps.
If we could obtain the size of the connected component deterministically in linear time, the running time analysis we have given so far would suffice.
Remember that the moving centers data structure internally maintains an ES-tree for every center.
Thus, it would seem intuitive to use the ES-trees for counting the number of nodes in the current component of each center.
However, we do not report edge deletions to the moving centers data structure immediately.
Therefore it is not clear how to use these ES-trees to determine the size of the connected components of a centers.

Instead, we do the following.
In parallel to our own algorithm we use the deterministic (fully) dynamic connectivity data structure of Henzinger and King~\cite{HenzingerK01}.\footnote{This is the fastest known deterministic dynamic connectivity data structure with \emph{constant} query time.}
This data structure can answer queries of the form ``are the nodes $ x $ and $ y $ connected?'' in constant time.
Its amortized time per deletion is $ O (n^{1/3} \log{n}) $.
Thus, its total update time over all deletions is $ O(m n^{1/3} \log{n}) $.
Trivially, this data structure allows us to compute the size of the connected component of a node $ x $ in time $ O(n) $: We simply iterate over all nodes and count how many of them are connected to $ x $.
We now explain how to perform the two tasks listed above using the dynamic connectivity data structure.

\begin{lemma}[Detail of Line~\ref{line:component size} in \Cref{alg:deterministic_algorithm}]\label{rem:detail of initialization}\label{rem:computing_set_of_uncovered_nodes}
Given a dynamic connectivity data structure with constant query time, performing the check in the if-condition of Line~\ref{line:component size} takes time $ O ((n + \opens) n) $ over all deletions, where $ \opens $ is the total number of open-operations.
\end{lemma}

\begin{proof}
Given a node $ x $ we have to check whether \FindCenter{x} $ = \bot $ and $ |\comp_{G_{i}} (x)| \geq \q $.
We first check whether $ x $ is covered by any center by querying the moving centers data structure (if $ x $ is covered, the procedure returns a center covering $ x $; otherwise it returns $ \bot $.)
This check takes constant time.
If a node $ x $ is not covered, we additionally have to check whether $ | \comp_i (x) | < \q $.
Note that if $ | \comp_i (x) | < \q $ for some node $ x $, we do not have to consider this node anymore after future deletions because connected components never increase their size in a decremental graph.
Therefore we may spend time $ O (n) $ for every node $ x $ to determine $ | \comp_i (x) | $.
If $ | \comp_i (x) | < \q $, then we charge this time to the node and will never process the node again in the future, and if $ | \comp_i (x) | \geq \q $, then we charge this time to the open-operation.
Therefore the total running time over all deletions for performing this check in the if-condition is $ O ((n + \opens) n) $, where $ \opens $ is the total number of open-operations.
\end{proof}

\begin{lemma}[Detail of Line~\ref{line:find bad component} of \Cref{alg:deterministic_algorithm}]\label{rem:detail of finding bad component}
Given a dynamic connectivity data structure with constant query time, we can, after the $ (i+1) $-th deletion, find a center $ j $ such that $ | \comp_{i+1} (\cen_i^j) | < \radius_i^j $ if it exists in time $ O(n) $.
\end{lemma}

\begin{proof}
Let $ (u, v) $ be the $ (i+1) $-th deleted edge.
For every center $ j $ such that $ | \comp_{i+1} (\cen_i^j) | < \radius_i^j $, we have $ u \in \ball_i^j $ by \Cref{lem:center_too_small_consequence}.
Moreover, by the disjointness property (\Cref{lem:disjointness_property}), there can only be at most one center $ j $ such that $ u \in \ball_i^j $.
The algorithm for finding this center now is simple: We find a center $j$ such that $u\in \ball_i^j$, which is unique if it exists; then we compute the size of the connected component containing $\cen_i^j$ using the dynamic connectivity data structure~\cite{HenzingerK01}.
In particular, we iterate over all centers in time $ O (n) $ to find a candidate center $ j $ such that $ \dist_i (u, \cen_i^j) \leq \radius^j $ (i.e., $ u \in \ball_i^j $) (as argued above, at most one such center exists).
We can determine the distance $ \dist_i (u, \cen_i^j) $ in constant time by querying the moving centers data structure.
For the candidate center $ j $ we now have to check whether $ | \comp_{i+1} (\cen_i^j) | < \radius_i^j $.
We determine the size of $ \comp_{i+1} (\cen_i^j) $ in time $ O(n) $ by using the dynamic connectivity data structure with constant query time.
Thus, the running time for this algorithm is $ O (n) $ per deletion.
\end{proof}

\paragraph*{Total Update Time.}
Now we state the total update time of \Cref{alg:deterministic_algorithm}.
The bounds on the number of centers opened and the total moving distance of the centers allow us to bound the running time of the moving centers data structure used by the algorithm.

\begin{theorem}
The deterministic center cover data structure of \Cref{alg:deterministic_algorithm} has constant query time and a total update time of $ O (m n \Q / \q) $.
\end{theorem}

\begin{proof}
By \Cref{lem:existence_dynamic_centers_ds} the moving centers data structure internally used by Algorithm~\ref{alg:deterministic_algorithm} has constant query time and a total deterministic update time of $ O (\opens m \Q + \moves m) $, where $ \opens $ is the total number of open-operations and $ \moves $ is the total moving distance.
Algorithm~\ref{alg:deterministic_algorithm} delegates all queries to the moving centers data structure and therefore also has constant query time.
By \Cref{lem:number_of_open_operations} the number of open-operations is $ O (n / \q) $, and by \Cref{lem:number_of_move_operations} the total moving distance is $ O (n) $.
Therefore the total update time of the moving centers data structure is
\begin{equation*}
O (\opens m \Q + \moves m) = O( m n \Q / \q + m n) = O( m n \Q / \q)
\end{equation*}
because $ \q \leq \Q $.
As argued in \Cref{rem:detail of initialization, rem:detail of finding bad component}, all other operations of the algorithm can be implemented within a total update time of $ O( m n \Q / \q) $.
Therefore the claimed running time follows.
\end{proof}

\subsection{Deterministic Fully Dynamic Algorithm}\label{sec:fully_dynamic}

There is a well-known reduction by Henzinger and King~\cite{HenzingerK95} for converting a decremental algorithm into a fully dynamic algorithm.
A similar approach has been used by Roditty and Zwick~\cite{RodittyZ12}, using the decremental algorithm we derandomized above as the starting point.
In the following we sketch the deterministic fully dynamic algorithm implied by \Cref{thm:deterministic}.
The fully dynamic algorithm allows two update operations: deleting an arbitrary set of edges and inserting a set of edges touching a node $ v $, called the center of the insertion.

\begin{theorem}
For every $ 0 < \epsilon \leq 1 $ and every $ t \leq \sqrt{n} $ there is a deterministic fully dynamic $ (1 + \epsilon, 0) $-approximate APSP data structure with amortized update time $ O (m n / (\epsilon t)) $ and query time $ \tilde O (t) $.
\end{theorem}

\begin{proof}
The algorithm works in phases.
After each $ t $ update operations we start a new phase.
At the beginning of each phase we re-initialize the decremental algorithm of \Cref{thm:deterministic}.
We report to this algorithm all future deletions of the phase, but no insertions.
For all nodes $ u $ and $ v $ let $ \delta_1 (u, v) $ denote the $ (1 + \epsilon) $-approximate distance estimate obtained by the decremental algorithm.
Additionally, after every update in the graph, we do the following:
Let $ I $ denote the set of centers of all insertions that so far happened in the current phase.
For every $ v \in I $, we compute the shortest paths from $ v $ to all nodes in the current graph, i.e., where all insertions and deletions are considered.
We use Dijkstra's algorithm for this task and denote by $ \delta_2 (u, v) $ the distance from $ u $ to $ v $ computed in this way.

To answer a query for the approximate distance between nodes $ u $ and $ v $ we compute and return the following value: $ \delta (u, v) = \min (\delta_1 (u, v), \min_{x \in I} (\delta_2 (u, x) + \delta_2 (x, v))) $.
Let $ \dist (u, v) $ denote the distance from $ u $ to $ v $ in the current graph.
If there is a shortest path from $ u $ to $ v $ that does not use any edge inserted in the current phase, then the decremental algorithm provides a $ (1 + \epsilon) $-approximation of the distance between $ u $ and $ v $, i.e., $ \delta_1 (u, v) \leq (1 + \epsilon) \dist (u, v) $.
Otherwise the shortest path from $ u $ to $ v $ contains an inserted node $ x \in I $.
In that case we have $ \dist (u, v) = \delta_2 (u, x) + \delta_2 (x, v) $ and thus $ \dist (u, v) = \min_{x \in I} (\delta_2 (u, x) + \delta_2 (x, v)) $.
This means that $ \delta (u, v) \leq (1 + \epsilon) \dist (u, v) $.
As both $ \delta_1 (u, v) $ and $ \min_{x \in I} (\delta_2 (u, x) + \delta_2 (x, v)) $ never underestimate the true distance, we also have $ \delta (u, v) \geq \dist (u, v) $.

As the query time of the decremental algorithm is $ O (\log{\log{n}}) $, the query time of the fully dynamic algorithm is $ O (t + \log{\log{n}}) = \tilde O (t) $.
The decremental approximate APSP data structure has a total update time of $ \tilde O (m n / \epsilon) $.
Amortized over the whole phase, we have to pay $ \tilde O (m n / (\epsilon t)) $ per update for this data structure.
Computing the shortest paths from the inserted nodes takes time $ \tilde O(t m) $ per update.
This gives an amortized update time of $ \tilde O (m n / (\epsilon t) + t m) $.
If $ t \leq \sqrt{n} $, the term $ t m $ is dominated by the term $ m n / t $, and thus the amortized update time is $ \tilde O (m n / (\epsilon t)) $.
\end{proof}

We remark that the fully dynamic result of Roditty and Zwick~\cite{RodittyZ12} is still a bit stronger.
Their trade-off is basically the same, but it holds for a larger range of $ t $, namely, $ t \leq m^{1/2-\delta} $ for every fixed $ \delta > 0 $.
The reason is that they use randomization not only for the decremental algorithm but also for some other part of the fully dynamic algorithm.

\section{Conclusion}

We obtained two new algorithms for solving the decremental approximate APSP algorithm in unweighted undirected graphs.
Our first algorithm provides a $ (1 + \epsilon, 2) $-approximation and has a total update time of $ \tilde O (n^{5/2}/\epsilon) $ and constant query time.
The main idea behind this algorithm is to run an algorithm of Roditty and Zwick~\cite{RodittyZ12} on a sparse dynamic emulator.
In particular, we modify the central shortest paths tree data structure of Even and Shiloach~\cite{EvenS81, King99} to deal with edge insertions in a monotone manner.
Our approach is conceptually different from the approach of Bernstein and Roditty~\cite{BernsteinR11}, who also maintain an ES-tree in a sparse dynamic emulator.
The sparsification techniques used here and at other places only work for undirected graphs.
Using a new sampling technique, we recently obtained a $ (1+\epsilon, 0) $-approximation for decremental SSSP in \emph{directed} graphs with constant query time and a total update time of $ o(mn) $~\cite{HenzingerKNSTOC14}.

Our second algorithm provides a $ (1 + \epsilon, 0) $-approximation and has a \emph{deterministic} total update time of $ O (m n \log{n} / \epsilon) $ and constant query time.
We obtain it by derandomizing the algorithm of~\cite{RodittyZ12} using a new amortization argument based on the idea of relocating ES-trees.

Our results directly motivate the following directions for further research.
It would be interesting to extend our derandomization technique to other randomized algorithms.
In particular, we ask whether it is possible to derandomize the \emph{exact} decremental APSP algorithm of Baswana, Hariharan, and Sen~\cite{BaswanaHS07} (total update time $ \tilde O (n^3) $).

Another interesting direction is to check whether our monotone ES-tree approach also works for other dynamic emulators, in particular for weighted graphs.
One of the tools that we used was a dynamic $ (1 + \epsilon, 2) $-emulator for unweighted undirected graphs.
Is it also possible to obtain \emph{purely additive} dynamic emulators or spanners with small additive error?

Maybe the most important open problem in this field is a faster APSP algorithm for the fully dynamic setting.
The fully dynamic algorithm of Demetrescu and Italiano~\cite{DemetrescuI04} provides exact distances and takes time $ \tilde O (n^2) $ per update, which is essentially optimal.
Is it possible to get a faster fully dynamic algorithm that still provides a good approximation---for example a $ (1 + \epsilon) $-approximation?

\printbibliography[heading=bibintoc] 

\appendix
\section{Proof of \Cref{fact:truly subcubic lower bound}}\label{sec:proof of fact:truly subcubic lower bound}

Due to a reduction by Dor, Halperin, and Zwick \cite{DorHZ00}, a combinatorial\footnote{The vague term ``combinatorial algorithm'' is usually used to refer to algorithms that do not use algebraic operations such as matrix multiplication.} algorithm for APSP, even a $(2-\epsilon, 0)$-approximation or $(1+\epsilon, 1)$-approximation one,\footnote{In general, the reduction of Dor et al.~holds for any $(\alpha, \beta)$ approximation as long as $2\alpha+\beta<4$.} with running time $O(n^{3-\delta})$, for any $\delta>0$, will imply a combinatorial algorithm for {\em Boolean matrix multiplication} with the same running time, another breakthrough result. Further, due to Vassilevska Williams and Williams \cite[Theorem 1.3]{WilliamsW10}, the $O(n^{3-\delta})$-time combinatorial algorithm will imply breakthrough results for a few other problems. Since combinatorial dynamic algorithms can be used to solve static APSP, the same argument applies. In particular, the additive error of $2$ in our $(1+\epsilon, 2)$-approximation algorithm is unavoidable if we wish to get an $O(n^{1-\delta})$ running time (a so-called {\em truly subcubic} time) and keep a small multiplicative error of $ 1+\epsilon $. For the same reason, a multiplicative error of $2$ in our $(2+\epsilon, 0)$-approximation algorithm is also unavoidable. Similarly, the running time of our deterministic algorithm cannot be improved further unless we allow larger additive or multiplicative errors.

\end{document}